\documentclass[10pt,twocolumn]{article}
\usepackage[a4paper,left=1.7cm,right=1.7cm,top=1.8cm,bottom=1.9cm]{geometry}
\usepackage[T1]{fontenc}
\usepackage{lmodern}
\usepackage{amsmath,amssymb,amsthm,mathtools}
\usepackage{bm}
\usepackage{graphicx}
\usepackage{booktabs}
\usepackage{longtable}
\usepackage{multirow}
\usepackage{array}
\usepackage{needspace}
\usepackage{float}
\usepackage{placeins}
\usepackage{xcolor}
\usepackage{tikz}
\usetikzlibrary{positioning,arrows.meta,fit}
\usepackage[numbers,sort&compress]{natbib}

\usepackage{authblk}

\usepackage[colorlinks=true,linkcolor=blue!50!black,citecolor=blue!50!black,urlcolor=blue!50!black]{hyperref}
\usepackage[capitalise]{cleveref}

\newtheorem{theorem}{Theorem}
\newtheorem{corollary}[theorem]{Corollary}
\newtheorem{proposition}[theorem]{Proposition}
\theoremstyle{definition}
\newtheorem{remark}[theorem]{Remark}
\theoremstyle{plain}
\newcommand{\E}{\mathbb{E}}
\newcommand{\R}{\mathbb{R}}
\newcommand{\C}{\mathbb{C}}
\newcommand{\tr}{\operatorname{tr}}
\newcommand{\supp}{\operatorname{supp}}
\newcommand{\Vbar}{\bar V}
\newcommand{\gbar}{\bar g}

\title{An Exact Polynomial Task--Risk Bridge\\ for Time-Multiplexed Photonic Quantum Reservoirs}
\author[1]{Yuqi Zhang}
\author[1]{Tianyu Zhou}
\author[1]{Yilun Jiang}
\author[2]{Tian Chen}
\author[1,*]{Hao Tang}
\affil[1]{School of Physics and Astronomy, Shanghai Jiao Tong University, Shanghai, China}
\affil[2]{School of Integrated Circuits, Shanghai Jiao Tong University, Shanghai, China}
\affil[*]{Corresponding author: Hao Tang (\href{mailto:htang2015@sjtu.edu.cn}{\texttt{htang2015@sjtu.edu.cn}})}
\date{}

\begin{document}
\twocolumn[{%
\begin{@twocolumnfalse}
\maketitle

\begin{abstract}
Time-multiplexed (TDM) photonic reservoirs are usually judged by their scores on a few datasets, which say little about a chip's overall performance, and designed by enumerating candidate circuits and simulating each one, which is inefficient and gives no guarantee of finding a good structure.
For a time-unrolled passive linear network with data-modulated gates or sources and a ridge readout, we derive an exact bridge from the task to the prediction risk in three steps.
(i)~\emph{Task features}: the finite-shot risk depends on the time series only through finitely many statistics of its training windows and their correlations with the target, selected by the encoding and the optical paths.
(ii)~\emph{Encoding}: every output moment is a polynomial in the features of the data-modulated sites, with degree and support bounded by tuples of optical paths. Gate encoding reads the characteristic function and creates interactions across time lags; squeezing encoding reads the moment-generating function and displacement encoding low-order moments, and with quadrature receivers they yield only additive and linear models, respectively, for every topology.
(iii)~\emph{Readout}: for homodyne, heterodyne and photon-number receivers the risk closes exactly, and its dependence on the number of shots is an explicit sum over signal-to-noise modes; threshold clicks admit finite-dictionary approximations with certified error. Parts of the network that share no light carry independent states, which a receiver combines only through measurement events that join them.
Checks on seven datasets confirm every exact statement to machine precision. The three steps thus yield a useful space of candidate TDM architectures, pruned by exact statements rather than by trial simulation, which can in future guide the selection of TDM chips for target requirements.
\end{abstract}
\vspace{0.5cm}
\end{@twocolumnfalse}
}]

\section{Introduction}\label{sec:intro}

Reservoir computing trains only a linear readout on the response of a fixed dynamical system~\cite{jaeger2004harnessing,lukosevicius2009reservoir}.
Delay-based reservoirs obtain many virtual nodes from a single nonlinear node and a delay line by time multiplexing~\cite{appeltant2011delay,paquot2012optoelectronic,brunner2013parallel},
and photonic reservoirs have been demonstrated on chip~\cite{vandoorne2014experimental}.
Quantum reservoirs replace the dynamical system by a quantum one~\cite{fujii2017harnessing,nakajima2019boosting,mujal2021opportunities};
Gaussian continuous-variable reservoirs are universal and versatile~\cite{nokkala2021gaussian}, and time-division multiplexing (TDM)
turns a single optical loop network into a large reservoir that processes a stream in real time~\cite{garciabeni2023scalable}.
Loop-based time-bin processing was proposed for scalable boson sampling~\cite{motes2014scalable}.
TDM is also the architecture behind large optical cluster states~\cite{menicucci2011temporal,yokoyama2013ultra,yoshikawa2016million,asavanant2019generation,larsen2019deterministic} and
programmable Gaussian boson sampling~\cite{madsen2022quantum,yu2023universal}, and space--time-multiplexed Gaussian boson sampling has recently been applied to learning tasks~\cite{fu2026chip}.
Loop processors can implement programmable Gaussian gates and temporal linear-optical transformations~\cite{enomoto2021programmable,yonezu2023time}.

A TDM reservoir is specified by its loop delays, number of rails, gate order within a time step,
input encoding (coupling angles, phases, source squeezing or displacement), retained output time bins,
measurement receiver and shot budget. Circuit-level simulation tools can evaluate any specified choice, and such evaluations are indispensable. Used as the only tool, however, they leave two problems. First, a circuit that scores well on a few datasets need not perform well in general: the scores do not show which time-lag interactions the circuit can represent, which of them a task uses, or which receiver can read them at a given shot budget. Second, designing a circuit by enumerating the choices above and simulating each candidate is costly, because the number of candidates grows combinatorially with the number of choices, and a finite enumeration gives no guarantee that a good structure is found. Stability and information-processing capacity provide useful task-independent criteria~\cite{dambre2012information,nokkala2021gaussian},
and approximate task-specific performance formulas have guided the parameters and input masks of classical delay reservoirs~\cite{grigoryeva2015optimal}.
For photonic Gaussian TDM reservoirs, loop reflectivity has also been related analytically to memory depth and finite-sample requirements~\cite{garciabeni2023scalable}, while the effect of squeezing under readout noise has been studied separately~\cite{garciabeni2024squeezing}.
What remains missing is a joint connection, for a given task, between its history statistics,
the path structure of a multi-rail Gaussian TDM circuit, receiver-dependent shot noise, and prediction risk.

We address this gap for time-multiplexed linear-optical reservoirs with a ridge readout, through the composite map
\begin{equation}\label{eq:chain}
G \xrightarrow{\text{enc.}} (\phi_g,\psi_s) \xrightarrow{\ \theta\ } \mathcal V_\theta \xrightarrow{\ b\ } z_b \xrightarrow{\text{ridge}} \hat y ,
\end{equation}
where $G$ is a window of drive values, $(\phi_g,\psi_s)$ the encoding features of the data-dependent gates and sources, $\mathcal V_\theta$ the output moments of the retained modes that the receiver uses (the covariance $V_\theta$ for zero-mean Gaussian states), and $z_b$ the receiver's features.
\Cref{fig:chain} follows this chain through the $U8$ circuit studied below: training windows supply the drive and the targets, the data enter at gate or source sites, the output moments are polynomials in the encoding features, the receiver maps them to features, and the risk depends on the task only through a few statistics of the windows.
Each arrow is treated exactly. The resulting bridge is not a faster simulator; several of its statements hold for every operating point and every shot budget, so they can exclude whole families of designs without simulating them. Our contributions follow the three steps of \eqref{eq:chain}: what a circuit extracts from the task, how the encoding and the optical paths imprint it on the light, and how the receiver reads it out.
\begin{enumerate}
\item[(i)] \emph{Task features.} For receivers whose features and shot noise are moment polynomials, the risk of the population ridge readout depends on the task only through finitely many statistics of the training windows and their correlations with the target (\cref{thm:bridge}): values of the characteristic function for gate encoding, of the moment generating function for squeezing encoding, and low-order moments for displacement encoding, at arguments fixed by the encoding; only drive components, and tuples of them, admitted by the optical paths enter (\cref{cor:task}). These statistics are computed from the data alone and state which features of a time series a given circuit can use.
\item[(ii)] \emph{Encoding.} Every output moment is a polynomial in per-site encoding features, with degree and support bounded by tuples of optical paths (\cref{thm:poly}); for Gaussian sources the covariance is affine in each source, whereas higher moments multiply sources (\cref{prop:gauss}). The encoding therefore fixes the function class available to a linear readout, and with it a classical baseline that no circuit of the class can beat: with quadrature receivers, gate encoding gives trigonometric interactions between lags joined by optical paths, squeezing encoding an additive model and displacement encoding a linear one, for every topology (\cref{cor:class}). The paths exclude whole families of interactions for every operating angle (\cref{thm:zeros}), which gives an optimistic, topology-only risk floor (\cref{cor:topology-floor}); when the path graph splits into components, the output state factorizes exactly (\cref{thm:components}).
\item[(iii)] \emph{Readout.} A receiver reads interactions between components only through events that span them (\cref{cor:span}). Exclusion certificates therefore extend from quadrature receivers to nonlinear component-local receivers such as single-mode clicks, and whether an excluded direction becomes readable is decided by the receiver's events. The finite-shot risk is an explicit sum over single-shot signal-to-noise modes (\cref{prop:shots}), which compares receivers on a given task without sampling. Threshold clicks, which are not moment polynomials, are approximated uniformly by finite dictionaries with a certified risk error (\cref{prop:click}).
\item[(iv)] \emph{Numerical verification.} On seven forecasting datasets and synthetic circuits, every exact statement is confirmed to machine precision, the predicted support of the click receiver's interactions matches the computed one pair by pair, and a sparse second-order truncation reproduces the exact risk to within $0.37\%$ (\cref{sec:truncation,sec:numerics}).
\end{enumerate}
All results concern the risk on the empirical distribution of training windows, not held-out generalization. The theorems use only the linearity of the optics and the form of the encoding. Nonclassicality enters through the source moments and does not change the polynomial or factorized structure, so the identities hold equally for classical light; ``quantum'' refers to the physical platform, whose TMSV and Fock-state sources we treat on the same footing. How the theory could guide circuit prediction is discussed in \cref{sec:discussion}.

\paragraph{Related work.}
Task-adaptive reservoir design is not new. In particular, Grigoryeva \emph{et al.}~\cite{grigoryeva2015optimal} derived an approximate task-dependent capacity model for classical delay reservoirs and used it to optimize architecture parameters and input masks; a magnetic implementation gives another task-adaptive example~\cite{lee2024taskadaptive}.
Classical capacity and Volterra views also describe available temporal computations~\cite{dambre2012information,boyd1985fading,gauthier2021next}; echo-state networks have an independent universality theory~\cite{grigoryeva2018universal}.
Quantum reservoir theory has established Gaussian universality and recurrent quantum-network approximation results~\cite{nokkala2021gaussian,gonon2026feedback,schuette2025expressive}.
Real-time photonic reservoirs and recent continuous-variable, integrated-photonic, and Gaussian-boson-sampler experiments demonstrate memory, forecasting, or large-scale correlations~\cite{garciabeni2023scalable,paparelle2026memory,dibartolo2026forecasting,cimini2026large}; feedback-based time-series quantum reservoirs provide another physical route~\cite{kobayashi2024feedback}.
Finite Fourier representations of data-encoded quantum circuits are known~\cite{schuld2021effect,perezsalinas2020data,yu2022power}; circuit-spectrum-based classical approximations have also been proved~\cite{landman2023classically}.
For linear photonic circuits with Fock-state inputs, the spectrum of a phase-encoded model is set by the number of input photons~\cite{gan2022fock}; \cref{cor:fock} recovers this as a special case.
Most directly, photonic photon-number-resolving reservoir work identifies an encoding-controlled Fourier spectrum and quantifies expressivity by the rank of the observable Gram matrix and a shot-limited conditioned rank, evaluated task-independently for a memoryless (extreme-learning) architecture~\cite{nerenberg2025photon}.
Thus neither finite spectra, task adaptation nor finite-shot effects
alone are claimed as new. The distinction here is that one theorem covers gate, source and displacement encodings in time-unrolled circuits with memory, that the risk is closed exactly in terms of a few statistics of the data, and that this closure yields consequences that hold for every operating point: factorization over path components, a function class fixed by the encoding, and an exact shot law. Gaussian click formulas, including the no-click determinant approximated in \cref{sec:clicks}, and
their relation to Gaussian boson sampling remain standard~\cite{hamilton2017gaussian,kruse2019detailed,quesada2018gaussian,bulmer2022threshold}.

Proofs are in \cref{app:proofs}.

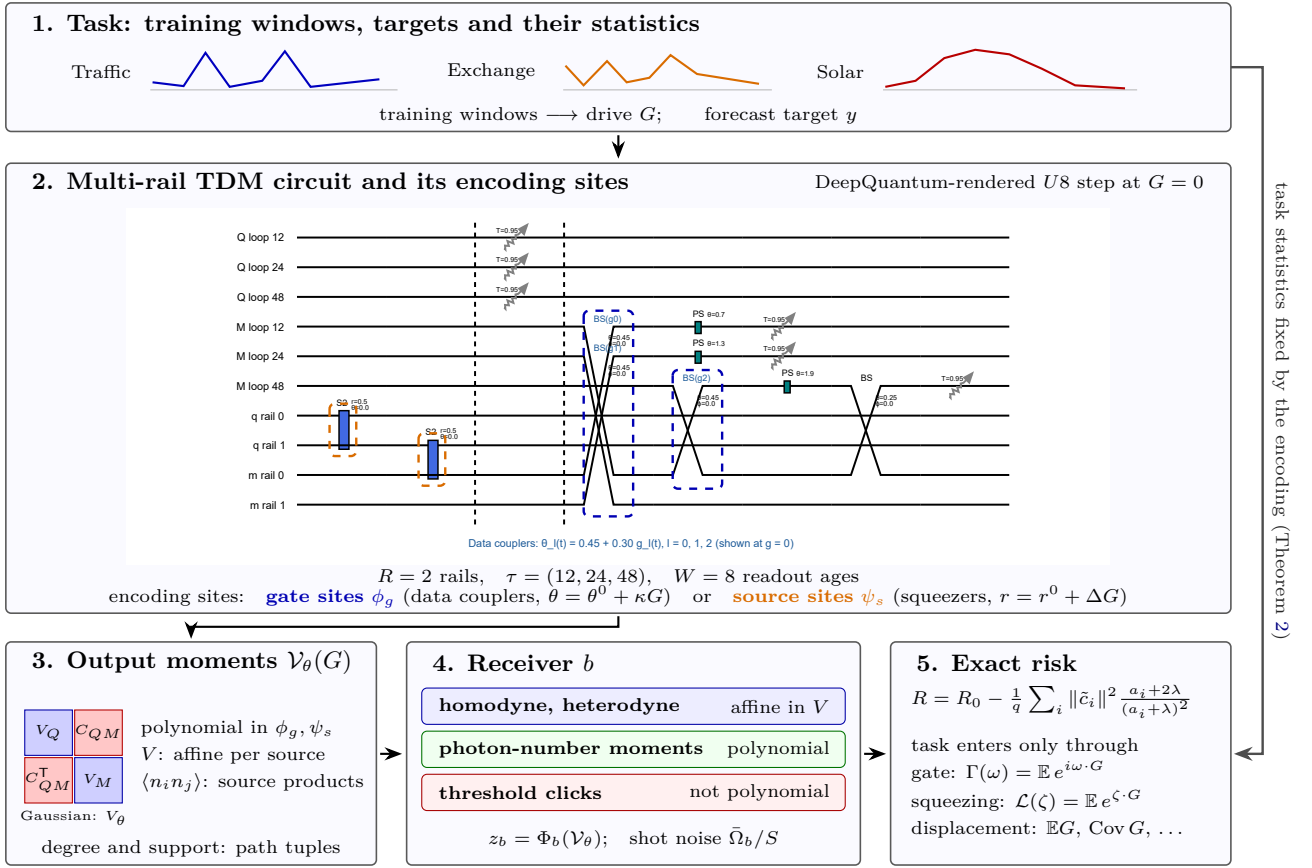
\begin{figure*}[t]
\centering
\begin{tikzpicture}[x=1cm,y=1cm,>=Stealth,
  panel/.style={draw=black!65,rounded corners=2pt,fill=blue!2,line width=.6pt},
  title/.style={font=\small\bfseries,anchor=west},
  note/.style={font=\scriptsize,align=center},
  line/.style={font=\scriptsize,anchor=west}]
\draw[panel] (0,9.70) rectangle (16.2,11.40);
\node[title] at (.22,11.11) {1. Task: training windows, targets and their statistics};
\begin{scope}[xshift=0.52cm]
\node[font=\scriptsize,anchor=west] at (.22,10.49) {Traffic};
\node[font=\scriptsize,anchor=west] at (5.19,10.49) {Exchange};
\node[font=\scriptsize,anchor=west] at (10.08,10.49) {Solar};
\draw[gray!55] (1.39,10.25)--(4.66,10.25) (6.85,10.25)--(9.61,10.25) (11.08,10.25)--(14.44,10.25);
\draw[blue!75!black,thick] (1.42,10.35)--(1.83,10.30)--(2.12,10.74)--(2.44,10.29)--(2.87,10.37)--(3.17,10.76)--(3.51,10.29)--(4.42,10.39);
\draw[orange!85!black,thick] (6.88,10.57)--(7.12,10.31)--(7.43,10.63)--(7.69,10.35)--(7.99,10.41)--(8.27,10.71)--(8.62,10.46)--(9.44,10.33);
\draw[red!72!black,thick] (11.11,10.29)--(11.51,10.37)--(11.89,10.67)--(12.30,10.78)--(12.75,10.72)--(13.18,10.53)--(13.62,10.31)--(14.28,10.27);
\end{scope}
\node[note] at (8.1,9.91) {training windows $\longrightarrow$ drive $G$;\qquad forecast target $y$};
\draw[->,line width=.8pt] (8.1,9.65)--(8.1,9.34);
\draw[panel] (0,3.33) rectangle (16.2,9.28);
\node[title] at (.22,9.02) {2. Multi-rail TDM circuit and its encoding sites};
\node[font=\scriptsize,anchor=east] at (15.97,9.02) {DeepQuantum-rendered $U8$ step at $G=0$};
\node[inner sep=0pt] at (8.1,6.33)
  {\includegraphics[width=13.0cm]{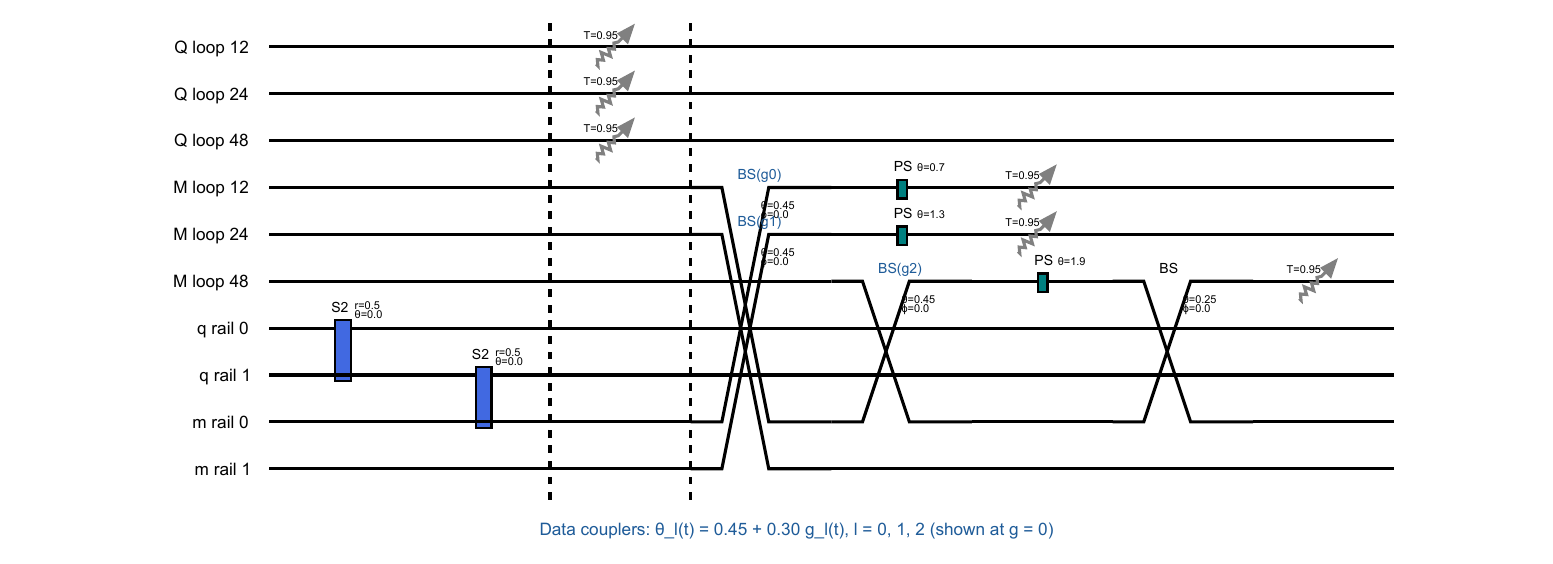}};
\begin{scope}[shift={(1.6,8.69)},x=0.0130811cm,y=-0.0130812cm]
  \draw[orange!85!black,line width=.9pt,dashed,rounded corners=2pt] (205,198) rectangle (232,251);
  \draw[orange!85!black,line width=.9pt,dashed,rounded corners=2pt] (295,228) rectangle (322,281);
  \draw[blue!70!black,line width=.9pt,dashed,rounded corners=2pt] (462,104) rectangle (512,312);
  \draw[blue!70!black,line width=.9pt,dashed,rounded corners=2pt] (552,164) rectangle (602,284);
\end{scope}
\node[note] at (8.1,3.80)
  {$R=2$ rails,\quad $\tau=(12,24,48)$,\quad $W=8$ readout ages};
\node[note] at (8.1,3.55)
  {encoding sites:\quad \textcolor{blue!70!black}{\textbf{gate sites} $\phi_g$} (data couplers, $\theta=\theta^0+\kappa G$)\quad or\quad \textcolor{orange!85!black}{\textbf{source sites} $\psi_s$} (squeezers, $r=r^0+\Delta G$)};
\draw[panel] (0,0) rectangle (4.90,2.95);
\draw[panel] (5.30,0) rectangle (11.30,2.95);
\draw[panel] (11.70,0) rectangle (16.2,2.95);
\draw[->,line width=.8pt] (8.1,3.29)--(8.1,3.15)--(2.45,3.15)--(2.45,2.99);
\draw[->,line width=.8pt] (4.95,1.47)--(5.24,1.47);
\draw[->,line width=.8pt] (11.35,1.47)--(11.64,1.47);
\node[title] at (.22,2.68) {3. Output moments $\mathcal V_\theta(G)$};
\draw[fill=blue!19,draw=blue!65!black] (.25,1.45) rectangle (.88,2.05);
\draw[fill=red!22,draw=red!70!black] (.91,1.45) rectangle (1.54,2.05);
\draw[fill=red!22,draw=red!70!black] (.25,.82) rectangle (.88,1.42);
\draw[fill=blue!19,draw=blue!65!black] (.91,.82) rectangle (1.54,1.42);
\node[font=\tiny] at (.565,1.75) {$V_Q$};
\node[font=\tiny] at (1.225,1.75) {$C_{QM}$};
\node[font=\tiny] at (.565,1.12) {$C_{QM}^{\mathsf T}$};
\node[font=\tiny] at (1.225,1.12) {$V_M$};
\node[font=\tiny] at (.895,.62) {Gaussian: $V_\theta$};
\node[line,align=left] at (1.66,1.44)
  {polynomial in $\phi_g,\psi_s$\\[2pt]
   $V$: affine per source\\[2pt]
   $\langle n_in_j\rangle$: source products};
\node[note] at (2.45,.22) {degree and support: path tuples};
\node[title] at (5.52,2.68) {4. Receiver $b$};
\draw[rounded corners=2pt,draw=blue!65!black,fill=blue!7] (5.50,1.86) rectangle (11.10,2.34);
\draw[rounded corners=2pt,draw=green!45!black,fill=green!6] (5.50,1.29) rectangle (11.10,1.77);
\draw[rounded corners=2pt,draw=red!65!black,fill=red!6] (5.50,.72) rectangle (11.10,1.20);
\node[font=\scriptsize\bfseries,anchor=west] at (5.60,2.10) {homodyne, heterodyne};
\node[font=\scriptsize,anchor=east] at (11.02,2.10) {affine in $V$};
\node[font=\scriptsize\bfseries,anchor=west] at (5.60,1.53) {photon-number moments};
\node[font=\scriptsize,anchor=east] at (11.02,1.53) {polynomial};
\node[font=\scriptsize\bfseries,anchor=west] at (5.60,.96) {threshold clicks};
\node[font=\scriptsize,anchor=east] at (11.02,.96) {not polynomial};
\node[note] at (8.3,.36) {$z_b=\Phi_b(\mathcal V_\theta)$;\quad shot noise $\bar\Omega_b/S$};
\node[title] at (11.92,2.68) {5. Exact risk};
\node[font=\scriptsize,anchor=west] at (11.85,2.17)
  {$R=R_0-\frac1q\sum_i\|\tilde c_i\|^2\frac{a_i+2\lambda}{(a_i+\lambda)^2}$};
\node[line,align=left] at (11.85,1.02)
  {task enters only through\\[2pt]
   gate: $\Gamma(\omega)=\E\,e^{i\omega\cdot G}$\\[2pt]
   squeezing: $\mathcal L(\zeta)=\E\,e^{\zeta\cdot G}$\\[2pt]
   displacement: $\E G$, $\operatorname{Cov}G$, \dots};
\draw[->,line width=.9pt,gray!60!black] (16.2,10.55)--(16.62,10.55)--(16.62,1.47)--(16.24,1.47);
\node[font=\scriptsize,text=gray!30!black,rotate=-90] at (16.86,6.0) {task statistics fixed by the encoding (\cref{cor:task})};
\end{tikzpicture}
\caption{The task--path--receiver bridge~\eqref{eq:chain}. (1) Training windows give the drive $G$ and targets $y$; the named tasks are examples from the seven evaluated datasets and the traces are schematic. (2) One time step of the $U8$ circuit, rendered with DeepQuantum~\cite{he2025deepquantum}; the full circuit is unrolled over $L=96$ steps. Data can enter at gate sites (the three couplers, blue: gate encoding, used in most of \cref{sec:numerics}) or at source sites (the squeezers, orange: squeezing encoding, \cref{tab:squeeze}). (3) Every retained output moment is a polynomial in the encoding features, with degree and support fixed by path tuples (\cref{thm:poly}); the covariance is affine in each source, while higher moments multiply sources (\cref{prop:gauss}). (4) Receivers differ in how their features depend on the moments; each has its own shot covariance. (5) The exact risk~\eqref{eq:risk} depends on the task only through statistics of $G$ and $y$ fixed by the encoding (\cref{cor:task}), indicated by the bypass arrow.}
\label{fig:chain}
\end{figure*}

\section{Setting}\label{sec:setting}

\subsection{Multi-rail TDM circuit}
The circuit has $R$ rails and $n_L$ fiber loops with delays $\tau_1,\dots,\tau_{n_L}$ (in time steps); loop $l$ couples to rail $r(l)$.
At each time step $t=1,\dots,L$ a two-mode squeezed vacuum (TMSV)
pair with squeezing $r$ is injected on every rail: one mode enters the
loop network (memory arm) and the other a reference arm. Temporally
multiplexed EPR-pair generation and programmable squeezed-light generation have been demonstrated experimentally~\cite{larsen2019epr,tomoda2023programmable}.
In a time step, the time-bin mode of each rail meets the loops in a fixed gate order.
For every loop $l$ the step applies an optional fixed \texttt{cross} gate between adjacent loops, a coupling beam splitter between loop $l$ and the time bin of rail $r(l)$,
optionally an electro-optic phase modulator (EOM) on the loop, and a fixed phase $\psi_l$. The step ends with a readout beam splitter of fixed angle,
optional inter-rail beam splitters, and loop loss $\eta_{\text{loop}}$.
Output time bins of both arms are retained at $W$ readout ages, giving $m=2WR$ retained modes.
Universal multiport decompositions provide one implementation route for programmable inter-rail linear optics~\cite{clements2016optimal}; no such universal mesh is assumed in the model.

\paragraph{How the drive enters.} The input is a multivariate series, standardized on the training segment and projected on its $P$ leading principal
directions $w_1,\dots,w_P$ (estimated on the training segment). A window of length $L$ gives the drive $G\in(-1,1)^{L\times P}$, $G_{t,c}=\tanh(w_c^\top s_t)$.
Channel $c$ drives every loop $l$ with $c(l)=c$, where $c(l)=l \bmod P$. At step $t$ the drive enters affinely:
\begin{equation}\label{eq:affine}
\begin{aligned}
\theta_{t,l}&=\theta_l^0+\kappa G_{t,c(l)}
  &&\text{(coupling angle)},\\
\varphi_{t,l}&=aG_{t,c(l)}
  &&\text{(EOM phase)}.
\end{aligned}
\end{equation}
The encoding is \emph{coupling} ($\kappa\neq0$, no EOM), \emph{phase}, or \emph{both}. For a drive component $\alpha=(t,c)$ let
$B_\alpha$ and $E_\alpha$ be the numbers of coupling gates and EOM gates it feeds; with the channel assignment above, $B_\alpha=E_\alpha=\#\{l:c(l)=c\}$ whenever the corresponding encoding is on.

\paragraph{General model: encoding sites.}
The theory below does not depend on this template. Consider any time-unrolled passive linear network, with losses represented by beam splitters coupled to vacuum environment modes, acting on a product of source states (a TMSV pair counts as one source). Call a gate or a source an \emph{encoding site} if it depends on the drive. A gate site $g$ is admitted when its matrix is affine in a finite feature vector $\phi_g(G)$ and its complex conjugate. A source site $s$ is admitted when each of its normally ordered moments of order at most $k$ is a polynomial in a finite feature vector $\psi_s(G)$ and its conjugate; $d_s(k)$ denotes the largest such degree. Examples:
coupling and EOM gates, with $\phi_g=(e^{i\kappa G_\alpha},e^{-i\kappa G_\alpha})$ or $e^{iaG_\alpha}$;
\emph{squeezing encoding}, where the squeezing $r=r(G_\alpha)$ of a source is data dependent and $\psi_s=(\sinh^2r,\sinh r\cosh r)$, with $d_s(k)=\lfloor k/2\rfloor$ (\cref{prop:gauss});
squeezing-phase encoding, with $\psi_s=e^{i\vartheta(G_\alpha)}$ and the same degrees;
and \emph{displacement encoding}, with coherent amplitude $\psi_s=\beta(G_\alpha)$ and $d_s(k)=k$.
Features may be arbitrary nonlinear functions of the drive; only polynomial dependence of the site on its features is required. Drive-independent gates, sources and losses contribute constants ($d_s=0$); Fock-state and other non-Gaussian sources are admitted in this way.

\subsection{Propagation of moments}
All gates are passive and linear. We write the circuit as a linear map $a_{\text{out}}=T a_{\text{in}}+Eb$ on annihilation operators, where losses are beam splitters with vacuum environment modes $b$, so $T$ is a submatrix of a unitary.
The circuit is unrolled in time: every gate at every step is a separate factor of $T$.
The mean $\mu_j=\langle a_j\rangle$ and the moment matrices $N_{jk}=\langle a_j^\dagger a_k\rangle$ and $M_{jk}=\langle a_j a_k\rangle$ transform as
\begin{equation}\label{eq:NM}
\mu_{\text{out}}=T\mu_{\text{in}},\quad
N_{\text{out}}=\bar T N_{\text{in}} T^\top,\quad M_{\text{out}}=T M_{\text{in}} T^\top,
\end{equation}
and every higher normally ordered moment transforms with one factor of $T$ or $\bar T$ per operator (\cref{sec:encoding}).
When all sources are zero-mean Gaussian, as in the template above and in gate and squeezing encoding, the output state is a zero-mean Gaussian state fixed by its covariance matrix
$V\in\R^{2m\times2m}$ (quadrature ordering $xxpp$, $\hbar=2$, vacuum $V=I$; physicality $V+i\Omega\succeq0$ with $\Omega$ the symplectic form~\cite{weedbrook2012gaussian}), which is an affine function of $(N,M,\bar M)$.
Displacement encoding adds the mean $\mu$; for non-Gaussian sources the higher moments must be propagated themselves.

\subsection{Receivers}\label{sec:setting-receivers}
Each receiver uses some of the normally ordered moments of the retained output modes. We collect them in $\mathcal V_\theta(G)$ and write $z_b=\Phi_b(\mathcal V_\theta)$ for its mean features and $\Omega_b(\mathcal V_\theta)$ for their single-shot noise covariance. For zero-mean Gaussian states $\mathcal V_\theta=V_\theta$; displacement adds the mean, and non-Gaussian inputs add the higher moments the receiver needs (for threshold clicks, in general all of them).
Detection efficiency $\eta_{\det}$ is applied as the same loss before every receiver; for Gaussian states $V_\eta=\eta_{\det}V+(1-\eta_{\det})I$. We consider
(i) \emph{heterodyne}: features are second moments of $V'=V_\eta+I$;
(ii) \emph{homodyne} with a finite family of local-oscillator settings, each measured on a fraction of the shots, whose moments are combined into a fixed feature set (time-multiplexed homodyne operations have been demonstrated~\cite{asavanant2021timedomain});
(iii) \emph{threshold click}: features are single-mode and selected two-mode click probabilities, which are non-polynomial functions of $V_\eta$~\cite{quesada2018gaussian,bulmer2022threshold} (approximated uniformly in \cref{sec:clicks});
(iv) \emph{photon-number moments}: features are $\langle n_i\rangle$ and $\langle n_in_j\rangle$ estimated from photon-number-resolving counts.
We call a receiver \emph{moment-polynomial} if its mean feature and single-shot noise covariance are polynomials in the normally ordered output moments. Heterodyne and homodyne are moment-polynomial and, for zero-mean states, moreover \emph{affine}: their raw mean feature is affine in $V$, and their single-shot noise covariance is a quadratic polynomial in $V$ (Wick's theorem). Photon-number moments are moment-polynomial; for zero-mean Gaussian states $\langle n_in_j\rangle$ is quadratic in $V$. Threshold clicks are not moment-polynomial for Gaussian states.

\subsection{Learning task and finite-shot ridge readout}
Each window has a target $y\in\R^q$. With $S$ shots per window the receiver returns $\hat z$ with $\E[\hat z\mid G]=z_b(G)$ and $\operatorname{Cov}(\hat z\mid G)=\Omega_b(\mathcal V_\theta(G))/S$.
A ridge readout with unpenalized intercept is trained on standardized features. Henceforth $\Phi_b$ and $\Omega_b$ include the fixed diagonal rescaling determined on the training segment (after dropping constant coordinates); the raw receiver map and noise law are rescaled accordingly. This scaling is task-dependent but introduces no new drive frequencies.
Throughout, ``risk'' means the prediction risk of the \emph{population} ridge readout, fitted to the exact moments of the empirical distribution of training windows under finite-shot feature noise. It contains no error from fitting a readout to one finite noisy sample and is not a held-out generalization bound.

\paragraph{Conditions.} Throughout:
(C1) moments are taken under the empirical measure of the $n$ training windows;
(C2) the measurement is conditionally unbiased with the covariance above;
(C3) given $G$, measurement noise is uncorrelated with $y$;
(C4) features and targets are centered by their training means and the intercept is not penalized;
(C5) $\lambda$ and the feature standardization are fixed;
(C6) at $\lambda=0$ a pseudo-inverse is used on the range of the Gram matrix.

\section{The exact polynomial bridge}\label{sec:bridge}
The bridge follows the chain~\eqref{eq:chain} in three steps: what the
risk reads from the task (\cref{sec:task-side}), how the encoding and
the optical paths imprint the drive on the light (\cref{sec:encoding}),
and how the receiver reads it out (\cref{sec:readout}). In
\cref{fig:chain} these steps are panels (1) and (5), panels (2) and (3),
and panel (4), respectively. The path structure itself is analyzed in
\cref{sec:structure}.

We first isolate the learning calculation from the optical one. Set
$\Sigma_z=\operatorname{Cov}(z_b)$, $C_z=\operatorname{Cov}(z_b,y)$,
$\bar\Omega_b=\E\,\Omega_b(\mathcal V_\theta(G))$, and
$R_0=\E\|y-\E y\|^2/q$. Under (C1)--(C6), the measured feature has
covariance $A=\Sigma_z+\bar\Omega_b/S$ and its ridge head is
$W_\lambda=(A+\lambda I)^{-1}C_z$. Diagonalize
$A=U\operatorname{diag}(a_i)U^\top$ and put
$\tilde c_i=(U^\top C_z)_{i,:}$. Expanding the prediction error gives
\begin{equation}\label{eq:risk}
\begin{aligned}
R_{\theta,b,S}(\lambda)
&=\frac1q\E\|y-\E y-W_\lambda^\top(\hat z-\E\hat z)\|^2\\
&=R_0-\frac1q\sum_i\|\tilde c_i\|^2
  \frac{a_i+2\lambda}{(a_i+\lambda)^2}.
\end{aligned}
\end{equation}
This is the risk of the population ridge readout on the empirical
distribution of training windows under finite-shot feature noise; it
contains no finite-sample fitting error and is not a held-out bound.
At this population level $\partial_\lambda R\ge0$, so the formula
cannot select $\lambda$; architecture comparisons use a common fixed
value. At $\lambda=0$ the shot noise itself regularizes, and the whole
dependence on the shot budget is explicit (\cref{prop:shots}).

\subsection{Task features: what the risk reads from the data}\label{sec:task-side}
The drive enters the circuit only through the encoding features
$\phi_g(G)$ and $\psi_s(G)$. \Cref{sec:encoding} shows that the output
moments, and with them the features and shot noise of every
moment-polynomial receiver, are finite sums of monomials in these
features. The learning calculation~\eqref{eq:risk} then touches the data
only through moments of such monomials.
Let the (empirical) characteristic function of the drive windows and
its label-weighted version be
\begin{equation}\label{eq:gamma}
\Gamma(\omega)=\E\,e^{i\omega\cdot G},\qquad
\Gamma_y(\omega)=\E\,e^{i\omega\cdot G}y^\top.
\end{equation}

\begin{theorem}[Exact task--risk bridge]\label{thm:bridge}
Under (C1)--(C6), let the receiver be moment-polynomial. By
\cref{thm:poly}, its mean feature and single-shot noise covariance are
finite sums $z_b(G)=\sum_{m\in\mathcal M}\hat z_m\,m(G)$ and
$\Omega_b(G)=\sum_{m\in\mathcal M'}\hat\Omega_m\,m(G)$ over monomials
in the encoding features. Then
$\E z_b=\sum_m\hat z_m\E m$,
$\E(z_bz_b^\top)=\sum_{m,m'}\hat z_m\hat z_{m'}^\top\E(mm')$,
$\E(z_by^\top)=\sum_m\hat z_m\E(my^\top)$ and
$\bar\Omega_b=\sum_m\hat\Omega_m\E m$, so the risk~\eqref{eq:risk}
depends on the task only through the joint empirical moments of the
finite dictionary $\mathcal M\cup\mathcal M\mathcal M\cup\mathcal M'$,
their label cross-moments, $\E y$ and $R_0$.
For gate encoding (\cref{cor:gate}) and an affine receiver the
monomials are $e^{i\omega\cdot G}$, $\hat z_\omega=H_b\hat V_\omega$ for
$\omega\ne0$, and the moments are characteristic-function values:
\begin{equation}\label{eq:spectral}
\begin{aligned}
\E z_b&=\sum_\omega\hat z_\omega\Gamma(\omega),\\
\E(z_bz_b^\top)&=\sum_{\omega,\omega'}
  \hat z_\omega\hat z_{\omega'}^\top\Gamma(\omega+\omega'),\\
\E(z_by^\top)&=\sum_\omega\hat z_\omega\Gamma_y(\omega),\\
\bar\Omega_b&=\sum_\nu\hat\Omega_\nu\Gamma(\nu),
\end{aligned}
\end{equation}
with $\nu\in\mathcal W_\theta+\mathcal W_\theta$. The risk then depends on the
task only through $\Gamma|_{\mathcal W_\theta+\mathcal W_\theta}$,
$\Gamma_y|_{\mathcal W_\theta}$, $\E y$ and $R_0$; the circuit supplies
$\{\hat V_\omega\}$, the raw receiver supplies its feature map and
shot-noise law, and the task moments also fix the training-segment
standardization included in $(H_b,\Omega_b)$.
\end{theorem}
\begin{proof}
Multiplying the finite sums and averaging over training windows gives
the moments; in particular
$\Sigma_z=\sum_{m,m'}\hat z_m\hat z_{m'}^\top[\E(mm')-\E m\,\E m']$ and
$C_z=\sum_m\hat z_m[\E(my^\top)-\E m\,\E y^\top]$. For zero-mean
Gaussian states, Wick's identity makes the homodyne and heterodyne $\Omega_b$
quadratic in $V$, and $\langle n_in_j\rangle$ and the variances of the
counts are polynomials in $V$ by Isserlis' theorem; in general every
entry is a polynomial in normally ordered moments, to which
\cref{thm:poly} applies.
Substitution in \eqref{eq:risk} completes the proof; the fixed feature
standardization is determined by $\operatorname{diag}\Sigma_z$. The
gate-encoding case is detailed in \cref{app:proof-trig}.
\end{proof}

For gate encoding the finite dictionary may be exponentially large, so the exact
identity is a structural description rather than a faster evaluator;
for squeezing encoding with an affine receiver it has two monomials per
source.

Which statistics of the data these monomial moments are depends only
on how the encoding features depend on the drive. Let $\mathcal W^{(k)}$ collect the frequency vectors with
components $\omega_\alpha=n_1\kappa+n_2a$, $|n_1|\le kB_\alpha$,
$|n_2|\le kE_\alpha$, so that $\mathcal W_\theta=\mathcal W^{(2)}$. For squeezing
encoding with $r=r^0+\Delta G$, let
$\Lambda=\{0\}\cup\{\pm2\Delta e_\alpha\}$. Define the empirical
two-sided Laplace transform of the windows and its label-weighted
version,
\begin{equation}\label{eq:mgf}
\mathcal L(\zeta)=\E\,e^{\zeta\cdot G},\qquad
\mathcal L_y(\zeta)=\E\,e^{\zeta\cdot G}y^\top,\qquad \zeta\in\C^{L\times P}.
\end{equation}
The characteristic function is $\Gamma(\omega)=\mathcal L(i\omega)$; on
real arguments $\mathcal L$ is the moment generating function.

\begin{corollary}[What the task supplies]\label{cor:task}
Under the conditions of \cref{thm:bridge}, the risk depends on the task
only through $\E y$, $R_0$ and the statistics of the training windows
listed in \cref{tab:task}. Only drive components, and pairs or tuples of
them, admitted by the path rule enter.
\end{corollary}

\begin{table}[t]
\centering\footnotesize
\caption{Task statistics that enter the exact risk (\cref{cor:task}). Feature and noise orders $k,k'$ count the normally ordered moments involved: $(2,4)$ for receivers affine in $V$, $(4,8)$ for photon-number second moments, $(4wd,8wd)$ for the order-$d$ click approximant of \cref{prop:click}. $\Lambda^{+j}$ is the $j$-fold sum of $\Lambda$.}
\label{tab:task}
\setlength{\tabcolsep}{2.5pt}
\begin{tabular}{@{}>{\raggedright}p{2.3cm}>{\raggedright}p{2.0cm}>{\raggedright\arraybackslash}p{3.8cm}@{}}
\toprule
Encoding & Receiver & Task statistics\\ \midrule
gate & orders $(k,k')$ & $\Gamma$ on $(\mathcal W^{(k)}+\mathcal W^{(k)})\cup\mathcal W^{(k')}$; $\Gamma_y$ on $\mathcal W^{(k)}$\\
gate, $\le N$ photons & detection probabilities & $\Gamma$ on $\mathcal W^{(2N)}+\mathcal W^{(2N)}$; $\Gamma_y$ on $\mathcal W^{(2N)}$\\
squeezing, $r=r^0+\Delta G$ & affine in $V$ & $\mathcal L$ on $\Lambda^{+2}$; $\mathcal L_y$ on $\Lambda$\\
squeezing, $r=r^0+\Delta G$ & quadratic in $V$ & $\mathcal L$ on $\Lambda^{+4}$; $\mathcal L_y$ on $\Lambda^{+2}$\\
squeezing, any $r(G)$ & affine in $V$ & first and second moments of $(u_t,v_t)$; their covariances with $y$\\
displacement, $\beta=\beta^0+bG$ & quadrature means & $\E G$, $\operatorname{Cov}G$, $\operatorname{Cov}(G,y)$\\
displacement, $\beta=\beta^0+bG$ & second moments, $\langle n_i\rangle$ & moments of $G$ to order $4$; $\E(G_\alpha G_\beta y)$, $\E(G_\alpha y)$\\ \bottomrule
\end{tabular}
\end{table}

The table orders the encodings by what they can read. Displacement
encoding with quadrature means reads only second-order statistics, so its
risk is bounded below by linear regression on the reachable drive
components. Squeezing encoding with an affine receiver reads the moment
generating function of single encoded values and of pairs of them at a
few real arguments, but its features contain no product across time. Gate encoding
reads the characteristic function on a lattice whose vectors involve
many lags at once. Gate and squeezing encodings therefore sample one
object, the two-sided Laplace transform $\mathcal L$ of the window
distribution, on the imaginary and on the real axis respectively. All entries are
empirical transforms of the training windows, available before any
circuit is simulated.

\subsection{Encoding: output moments are polynomials in encoding features}\label{sec:encoding}
For a passive linear network, $a_{\rm out}=Ta_{\rm in}+Eb$ with
vacuum environment modes $b$, which do not contribute to normally
ordered moments. Each output operator in an order-$k$ normally
ordered moment therefore contributes one entry of $T$ or $\bar T$,
and the moment is a sum over $k$-tuples of input modes and optical
paths (\cref{sec:structure}) of the path amplitudes times the
corresponding order-$k$ input moment. Call such a $k$-tuple
\emph{source-compatible} when that joint input moment is nonzero.

\begin{theorem}[Encoding-feature polynomial representation]\label{thm:poly}
For every admitted network of \cref{sec:setting}, each normally
ordered output moment of order $k$ is a polynomial in the encoding
features $\{\phi_g\}$, $\{\psi_s\}$ and their complex conjugates, and:
\begin{enumerate}
\item[(a)] its degree in the features of a gate site $g$ is at most the
number of legs of a source-compatible $k$-tuple whose paths can pass
$g$, hence at most $k$;
\item[(b)] its degree in the features of a source site $s$ is at most
$d_s(k)$;
\item[(c)] a monomial containing features of the gate sites in $S$ and
of the source sites in $\Sigma$ can occur only if some
source-compatible $k$-tuple of paths passes every gate in $S$ and has a
leg starting at a mode of every source in $\Sigma$.
\end{enumerate}
\end{theorem}
\begin{proof}[Proof sketch]
The transfer matrix is an ordered product in which each gate instance
occurs once, so every path amplitude is multilinear in the gate
features on that path. An order-$k$ moment multiplies $k$ path
amplitudes (or their conjugates) by one input moment of order $k$,
which factorizes over sources for a product input state; a source
reached by no leg contributes the factor $1$. Summing terms can cancel
monomials but never create one absent from every term. Details are in
\cref{app:proof-poly}.
\end{proof}

The absence of cross-source products is a second-order property.
\begin{proposition}[Gaussian sources: covariance versus higher moments]\label{prop:gauss}
Let every source be zero-mean Gaussian, with second moments affine in
its features and their conjugates, as in squeezing and squeezing-phase
encoding.
\begin{enumerate}
\item[(i)] The output moments $N$ and $M$, and hence $V$, are affine in
each $\psi_s$ and contain no product of features of two distinct
sources.
\item[(ii)] An output moment of order $2j$ is a sum of products of $j$
output second moments; its total degree in the source features is at
most $j$, and for $j\ge2$ it can contain products of features of
distinct sources. Odd orders vanish.
\end{enumerate}
\end{proposition}
For example, with $u=\sinh^2r$, a single squeezed vacuum has
$\langle a^{\dagger2}a^2\rangle=3u^2+u$, and two independent squeezed
vacua on unmixed modes have $\langle n_1n_2\rangle=u_1u_2$. Receivers
affine in $V$ therefore see no cross-source products, while
photon-number correlations and click probabilities do.

The theorem specializes as follows.

\begin{corollary}[Gate encoding: a finite trigonometric polynomial]\label{cor:gate}
With coupling and EOM encoding~\eqref{eq:affine} and drive-independent
zero-mean Gaussian sources, at every drive value
\begin{equation}\label{eq:trig}
\begin{aligned}
V_\theta(G)&=\sum_{\omega\in\mathcal W_\theta}
  \hat V_\omega e^{i\omega\cdot G},\\
\mathcal W_\theta&\subseteq\prod_\alpha\mathcal W_\alpha,
  \qquad\mathcal W_\alpha=\{n_1\kappa+n_2a\},
\end{aligned}
\end{equation}
where $|n_1|\le2B_\alpha$, $|n_2|\le2E_\alpha$,
$\omega\cdot G=\sum_\alpha\omega_\alpha G_\alpha$, and
$\hat V_{-\omega}=\overline{\hat V_\omega}$.
The candidate set $\mathcal W_\theta$ includes the zero frequency;
coefficients that vanish are allowed, and coincident numerical
frequencies are collected.
\end{corollary}

\begin{corollary}[Squeezing encoding: a linear filter]\label{cor:squeeze}
Let all gates be drive independent and let the source injected at step
$t$ have squeezing $r_t=r(G_t)$ with a fixed squeezing phase. Then,
with $u_t=\sinh^2r_t$ and $v_t=\sinh r_t\cosh r_t$,
\begin{equation}\label{eq:squeeze}
V(G)=V_0+\sum_t\bigl(A_tu_t+B_tv_t\bigr)
\end{equation}
for fixed matrices $V_0,A_t,B_t$; no product of two time steps occurs
in the covariance. An affine receiver's features are a fixed linear filter of
$(u_t,v_t)$, and its risk depends on the task only through the empirical
means and lagged covariances of $(u_t,v_t)$ and their covariances with
$y$, at a cost polynomial in the number of sources. Cross-time products
arise only from a receiver that is nonlinear in $V$; photon-number second
moments, for example, are quadratic in $(u_t,v_t)$ (\cref{prop:gauss}).
\end{corollary}

\begin{remark}[Displacement encoding]\label{cor:displace}
If all gates are drive independent and the Gaussian source $s$ carries
a data-dependent coherent amplitude $\beta_s(G)$, the output mean is
the linear filter $\mu=\sum_sT_{:,s}\beta_s$ and the covariance is drive
independent. Homodyne mean features are then linear in $\{\beta_s\}$,
and photon-number features add terms quadratic in
$\{\beta_s,\bar\beta_s\}$.
\end{remark}

\begin{corollary}[Bounded photon number]\label{cor:fock}
If the input state contains at most $N$ photons, for example a Fock
state, every detection probability is a polynomial of degree at most
$2n_g$ in the features of gate $g$, where $n_g\le N$ is the largest
number of photons whose paths can pass $g$. A coupling gate therefore
contributes frequencies $|\omega|\le2n_g\kappa$ and a phase gate
$|\omega|\le n_ga$, because each amplitude contains only nonnegative
powers of $e^{i\varphi}$; the bounds of gates fed by one drive
component add. For a single phase-encoded gate passed by all
$N$ photons this is the spectrum $\{-N,\ldots,N\}$ of~\cite{gan2022fock}.
\end{corollary}

\begin{remark}[Drive noise damps harmonics]\label{prop:noise}
If the drive component $G_\alpha$ is perturbed in each shot by
$\epsilon\sim\mathcal N(0,\sigma^2)$, independent of the window and
shared by the gates that $G_\alpha$ feeds, every expectation value of
the shot-averaged state whose dependence on $G_\alpha$ is a finite or
absolutely convergent sum of harmonics, in particular every
moment-polynomial feature under gate encoding, keeps its expansion with
each coefficient of $e^{i\omega_\alpha G_\alpha}$ multiplied by
$e^{-\omega_\alpha^2\sigma^2/2}$. Independent noise on different
components multiplies the factors, and the feature noise covariance
gains the variance of the conditional mean over $\epsilon$. Fabrication
errors that are fixed in time only change the coefficients of these
expansions, not their form.
\end{remark}

The encoding therefore fixes the function class a linear readout can
use, and with it a classical baseline that no circuit of the class can
beat.

\begin{corollary}[Encoding decides the function class]\label{cor:class}
Let the sources be Gaussian and the receiver affine in $V$ for
zero-mean states, or return quadrature means, and let
$\mathcal R$ be the readable drive components.
\begin{enumerate}
\item[(i)] Displacement encoding, $\beta_s=\beta_s^0+b_sG_{\alpha(s)}$, with
quadrature means: $z_b$ is affine in $G$, and
$R\ge R_{\rm or}(G_{\mathcal R})$, the risk of linear regression on
$G_{\mathcal R}$.
\item[(ii)] Squeezing encoding with an affine receiver: $z_b$ is the
additive model
$z_0+\sum_s[a_s\,u(G_{\alpha(s)})+b_s\,v(G_{\alpha(s)})]$, and
$R\ge R_{\rm or}(\psi_{\mathcal R})$. A target orthogonal to every
additive function of single drive components is not captured at all,
for every topology, operating point and shot budget.
\item[(iii)] Gate encoding with an affine receiver: $z_b$ is a
trigonometric polynomial whose interactions join only sites connected
by source-compatible path pairs, all inside one component of
\cref{thm:components}.
\end{enumerate}
\end{corollary}
Products across time steps therefore enter a linear readout only through
the gates in (iii), or at detection, through receivers nonlinear in $V$
or product observables spanning several components
(\cref{cor:span}).

\subsection{Readout: receivers, shot noise and threshold clicks}\label{sec:readout}
The receiver turns the output moments into the features
$z_b=\Phi_b(\mathcal V_\theta)$ of the ridge readout
(\cref{sec:setting-receivers}). The state identities of
\cref{sec:encoding} hold for any receiver, and the risk closes exactly
through \cref{thm:bridge} for every moment-polynomial receiver, which
includes homodyne, heterodyne and photon-number moments.
For Gaussian states, threshold-click probabilities are not polynomial in $V$; they
admit finite dictionaries of arbitrary precision (\cref{sec:clicks}) and
a cheaper local expansion (\cref{sec:truncation}), and with at most $N$
input photons they are polynomial (\cref{cor:fock}).
Two further properties are decided by the receiver alone: how the risk
depends on the shot budget, and which interactions between independent
parts of the network a linear readout can combine (\cref{sec:components}).

\paragraph{Shot noise.} At $\lambda=0$ the dependence of the
risk~\eqref{eq:risk} on the shot budget has a closed form in the
receiver's single-shot signal-to-noise spectrum.
\begin{proposition}[Exact shot law]\label{prop:shots}
Under (C1)--(C6) with $\lambda=0$, drop drive-independent features and
let $\bar\Omega_b\succ0$. Let $\Sigma_zw_i=g_i\bar\Omega_bw_i$ with
$w_i^\top\bar\Omega_bw_j=\delta_{ij}$, and put
$\omega_i=\|w_i^\top C_z\|^2/(qg_i)$ for $g_i>0$. Then
\begin{equation}\label{eq:shotlaw}
R_{\theta,b,S}(0)=R_\infty+\sum_{i:\,g_i>0}\frac{\omega_i}{1+g_iS},
\qquad R_\infty=R_0-\sum_i\omega_i .
\end{equation}
\end{proposition}
Here $g_i$ is the single-shot signal-to-noise ratio of feature direction
$i$ and $\omega_i$ the share of target variance it can capture. Each
direction contributes half its share at $S=1/g_i$; the number of
task-relevant directions resolved at budget $S$ is
$\#\{i:g_iS\ge1\}$, and $S\le\sum_i\omega_i/(g_i\varepsilon)$ suffices to
come within $\varepsilon$ of $R_\infty$. A receiver thus enters the
finite-shot risk only through its spectrum $\{(g_i,\omega_i)\}$, which
compares receivers on a given task without sampling. Unlike a
task-independent conditioned rank~\cite{nerenberg2025photon}, the
weights $\omega_i$ select the directions the task uses.

\paragraph{Threshold clicks: finite dictionaries of arbitrary precision.}\label{sec:clicks}
For a zero-mean Gaussian state the probability that no mode of a set
$J$ clicks is~\cite{quesada2018gaussian,bulmer2022threshold}
\begin{equation}\label{eq:p0}
p_0(J\mid G)=\det Q_J(G)^{-1/2},\qquad
Q_J=\tfrac12\bigl(V_{\eta,J}(G)+I\bigr),
\end{equation}
where $V_{\eta,J}$ is the $2|J|\times2|J|$ block of the modes in $J$.
The probability that every mode of $A$ clicks is
$p_A=\sum_{J\subseteq A}(-1)^{|J|}p_0(J)$, and the single-shot noise
covariance of the click features is $p_{A\cup B}-p_Ap_B$. The determinant
is not a polynomial in $V$, but on a spectrally bounded set it is
uniformly approximable by polynomials.

\begin{proposition}[Threshold clicks: finite dictionaries of arbitrary precision]\label{prop:click}
Let the network be passive with vacuum environment modes and let every
source be Gaussian with covariance spectrum in
$[e^{-2r},e^{2r}]$. Then, for every drive value and mode set $J$, the
spectrum of $Q_J$ lies in $[a,b]$ with
$a=1-\tfrac12\eta_{\det}(1-e^{-2r})$ and
$b=1+\tfrac12\eta_{\det}(e^{2r}-1)$. Let $p_d$ be a polynomial of
degree $d$ with $\varepsilon_d=\max_{q\in[a,b]}|\sqrt q\,p_d(q)-1|$; for
the Chebyshev interpolant of $q^{-1/2}$ on $[a,b]$,
$\varepsilon_d=O(\varrho^{-d})$ for every
$1<\varrho<(\sqrt b+\sqrt a)/(\sqrt b-\sqrt a)$. Then
$P_{J,d}=\det p_d(Q_J)$ is a polynomial of degree $2|J|d$ in the
entries of $V$, and uniformly in $G$
\begin{equation}\label{eq:p0-approx}
|P_{J,d}-p_0(J)|\le\bigl[(1+\varepsilon_d)^{2|J|}-1\bigr]\,p_0(J).
\end{equation}
Replacing $p_0$ by $P_{J,d}$ in the click features and their noise
defines the \emph{order-$d$ click approximant}, a moment-polynomial
receiver. For events on at most $w$ modes its features and noise are
polynomials in $V$ of degrees $2wd$ and $4wd$; for gate encoding they
are trigonometric polynomials with frequencies in $\mathcal W^{(4wd)}$ and
$\mathcal W^{(8wd)}$, so its risk obeys \cref{thm:bridge} exactly. Its
feature for event $A$ differs from the click feature by at most
$e_A=\sum_{\emptyset\ne J\subseteq A}[(1+\varepsilon_d)^{2|J|}-1]$, and
its noise entry by at most $(e_{A\cup B}+e_A+e_B+e_Ae_B)/S$, uniformly
in $G$. With these errors the certificate of
\cref{app:proof-riskerror} bounds the difference between the click risk
and the approximant's risk whenever its spectral condition holds.
\end{proposition}

The approximant is a finite dictionary of arbitrary precision, not an
exact finite representation of the click receiver, whose risk is the
limit $d\to\infty$. The spectral bound is uniform because passive
optics cannot amplify: every output covariance block has its spectrum
between those of the sources and of the vacuum. Unlike the second-order
truncation of \cref{sec:truncation}, which expands in the drive around
its mean, the approximant holds for every drive value, at the price of a
larger dictionary. At $U8$ ($r=0.5$, $\eta_{\det}=0.85$) the interval is
$[0.731,1.730]$ and $\varepsilon_d$ falls by about a factor of five per
degree (\cref{sec:numerics}). For squeezing encoding the same
approximant, with $r$ the largest squeezing, is a polynomial of degree
$2wd$ in the $(u_t,v_t)$.

\section{Topology-imposed limits}\label{sec:structure}

\subsection{Path amplitudes and structural zeros}
Unroll the circuit in time and gate order into a directed acyclic graph
whose vertices are (mode, gate) incidences; \cref{fig:dag} shows this
graph for a single rail with two loops over four steps. A transfer entry is the sum
of its source-to-output path amplitudes,
\begin{equation}\label{eq:path-sum}
T_{os}=\sum_{\pi:s\to o}\operatorname{amp}(\pi),\qquad
\operatorname{amp}(\pi)=\prod_{e\in\pi}w_e.
\end{equation}
A mixing gate contributes $\cos\theta$ for staying on a mode and
$\pm\sin\theta$ for switching modes; a phase gate contributes
$e^{i\varphi}$ and a loss contributes $\sqrt\eta$.
The staying factors include beam splitters that a time bin traverses without entering the loop, such as the readout gate on every round trip (\cref{fig:dag}).
Within a single time step, the time-bin mode meets the loops in gate order, so a path can move from loop $l$ to loop $l'$ within that step only if $l'$ comes later in the gate order.
We call this gate-order-induced directed loop switching. It creates structural zeros that number-theoretic arguments on the delays alone do not see.

\begin{figure*}[t]
\centering
\begin{tikzpicture}[>=Stealth, font=\scriptsize, x=1.2cm, y=1.05cm,
  gnode/.style={draw, fill=white, minimum size=4.6mm, inner sep=0pt},
  snode/.style={draw, circle, fill=gray!15, minimum size=5mm, inner sep=0pt},
  onode/.style={draw, circle, minimum size=5mm, inner sep=0pt},
  eg/.style={->, gray!60, thin},
  pa/.style={->, line width=1.1pt, red!75!black},
  pb/.style={->, line width=1.1pt, blue!70!black},
  pc/.style={->, line width=1.1pt, green!50!black}]
\foreach \t in {1,...,4} {
  \pgfmathsetmacro{\x}{3.0*(\t-1)}
  \node[snode] (s\t) at (\x,3) {$s_{\t}$};
  \node[anchor=west, font=\tiny, text=orange!85!black, inner sep=1pt] at (s\t.east) {$\psi_{\t}$};
  \node[gnode] (a\t) at (\x,2) {};
  \node[gnode] (b\t) at (\x+1.0,1) {};
  \node[onode] (o\t) at (\x+2.0,0) {$o_{\t}$};
  \node[anchor=south west, font=\tiny, inner sep=1pt] at (a\t.north east) {$\theta_{\t,1}$};
  \node[anchor=north east, font=\tiny, inner sep=1pt] at (b\t.south west) {$\theta_{\t,2}$};
  \draw[eg] (s\t) -- (a\t);
  \draw[eg] (a\t) -- (b\t);
  \draw[eg] (b\t) -- (o\t);
}
\draw[eg] (a1) -- (a2); \draw[eg] (a2) -- (a3); \draw[eg] (a3) -- (a4); \draw[eg] (a4) -- ++(1.2,0);
\draw[eg] (b1) to[bend left=18] (b3);
\draw[eg] (b2) to[bend right=18] (b4);
\draw[pa, transform canvas={xshift=-1.2pt}] (s1) -- (a1);
\draw[pa] (a1) -- (a2);
\draw[pa] (a2) -- (b2);
\draw[pa] (b2) -- (o2);
\draw[pb, transform canvas={xshift=1.2pt}] (s1) -- (a1);
\draw[pb] (a1) -- (b1);
\draw[pb] (b1) to[bend left=18] (b3);
\draw[pb, transform canvas={xshift=1.2pt}] (b3) -- (o3);
\draw[pc] (s2) -- (a2);
\draw[pc] (a2) -- (a3);
\draw[pc] (a3) -- (b3);
\draw[pc, transform canvas={xshift=-1.2pt}] (b3) -- (o3);
\node[align=left, anchor=north west, font=\scriptsize] at (-0.4,-0.5) {%
  horizontal arrows: loop 1 ($\tau_1=1$);\ \ arcs: loop 2 ($\tau_2=2$);\ \ vertical/diagonal arrows: time bin within a step;\ \ \textcolor{orange!85!black}{$\psi_t$}: source features\\[2pt]
  \textcolor{red!75!black}{$\pi$}: $\operatorname{amp}(\pi)=\sin\theta_{1,1}\,\sin\theta_{2,1}\,\cos\theta_{2,2}$ \qquad
  \textcolor{blue!70!black}{$\pi'$}: $\operatorname{amp}(\pi')=\cos\theta_{1,1}\,\sin\theta_{1,2}\,\sin\theta_{3,2}$ \qquad
  \textcolor{green!50!black}{$\pi''$}: $\operatorname{amp}(\pi'')=\sin\theta_{2,1}\,\sin\theta_{3,1}\,\cos\theta_{3,2}$};
\end{tikzpicture}
\caption{Time-unrolled graph of a single rail with two loops ($\tau_1=1$, $\tau_2=2$) over four steps; fixed gates (readout, phases, loss) are omitted.
Each square is a coupling beam splitter with angle $\theta_{t,l}=\theta^0_l+\kappa G_{t,c(l)}$, a gate site under gate encoding; each source $s_t$ is a source site with features $\psi_t$ under squeezing encoding. Within a step the time bin meets loop~1 before loop~2, so a path can switch from loop~1 to loop~2 within a step but not back.
Second order: two paths leave the same source $s_1$. This path pair's contribution to the covariance between $o_2$ and $o_3$ depends only on the drives of gates on $\pi\cup\pi'$ and on $\psi_1$; the full covariance sums all source-compatible path pairs (\eqref{eq:path-sum} and \cref{thm:zeros}).
The pair $(\theta_{1,1},\theta_{2,1})$ interacts along one path, $(\theta_{2,1},\theta_{3,2})$ across the two legs, and $\theta_{1,1}$ appears on both legs; these are the three mechanisms separated by the gate-level recursion in \cref{app:gate-recursion}.
Fourth order: $\langle n_{o_2}n_{o_3}\rangle$ also contains the product of a leg pair from $s_1$ along $\pi$ and a leg pair from $s_2$ along $\pi''$. This term carries $\psi_1\psi_2$, which no covariance entry contains (\cref{prop:gauss}), and $\theta_{2,1}$ lies on all four legs, so it can carry harmonics up to $4\kappa$ in that gate's drive, twice the covariance bound (\cref{thm:poly}).}
\label{fig:dag}
\end{figure*}

\begin{theorem}[All-order structural zeros]\label{thm:zeros}
Call a pair of paths $\pi,\pi'$ \emph{source-compatible} when their input modes are coupled by a nonzero entry of $N_{\rm in}$ or $M_{\rm in}$; the two modes may be distinct members of one TMSV pair.
For $\omega\ne0$, if $\hat V_\omega\neq0$, there exist a source-compatible pair ending in retained outputs and a selection of local gate harmonics on that path pair whose signed sum equals $\omega$ (with the conjugate sign on a $\bar T$ leg). The vacuum contribution to the constant coefficient is excluded from this statement.
In particular, for every $\alpha\in\supp\omega$, at least one modulated gate fed by $G_\alpha$ lies on $\pi\cup\pi'$.
Thus source-compatible path pairs give an outer support for every Fourier coefficient, not necessarily its exact nonzero support.
For a set $S$ of drive components, define its interaction part as
$\mathcal I_S V_\theta(G)=\sum_{\omega:S\subseteq\supp\omega}\hat V_\omega e^{i\omega\cdot G}$.
If no source-compatible path pair contains at least one gate fed by each component in $S$, then $\mathcal I_S V_\theta\equiv0$.
More generally (\cref{thm:poly}(c)), for order-$k$ moments and for source sites, a product of features of the gate sites in $S$ and of the source sites in $\Sigma$ vanishes identically unless a source-compatible $k$-tuple of paths passes every gate in $S$ and has a leg starting at every source in $\Sigma$.
\end{theorem}
\begin{proof}
By \eqref{eq:path-sum} and~\eqref{eq:NM}, each entry of $N$ or $M$ is a sum over sources and path pairs $(\pi,\pi')$ of products of edge weights,
and a term depends only on the drives of gates on $\pi\cup\pi'$.
Expanding those edge weights into their finite harmonics, each term has a frequency equal to the sum of its selected local harmonics.
Consequently a nonzero coefficient needs at least one source-compatible path pair that produces its frequency; cancellations between path pairs may still make an allowed coefficient zero.
\end{proof}
For one component, \cref{thm:zeros} says that a drive component is
\emph{readable} only if one of its gates reaches a retained output.
For two components it recovers the second-order structural zeros:
no common gate, ordered gate pair on one path, or compatible source
with the two gates on separate legs means no interaction. These are
necessary support conditions, not sufficient nonzero conditions.
A path accumulates delay only by circulating in loops, so the offsets between a modulated gate and a retained output that a path can realize lie in the numerical semigroup generated by the delays,
independently of $\kappa$. If $\gcd(\tau)=g>1$, only gate-to-output offsets that are multiples of $g$ are reachable (the readable lags are these offsets shifted by the readout ages), and a nonzero second-order interaction between components at times $t,t'$ requires $t\equiv t'\pmod g$.
Rails that share no gate decouple exactly. These statements are necessary conditions only. In the $U8$ design, about $20\%$ of the components that pass the number-theoretic test are structural zeros of \cref{thm:zeros}.

\subsection{Independent components and the receiver's event span}\label{sec:components}
The path rule constrains single moments. A stronger, all-order statement
follows when the circuit splits into parts that share no light. Connect a
source to a retained mode whenever some path leads from a mode of the
source to it. The connected components of this graph partition the
retained modes into $C_1,\dots,C_K$; the sites on paths into different
components are disjoint.

\begin{theorem}[Component factorization]\label{thm:components}
Let the input be a product over sources. Then the retained output state
is a product, $\rho_{\rm out}=\bigotimes_k\rho_k$, where $\rho_k$ depends
only on the encoding features of the sites on paths into $C_k$. Hence
every product observable $O=\bigotimes_kO_k$ has
$\langle O\rangle=\prod_k\operatorname{tr}(\rho_kO_k)$. In particular:
\begin{enumerate}
\item[(a)] for zero-mean states, every quadrature second moment between
modes of different components vanishes for all drive values;
\item[(b)] the probability of a threshold-click pattern, or a
photon-number moment, on a mode set $E$ is a product
$\prod_{k\in K(E)}f_{E,k}$ over the components $K(E)$ that $E$ meets, with
$f_{E,k}$ depending only on the sites of $C_k$.
\end{enumerate}
\end{theorem}

\begin{corollary}[Event span]\label{cor:span}
A monomial that combines sites of several components can occur in the
feature of an event $E$ only if $E$ meets all of these components. In
particular:
\begin{enumerate}
\item[(i)] a receiver whose events each lie in one component, such as
homodyne and heterodyne for zero-mean states, single-mode clicks or
single-mode photon numbers, has features that are sums of
component-local functions; every exclusion certificate built from such
additive dictionaries, including~\eqref{eq:analytic-floor}, holds for
it exactly, not only for affine receivers;
\item[(ii)] a receiver whose events meet at most $w$ components reads
interactions of order at most $w$ across components, and only between
components that some event joins;
\item[(iii)] every quadrature second moment between different components
is a drive-independent feature and can be removed before any simulation.
\end{enumerate}
\end{corollary}
The theorem is a statement about light that never meets, so it applies
to classical light as well. A cross-component coincidence probability is
exactly the product of the two marginal probabilities; what the
coincidence adds is a product feature that a restricted linear head
cannot form from single-mode features, not a quantum resource.
In a TDM circuit with $\gcd(\tau)=g$ and independent sources, the
residues of the output time modulo $g$ label components on each rail
that no inter-rail gate joins. At $U8$ this gives $16$ components of two
modes each (one $Q$ and one $M$ mode of an age and rail). By
\cref{cor:span}(iii), $356$ of the $452$ heterodyne features, namely every
selected moment between different modes, are drive independent; with
the data-free reference arm and a phase symmetry only $32$ vary. Of the
$116$ click events, $68$ join two components, and these alone create the
click receiver's additional interactions (\cref{sec:numerics-components}).

\subsection{Topology-conditioned task subspace and risk floor}\label{sec:topology-floor}
In this subsection the data enter only through gates, and the source
moments are drive independent.
Fix the gate types, drive assignments, gains, source-correlation pattern and
retained outputs of a topology $\mathcal T$, while allowing its
unmodulated angles and source squeezing to vary. Source-compatible path pairs and their local
harmonics define an outer frequency set $\mathcal F_{\mathcal T}$.
Let $\phi_{\mathcal T}^{\rm ex}(G)$ contain the corresponding real sines
and cosines, with one representative per $\pm\omega$ pair. Cancellations
may remove terms, but cannot create frequencies outside this dictionary;
its frequencies are dual to drive \emph{values}, not physical time.
For any real dictionary $\phi$, set
$\Sigma_\phi=\operatorname{Cov}(\phi)$ and
$C_\phi=\operatorname{Cov}(\phi,y)$, and define the optimistic oracle risk
\begin{equation}\label{eq:oracle-floor}
R_{\rm or}(\phi)=R_0-\frac1q
  \tr(C_\phi^\top\Sigma_\phi^\dagger C_\phi).
\end{equation}
It grants a noiseless, unpenalized linear head access to every dictionary
coordinate on the same empirical distribution of training windows.

\begin{corollary}[Topology-conditioned oracle floor]
\label{cor:topology-floor}
Under (C1)--(C6), for gate encoding with drive-independent source
moments, every affine-receiver circuit of topology
$\mathcal T$, at any shot budget $S$ and fixed $\lambda\ge0$, obeys
\begin{equation}\label{eq:exact-floor}
R_{\theta,b,S}(\lambda)\ge R_{\rm or}(\phi_{\mathcal T}^{\rm ex}).
\end{equation}
If the receiver has at most $d_b$ features, let $\mu_i$ be the
descending eigenvalues of
$\Sigma_\phi^{\dagger/2}C_\phi C_\phi^\top
 \Sigma_\phi^{\dagger/2}$ for
$\phi=\phi_{\mathcal T}^{\rm ex}$. Then the stronger bound
\begin{equation}\label{eq:oracle-rank}
R_{\theta,b,S}(\lambda)\ge
R_0-\frac1q\sum_{i=1}^{d_b}\mu_i
\end{equation}
also holds, with missing eigenvalues taken as zero.
\end{corollary}
\begin{proof}[Key step]
The exact affine features satisfy
$\tilde{\hat z}=B_\phi\tilde\phi+\xi$, with centered dictionary
$\tilde\phi$, conditionally unbiased measurement noise $\xi$, and
$\E\xi\xi^\top=\bar\Omega_b/S\succeq0$. Put
$L=B_\phi^\top W$. Completing the square gives
\begin{equation}\label{eq:oracle-gap}
\begin{aligned}
\mathcal R_\phi(W)=R_{\rm or}(\phi)
&+\frac1q\|\Sigma_\phi^{1/2}
 (L-\Sigma_\phi^\dagger C_\phi)\|_F^2\\
&+\frac1{qS}\tr(W^\top\bar\Omega_bW),
\end{aligned}
\end{equation}
so both omitted terms are nonnegative. The rank-$d_b$ form follows by
granting the receiver the best possible $d_b$-dimensional row space and
applying the variational principle. The singular case and projection
identity are proved in \cref{app:proof-topology-floor}.
\end{proof}
The proof uses only that the centered features are a fixed linear map
of a topology-determined dictionary. It therefore extends to
moment-polynomial receivers and to source encodings once
$\phi_{\mathcal T}^{\rm ex}$ is replaced by the outer monomial dictionary
that \cref{thm:poly}(c) and \cref{prop:gauss} allow (for squeezing
encoding with an affine receiver, the $(u_t,v_t)$ of the sources that
reach a retained output); we do not use these extensions here. For
threshold clicks, the order-$d$ approximant of \cref{prop:click} gives
$R\ge R_{\rm or}(\phi^{(d)})-\varepsilon_R$ with its dictionary
$\phi^{(d)}$ and its certified risk error $\varepsilon_R$.
The floor is rigorous but optimistic: it is exact as an inequality, not
tight.
A topology is excluded before evaluating optical amplitudes if its
applicable floor exceeds an incumbent risk on the same task and risk
scale. This is a one-sided certificate: a surviving design need not be
good, and a large exact dictionary can give a vacuous floor on finite
training data. The result alone does not establish a measured runtime
saving.

\paragraph{Analytic screening example: an entire delay family.}
Consider single-rail circuits $\mathcal T_d$ with one lossless loop of
integer delay $d$, independent TMSV pairs injected at each step, an
identity reference arm, and all output ages retained. The only
time-nonlocal element is the loop. Eight consecutive coupling gates,
at steps $2,\ldots,9$, encode $G_i=\gamma u_i$ through
$\theta_i=\pi/4+\delta u_i$, where $u_i\in\{-1,1\}$,
$0<\gamma<1$ and $0<\delta=\kappa\gamma<\pi/4$.
Let the $256$ training windows enumerate all sign vectors equally and
take the scalar target to be the normalized sum of \emph{all} pairwise
history interactions,
\begin{equation}\label{eq:analytic-target}
y(u)=\frac{1}{\sqrt{28}}\sum_{1\le i<j\le8}u_i u_j,
\qquad R_0=1.
\end{equation}

A path in $\mathcal T_d$ remains within one time-residue class modulo
$d$; source-compatible path pairs do too, because the TMSV pairs are
independent between steps. Set
$n_r=\#\{i\in\{1,\ldots,8\}:i\equiv r\pmod d\}$.
The $28$ pair monomials in \eqref{eq:analytic-target} are orthonormal
under the full-factorial training measure. Every cross-residue pair is
therefore orthogonal to every affine covariance feature. Even granting
noiseless access to \emph{all} within-residue pairs gives the
topology-only certificate
\begin{equation}\label{eq:analytic-floor}
R_{\mathcal T_d,b,S}(\lambda)\ge
1-\frac1{28}\sum_{r=0}^{d-1}\binom{n_r}{2}
\quad (b\text{ affine},\ S>0,\ \lambda\ge0).
\end{equation}
For $d=2,3,4,8$ its values are respectively
$4/7,3/4,6/7,1$; \cref{fig:analytic-witness}(b) plots the floor for
eight delays. The $d=4$ topology, for example, misses $24$ of the $28$
task directions (red in \cref{fig:analytic-witness}a) regardless of its
operating angles, squeezing, or shot budget. These are optimistic risk floors, not claims that the surviving
directions are actually readable.

The $d=1$ member provides a constructive incumbent. Inject the memory
mode of the TMSV pair created at step $1$ into the loop, keep it through all
eight modulated gates, and extract it at step $10$, with fixed
$\pi/4$ injection and extraction couplers, zero fixed phases, and a
transparent readout splitter. The cross-arm homodyne moment between
this output and the reference mode of the same pair has the exact form
\begin{equation}\label{eq:analytic-feature}
z(u)=A_0\,2^{-4}\prod_{i=1}^8(c-su_i),
\end{equation}
where $c=\cos\delta$, $s=\sin\delta$, and
$|A_0|=\sinh r\cosh r>0$ for squeezing $r>0$.
Walsh orthogonality gives
$\operatorname{Cov}(z,y)^2=A_0^2 2^{-8}28c^{12}s^4$ and
$\operatorname{Var}(z)=A_0^2 2^{-8}(1-c^{16})$.
With this single feature and $\lambda=0$, its exact finite-shot risk is
\begin{equation}\label{eq:analytic-gap}
R_{\mathcal T_1,\mathrm{hom},S}(0)
=1-\frac{28c^{12}s^4}
 {1-c^{16}+2^8\bar\Omega_z/(A_0^2S)},
\end{equation}
where $0<\bar\Omega_z<\infty$ is the averaged single-shot variance.
At $\delta=\pi/6$ the infinite-shot limit is
$38563/58975\approx0.654$, the dashed line in
\cref{fig:analytic-witness}(b). Hence, at sufficiently large \emph{finite}
$S$, this incumbent lies below the $3/4$ floor and excludes all
$d\in\{3,\ldots,8\}$ before their optical path amplitudes are
evaluated; the $d=2$ floor does not settle that case.
This is a training-window, affine-receiver certificate for the stated
task, not a universal preference for short delays or a quantum
advantage claim.

The residue classes are the components of \cref{thm:components}, so the
certificate is not specific to affine receivers. By \cref{cor:span}(i)
it holds exactly for single-mode clicks, single-mode photon numbers and
coincidences within one class, and it fails precisely for events that
join two classes. The receiver, not the optics, then decides whether the
excluded directions are reachable: at $d=8$, where the floor is $1$, the
same optical state read by two-mode click coincidences reaches the
oracle risk $0$ (\cref{sec:numerics-components}). Conversely, by
\cref{cor:class}(ii) and (i), squeezing encoding with an affine receiver,
or displacement encoding read through quadrature means, has risk
exactly $R_0=1$ on this task for every delay, including $d=1$: its
features are additive over steps, and every pair character is orthogonal
to additive functions.

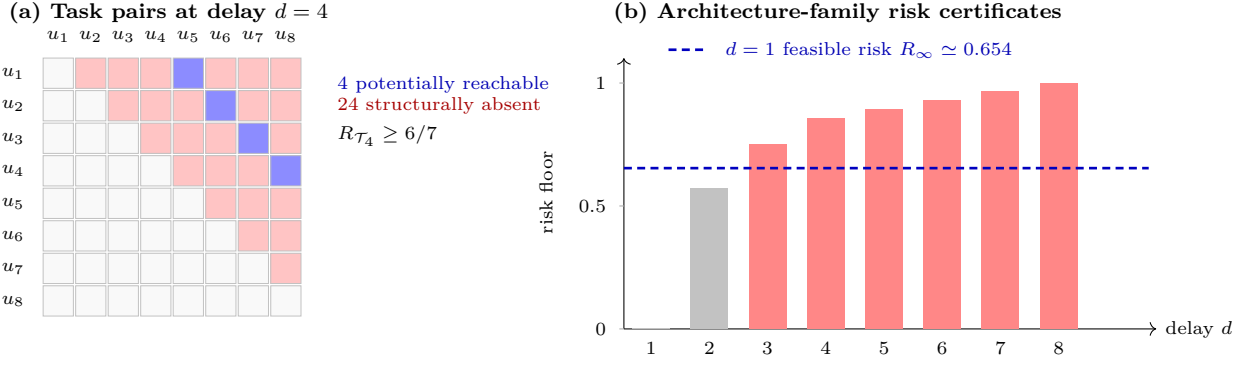
\begin{figure*}[t]
\centering
\begin{tikzpicture}[x=1cm,y=1cm,font=\scriptsize]
  \node[anchor=west,font=\footnotesize\bfseries] at (0.2,4.65)
    {(a) Task pairs at delay $d=4$};
  \foreach \i in {1,...,8} {
    \pgfmathsetmacro{\yy}{4.05-0.43*(\i-1)}
    \node[anchor=east] at (0.66,\yy-0.20) {$u_{\i}$};
    \pgfmathsetmacro{\xx}{0.78+0.43*(\i-1)}
    \node[anchor=south] at (\xx+0.20,4.14) {$u_{\i}$};
    \foreach \j in {1,...,8} {
      \pgfmathsetmacro{\xcell}{0.78+0.43*(\j-1)}
      \ifnum\j>\i
        \pgfmathtruncatemacro{\rem}{mod(\j-\i,4)}
        \ifnum\rem=0
          \fill[blue!45] (\xcell,\yy-0.40) rectangle ++(0.40,0.40);
        \else
          \fill[red!23] (\xcell,\yy-0.40) rectangle ++(0.40,0.40);
        \fi
      \else
        \fill[gray!5] (\xcell,\yy-0.40) rectangle ++(0.40,0.40);
      \fi
      \draw[gray!48,line width=0.25pt]
        (\xcell,\yy-0.40) rectangle ++(0.40,0.40);
    }
  }
  \node[anchor=west,align=left] at (4.55,3.35)
    {\textcolor{blue!70!black}{4 potentially reachable}\\
     \textcolor{red!65!black}{24 structurally absent}\\[3pt]
     $R_{\mathcal T_4}\ge 6/7$};
  \node[anchor=west,font=\footnotesize\bfseries] at (8.20,4.65)
    {(b) Architecture-family risk certificates};
  \draw[->] (8.46,0.47) -- (15.50,0.47) node[right] {delay $d$};
  \draw[->] (8.46,0.47) -- (8.46,4.05);
  \node[rotate=90] at (7.43,2.28) {risk floor};
  \foreach \val/\lab in {0/0,0.5/{0.5},1/1} {
    \pgfmathsetmacro{\ytick}{0.47+3.25*\val}
    \draw[gray!65] (8.40,\ytick) -- (8.46,\ytick);
    \node[anchor=east] at (8.35,\ytick) {$\lab$};
  }
  \foreach \d/\v in {1/0,2/0.57143,3/0.75000,4/0.85714,5/0.89286,6/0.92857,7/0.96429,8/1.00000} {
    \pgfmathsetmacro{\bx}{8.57+0.77*(\d-1)}
    \pgfmathsetmacro{\bh}{3.25*\v}
    \ifnum\d<3
      \fill[gray!48] (\bx,0.47) rectangle ++(0.50,\bh);
    \else
      \fill[red!47] (\bx,0.47) rectangle ++(0.50,\bh);
    \fi
    \node[anchor=north] at (\bx+0.25,0.42) {$\d$};
  }
  \pgfmathsetmacro{\inc}{0.47+3.25*0.653887}
  \draw[blue!75!black,densely dashed,line width=1pt]
    (8.46,\inc) -- (15.40,\inc);
  \draw[blue!75!black,densely dashed,line width=1pt]
    (9.05,4.16) -- (9.55,4.16);
  \node[anchor=west,text=blue!70!black]
    at (9.68,4.16) {$d=1$ feasible risk $R_{\infty}\simeq0.654$};
\end{tikzpicture}
\caption{An analytic, topology-only screening certificate for the
all-pairs task in \eqref{eq:analytic-target}. (a) At delay $4$, only four
same-residue drive pairs remain potentially reachable (blue); the
other $24$ pairs (red) are structural zeros for affine covariance
readout. (b) The optimistic floors from \eqref{eq:analytic-floor} for
eight one-loop delays. The dashed line is the risk of one explicit
$d=1$ homodyne feature at $S\to\infty$; for sufficiently large finite
$S$, every $d\ge3$ floor exceeds that feasible risk. The result is
task- and receiver-specific; it does not claim held-out performance or
quantum advantage.}
\label{fig:analytic-witness}
\end{figure*}

\section{The computable second-order truncation}\label{sec:truncation}
For gate encoding the exact dictionary can be exponentially large (for squeezing encoding with an affine receiver no truncation is needed, \cref{cor:squeeze}), and so can the click approximants of \cref{prop:click}. A local expansion in the drive is cheaper. Center the
drive at its training mean, $G=\gbar+x$, and truncate the exact state
identity along each window:
\begin{equation}\label{eq:taylor}
\begin{aligned}
T_k(x)
&=\sum_\omega\hat V_\omega e^{i\omega\cdot\gbar}
  \sum_{j\le k}\frac{(i\omega\cdot x)^j}{j!}\\
&=\sum_{j\le k}\frac1{j!}
  D^jV_\theta(\gbar)[x,\dots,x].
\end{aligned}
\end{equation}
The second-order member is the lowest non-affine one. In this
zero-mean Gaussian setting $\mathcal V_\theta=V_\theta$, and applying a
receiver $\Phi_b$ gives a single composite expansion,
\begin{equation}\label{eq:second}
\begin{aligned}
z_b(\gbar+x)
={}&z_0+H_bK_\theta x
 +\tfrac12 H_bD^2V_\theta[x,x]\\
&+\tfrac12 D^2\Phi_b[K_\theta x,K_\theta x]+r(x),
\end{aligned}
\end{equation}
where $z_0=\Phi_b(\Vbar)$, $\Vbar=V_\theta(\gbar)$,
$K_\theta=DV_\theta(\gbar)$, and $H_b=D\Phi_b(\Vbar)$.
The last displayed quadratic term vanishes for homodyne and heterodyne
but can add directions for threshold clicks.

To specify the sparse task dictionary once, let $\mathcal R_{\mathcal T}$
be the drive components readable by a source-compatible path pair and
$\mathcal P_{\mathcal T}^{\rm state}$ its quadratic path-pair outer
support. Define
\begin{equation}\label{eq:second-support}
\begin{aligned}
\mathcal P_{\mathcal T,b}&=\mathcal P_{\mathcal T}^{\rm state}\cup\mathcal P_b^{\rm rec},\\
\mathcal P_b^{\rm rec}&=\bigcup_{E\in\mathcal E_b}
 \bigl\{\{\alpha,\beta\}:\alpha,\beta\in\mathcal R(E)\bigr\},
\end{aligned}
\end{equation}
with $\mathcal P_b^{\rm rec}=\emptyset$ for affine receivers. Here
$\mathcal E_b$ is the event set of a nonlinear receiver and
$\mathcal R(E)$ the drive components readable at the modes of $E$.
Receiver curvature can multiply two directions that the state-layer
paths do not join, but only if one event reads both; across components
this requires an event that meets both (\cref{cor:span}). Let
$u_{\alpha\beta}=w_{\alpha\beta}
(x_\alpha x_\beta-\E[x_\alpha x_\beta])$, with
$w_{\alpha\alpha}=1$, $w_{\alpha\beta}=2$ for $\alpha\ne\beta$,
and $\psi_b=[x;u_{\mathcal P_{\mathcal T,b}}]$ after removing unreadable
linear coordinates. The centered expansion has the lifted form
\begin{equation}\label{eq:lifted}
\begin{aligned}
z_b-\E z_b&=B_{\theta,b}\psi_b+r_c,\\
\Sigma_z&=B\Sigma_{\psi_b}B^\top+\Delta_\Sigma,\\
C_z&=BC_{\psi_b}+\Delta_C,
\end{aligned}
\end{equation}
where $r_c=r-\E r$ and the columns of $B$ are
$H_bK_\theta e_\alpha$ and
$\tfrac12(H_bD^2V_\theta[e_\alpha,e_\beta]
+D^2\Phi_b[K_\theta e_\alpha,K_\theta e_\beta])$ on the selected
coordinates. Thus this surrogate requires task moments only through
order four on the stated support and label cross-moments through order
two; $\operatorname{Cov}(G,y)$ alone does not determine it.

Let $\phi_{\mathcal T,b}^{(2)}$ contain these readable linear and
quadratic drive coordinates. Repeating the argument of
\eqref{eq:oracle-gap} after setting $r_c=0$ gives the \emph{surrogate}
bound
\begin{equation}\label{eq:second-floor}
R_{\theta,b,S}^{(2)}(\lambda)
\ge R_{\rm or}(\phi_{\mathcal T,b}^{(2)}).
\end{equation}
It assumes conditionally unbiased, positive-semidefinite shot noise.
For the full nonlinear click receiver, it implies a floor only if an
independent certificate gives
$|R_{\theta,b,S}-R_{\theta,b,S}^{(2)}|\le\varepsilon_R$;
the resulting floor is
$R_{\rm or}(\phi_{\mathcal T,b}^{(2)})-\varepsilon_R$.
Within that validity regime the noiseless second-order map is a
topology-sparse quadratic NVAR on the \emph{same encoded drives},
not a claim of finite-sample dominance over an NVAR baseline~\cite{gauthier2021next}.

Directional derivatives of the time-ordered transfer matrix compute
$K_\theta$ and $D^2V_\theta$ without forming a full Hessian. A gate-level
recursion separates single-gate curvature, ordered pairs of gates on
one path, and pairs on the two legs of a source; it costs about five
circuit propagations per direction pair. The update equations are in
\cref{app:gate-recursion}. Which is cheaper, one direction per window
or one per supported pair, depends on the window count and support size.

\subsection{Validity of the approximation}\label{sec:validity}
For every $x$ and order $k$, the scalar Taylor remainder of each
harmonic in \eqref{eq:trig} gives
\begin{equation}\label{eq:remainder}
\|V_\theta(\gbar+x)-T_k(x)\|
\le\sum_\omega\|\hat V_\omega\|
 \frac{|\omega\cdot x|^{k+1}}{(k+1)!}.
\end{equation}
The series converges for every $x$, but a low-order truncation is useful
only when the phase $|\omega\cdot x|$ accumulated along relevant paths
is small. Large accumulated EOM phases can make early orders worse
before they improve (\cref{sec:numerics-limits}).

A state error need not induce the same prediction-risk error. If the
centered feature remainder has
$\rho^2=\E\|r_c\|^2$, then
$\|\Delta_C\|_F\le\rho\sqrt{qR_0}$ and
$\|\Delta_\Sigma\|\le2\rho\sqrt{\E\|z^{(k)}\|^2}+\rho^2$.
Together with a bound on the shot-covariance difference, the
spectral-stability condition
$\|A-A_k\|<\lambda_{\min}(A_k)+\lambda$
gives an explicit, though generally loose, risk-remainder certificate
in \cref{app:proof-riskerror}. We report state and risk errors
separately rather than identifying either one with the other.

At the weather $U8$ operating point, the click curvature term in
\eqref{eq:second} resolves $997$ drive pairs above a relative
$10^{-6}$ threshold. \Cref{fig:support} shows where they lie among all
$41{,}616$ drive pairs: $238$ lie in the state-layer structural outer
support and $759$ are added by the receiver. The
event-span support~\eqref{eq:second-support} predicts exactly these
$997$ pairs, compared with $2346$ for the union of all readable pairs:
all $759$ additions join two components through one of the $68$
cross-component coincidence events, and no curvature appears between
components that no event joins (\cref{sec:numerics-components}).
Classical quadratic post-processing can generate the same products, so
the difference is structural, not a quantum-advantage claim.

\begin{figure*}[t]
\centering
\begin{tikzpicture}[x=1cm,y=1cm,>=Stealth,font=\small]
\draw[rounded corners=3pt,fill=gray!5,draw=black!60] (0,0) rectangle (14.8,3.55);
\node[anchor=west,font=\small\bfseries] at (.25,3.24)
  {All unordered drive pairs: $41{,}616$};
\draw[rounded corners=3pt,fill=blue!5,draw=blue!55!black] (.52,.30) rectangle (14.28,2.91);
\node[anchor=west,font=\small] at (.80,2.60)
  {Pairs of first-order readable components: $2{,}346$};
\draw[rounded corners=2pt,fill=blue!19,draw=blue!70!black] (1.0,1.02) rectangle (6.62,2.17);
\node[align=center,font=\small] at (3.81,1.595)
  {$238$ in the TDM path-pair\\state-layer outer support};
\draw[rounded corners=2pt,fill=orange!19,draw=orange!75!black] (7.08,1.02) rectangle (13.82,2.17);
\node[align=center,font=\small] at (10.45,1.595)
  {$759$ cross-component pairs added by\\the $68$ joining click coincidences};
\node[font=\scriptsize] at (7.4,.64)
  {$997=238+759$ nonzero pairs at weather $U8$, equal to the event-span prediction~\eqref{eq:second-support}};
\end{tikzpicture}
\caption{Second-order sparsity and the distinction between state and receiver mechanisms. The nesting and disjoint additions are schematic, not area-proportional. The $238$ state pairs form a structural \emph{outer} support; cancellations can remove terms. The click numbers count receiver-curvature directions resolved above a relative $10^{-6}$ threshold at one operating point; the event-span support predicts them with no miss and no false alarm. Affine homodyne/heterodyne receivers have no curvature addition.}
\label{fig:support}
\end{figure*}
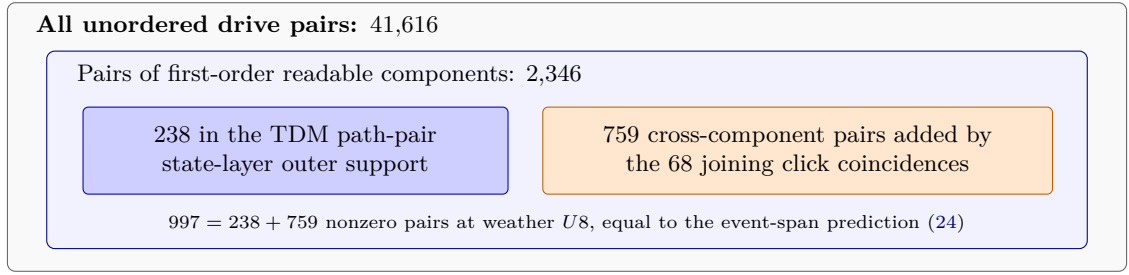

\section{Numerical verification}\label{sec:numerics}

\subsection{Protocol}\label{sec:numerics-setup}
We use seven multivariate series: weather, electricity, traffic, solar, ETTh1, PM2.5 and exchange~\cite{wu2021autoformer,lai2018modeling,zhou2021informer,liang2015pm25}, split $70/10/20$ in time~\cite{wu2021autoformer}.
Standardization, principal directions, the baseline weights below and every readout are fitted on training rows only. The window length is $L=96$ and the horizon $H=96$.
The target is the residual of a per-variable linear predictor on the variable's own history (NLinear-type~\cite{zeng2023transformers}) fitted on the training segment, so $q=H\cdot D$.
The reservoir therefore supplies corrections to a linear baseline. There are $n=512$ training windows.
Exchange is encoded by first differences. Every risk in this section is computed on training windows. Two inputs come from an earlier development stage of the project that used the validation rows: the ridge penalty of the baseline, selected there once and held fixed, and the choice of the operating point $U8$ below. They affect which target and which circuit are studied, not the identities being checked.

The operating point $U8$ has $\tau=(12,24,48)$, $R=2$, $P=3$, coupling encoding with $\theta^0=0.45$ and $\kappa=0.3$, TMSV squeezing $r=0.5$,
readout angle $0.25$, identity reference arm, and $W=8$ equally spaced readout ages $(95,81,68,54,41,27,14,0)$, giving $m=32$ retained modes.
The receivers use $S=2\times10^4$ shots and $\eta_{\det}=0.85$.
Heterodyne uses $452$ selected quadrature second moments; three
homodyne local-oscillator settings resolve $363$ combinations and six
resolve all $452$, with shots split among settings. The threshold-click
receiver uses $116$ events: all $32$ single-mode clicks and $84$
selected two-mode coincidences. Their full noise covariances, rather
than diagonal approximations, enter the risk.
The replay used for all states is built on DeepQuantum~\cite{he2025deepquantum} and agrees with an independent NumPy transfer-matrix oracle to $1.8\times10^{-16}$.
The gate-level derivatives of \cref{app:gate-recursion} agree with Richardson finite differences to $4\times10^{-9}$--$2\times10^{-8}$.
The noise covariances $\Omega_b$ were checked against raw quadrature samples, and every state was checked for physicality.

\paragraph{Exact statements.}\label{sec:numerics-exact}
Every exact statement of \cref{sec:bridge,sec:structure} was checked numerically before use; \cref{app:checks} gives the tables. For gate encoding at $U8$ and at F1-29, discrete Fourier transforms over a full period reconstruct $V_\theta$ at random off-grid points to $4\times10^{-13}$, with out-of-band energy below $2\times10^{-26}$, and the characteristic-function form~\eqref{eq:spectral} reproduces a directly computed risk to $2.4\times10^{-15}$ (\cref{tab:exact}). On small synthetic circuits the source-encoding statements hold to $10^{-14}$, and the path rule predicts exactly which gate--source and source--source products appear: $36$ of $100$ and $239$ of $1620$ candidates, with no miss and no false alarm. The Fock-state bounds hold and are attained for coupling encoding, and the drive-noise damping agrees with direct quadrature to $8\times10^{-12}$ (\cref{tab:general}).

\subsection{Squeezing encoding on the seven tasks}\label{sec:numerics-squeeze}
To use \cref{cor:squeeze,cor:task} on real data, we keep the $U8$ geometry with every gate fixed at its operating angle and move the data into the sources: the TMSV pair injected at step $t$ on rail $\rho$ has squeezing $r_{t,\rho}=0.5+0.3\,G_{t,\rho}$, driven by the first two principal channels. For the heterodyne receiver we compute the exact training-window risk in three ways. The first propagates every window. The second uses the response matrices $A_t,B_t$, obtained from $385$ circuit evaluations, together with the empirical mean and covariance of $(u_t,v_t)$ and their covariance with $y$. The third uses the same response matrices and nothing from the task except $6385$ values of $\mathcal L$ and $113\times q$ values of $\mathcal L_y$ at real arguments~\eqref{eq:mgf}. \Cref{tab:squeeze} lists the resulting risks: all three routes agree to within $10^{-14}$ in the captured variance $R_0-R$ on every dataset.

The path rule fixes the dictionary before any data are used. Of the $192$ sources, $56$ reach a retained output, as the delay semigroup predicts: $36$ on the rail with $\tau=(12,48)$ and $20$ on the rail with $\tau=24$. The anomalous feature $v_t$ needs both modes of a pair to arrive, which with an identity reference arm happens only for the $16$ sources injected at readout steps. The dictionary therefore has $72$ coordinates, and its oracle floor bounds every affine-receiver circuit with this topology and encoding. These are identity checks and certificates on training windows, not a comparison between encodings.

\begin{table}[tb]
\centering\footnotesize
\caption{Squeezing-encoded $U8$ with heterodyne receiver ($\lambda=1$, $S=2\times10^4$). $R_{\rm or}$ is the oracle floor of the $72$-coordinate dictionary. The last two columns give $10^{15}\times$ the maximal relative difference of $R_0-R$ over $\lambda\in\{0.01,1,100\}$ from direct propagation, for the moment route and the $\mathcal L$ route.}
\label{tab:squeeze}
\setlength{\tabcolsep}{3pt}
\begin{tabular*}{\columnwidth}{@{\extracolsep{\fill}}lccccc@{}}
\toprule
Dataset & $R_0$ & $R$ & $R_{\rm or}$ & moments & $\mathcal L$\\ \midrule
ETTh1       & $0.3774$ & $0.3566$ & $0.3048$ & $0$   & $2.7$\\
exchange    & $0.1206$ & $0.1154$ & $0.0984$ & $2.1$ & $2.1$\\
weather     & $0.5219$ & $0.4852$ & $0.3967$ & $2.2$ & $8.0$\\
PM2.5       & $0.5513$ & $0.4884$ & $0.4109$ & $0.9$ & $1.8$\\
solar       & $0.2936$ & $0.2537$ & $0.2006$ & $1.3$ & $0.6$\\
electricity & $0.1986$ & $0.1934$ & $0.1630$ & $3.1$ & $9.3$\\
traffic     & $0.4097$ & $0.3744$ & $0.2752$ & $2.2$ & $1.1$\\ \bottomrule
\end{tabular*}
\end{table}

\subsection{Components decide what a receiver can read}\label{sec:numerics-components}
\Cref{tab:components} tests \cref{thm:components,cor:span,cor:class} on the delay family of \cref{sec:topology-floor}, with the same $256$ windows and all-pairs target; each entry is the oracle risk of the receiver's feature span, a lower bound on its risk at every shot budget and ridge parameter. The output state factorizes exactly over residue classes: cross-class moments $N_{ij}$ and $M_{ij}$ vanish, and cross-class coincidence probabilities and photon-number correlations equal products of marginals to $2\times10^{-16}$. Component-local receivers never cross the floor~\eqref{eq:analytic-floor}, and for $d=3,4,8$ the second moments and the within-class coincidences attain it, so the certificate is tight there. Events that join classes cross it: at $d=8$, where no component-local receiver captures anything, two-mode coincidences and photon-number products reach $R=0$. Squeezing encoding read by all second moments, and displacement encoding read by quadrature means, give $R=1$ to $10^{-12}$ for $d=4$ and for $d=1$, as \cref{cor:class} requires.

At the weather $U8$ operating point the same structure explains the click receiver. The $32$ retained modes form $16$ components, and each of the $68$ readable drive components lies in exactly one of them. The event-span support~\eqref{eq:second-support} predicts $997$ curvature pairs; the numerically nonzero set is exactly these $997$, with no miss and no false alarm. All $759$ pairs beyond the state-layer support join two components through one of the $68$ cross-component coincidence events, and the $1349$ cross-component pairs that no event joins have curvature exactly zero.

\begin{table}[tb]
\centering\scriptsize
\caption{Oracle risk of each receiver's feature span on the all-pairs task (delay family, gate encoding, $256$ windows); $F_d$ is the floor~\eqref{eq:analytic-floor}. Component-local receivers: all second moments (2nd), single-mode clicks (1-clk), single-mode clicks with within-class coincidences (in). Joining receivers: all two-mode coincidences (all), photon-number first and second moments (PN).}
\label{tab:components}
\setlength{\tabcolsep}{2.5pt}
\begin{tabular*}{\columnwidth}{@{\extracolsep{\fill}}rcccccc@{}}
\toprule
$d$ & $F_d$ & 2nd & 1-clk & in & all & PN\\ \midrule
$2$ & $0.571$ & $0.581$ & $0.808$ & $0.579$ & $0.334$ & $0.436$\\
$3$ & $0.750$ & $0.750$ & $0.862$ & $0.750$ & $0.358$ & $0.408$\\
$4$ & $0.857$ & $0.857$ & $0.917$ & $0.857$ & $0.262$ & $0.328$\\
$8$ & $1.000$ & $1.000$ & $1.000$ & $1.000$ & $0.000$ & $0.000$\\ \bottomrule
\end{tabular*}
\end{table}

\subsection{Shot budgets}\label{sec:numerics-shots}
\Cref{prop:shots} agrees with direct evaluation of the risk at $S=10^2,\dots,10^6$ on all seven datasets and four receivers at $U8$, to $7\times10^{-13}$ in the captured variance ($\lambda=0$). \Cref{tab:shots} reads the spectra. The homodyne-type receivers have small function classes that are nearly used up at the operating budget: at $S=2\times10^4$ they realize $66$--$91\%$ of their attainable risk reduction, $90\%$ needs $2\times10^4$--$1.2\times10^5$ shots, and $12$--$16$ of their $32$ varying directions are resolved. The click receiver's larger class, which contains the cross-component products of \cref{cor:span}, lowers $R_\infty$ by $0.015$--$0.09$ on every dataset, but its additional directions have low single-shot signal-to-noise ratio. At $S=2\times10^4$ it realizes only $45$--$67\%$ of its reduction, with $28$--$29$ of $66$ directions resolved, and $90\%$ would need $1.3\times10^6$--$1.3\times10^7$ shots, $30$--$100$ times more. The trade-off between function class and shot cost is read off the spectrum $\{(g_i,\omega_i)\}$ without sampling.

\begin{table*}[t]
\centering\footnotesize
\caption{Exact shot law at $U8$ ($\lambda=0$). For each receiver: infinite-shot risk $R_\infty$, fraction of the attainable reduction $R_0-R_\infty$ realized at $S=2\times10^4$, and the shots $S_{90}$ needed for $90\%$ of it. Homodyne has the same $R_\infty$ as heterodyne; with three settings its $S_{90}$ is within $2\%$ of heterodyne, and with six settings about $20\%$ lower, because it avoids the extra vacuum noise of joint quadrature measurement.}
\label{tab:shots}
\begin{tabular}{@{}lccccccc@{}}
\toprule
 & & \multicolumn{3}{c}{heterodyne} & \multicolumn{3}{c}{threshold clicks}\\ \cmidrule(lr){3-5}\cmidrule(l){6-8}
Dataset & $R_0$ & $R_\infty$ & at $2\times10^4$ & $S_{90}$ & $R_\infty$ & at $2\times10^4$ & $S_{90}$\\ \midrule
ETTh1       & $0.377$ & $0.353$ & $77\%$ & $5.8\times10^{4}$ & $0.297$ & $53\%$ & $5.7\times10^{6}$\\
exchange    & $0.121$ & $0.116$ & $76\%$ & $6.2\times10^{4}$ & $0.101$ & $45\%$ & $5.2\times10^{6}$\\
weather     & $0.522$ & $0.473$ & $77\%$ & $6.6\times10^{4}$ & $0.393$ & $58\%$ & $5.1\times10^{6}$\\
PM2.5       & $0.551$ & $0.484$ & $89\%$ & $2.3\times10^{4}$ & $0.401$ & $66\%$ & $1.5\times10^{6}$\\
solar       & $0.294$ & $0.257$ & $82\%$ & $4.2\times10^{4}$ & $0.205$ & $67\%$ & $1.3\times10^{6}$\\
electricity & $0.199$ & $0.182$ & $74\%$ & $7.3\times10^{4}$ & $0.152$ & $52\%$ & $7.7\times10^{6}$\\
traffic     & $0.410$ & $0.370$ & $66\%$ & $1.2\times10^{5}$ & $0.280$ & $53\%$ & $1.3\times10^{7}$\\ \bottomrule
\end{tabular}
\end{table*}

\subsection{Approximations}\label{sec:numerics-hierarchy}
\paragraph{Taylor hierarchy.}
\Cref{fig:truncation}(a), and \cref{tab:hierarchy} in \cref{app:checks}, show the Taylor truncations~\eqref{eq:taylor} of the state at $U8$ on five datasets, computed along each window's path with fourth-order finite-difference stencils and Richardson extrapolation.
The state error falls from about $27\%$ at first order to $2.6$--$4.5\%$ at second order, $0.8$--$1.2\%$ at third and $0.15$--$0.45\%$ at fourth.
At the feature level the error decreases monotonically for every receiver. At the risk level it occasionally does not, because errors of opposite sign cancel, but it is at most $0.05\%$ at fourth order everywhere.
The residual risk error of the second-order truncation on the traffic task with homodyne-type receivers ($0.26$--$0.28\%$) falls to $0.14\%$ at third order and to zero at fourth order, so it is truncation error.
Truncating in interaction order instead (\cref{tab:interaction}) places exact single components plus exact structurally allowed pairs between the second- and third-order Taylor truncations: about half of the second-order residual comes from higher powers of one or two components, and half from interactions among three or more, which arise because photons pass modulated gates repeatedly. At this operating point ($\kappa|G|\le0.3$) truncating in total order is more economical.

\begin{figure*}[t]
\centering
\includegraphics[width=.92\linewidth]{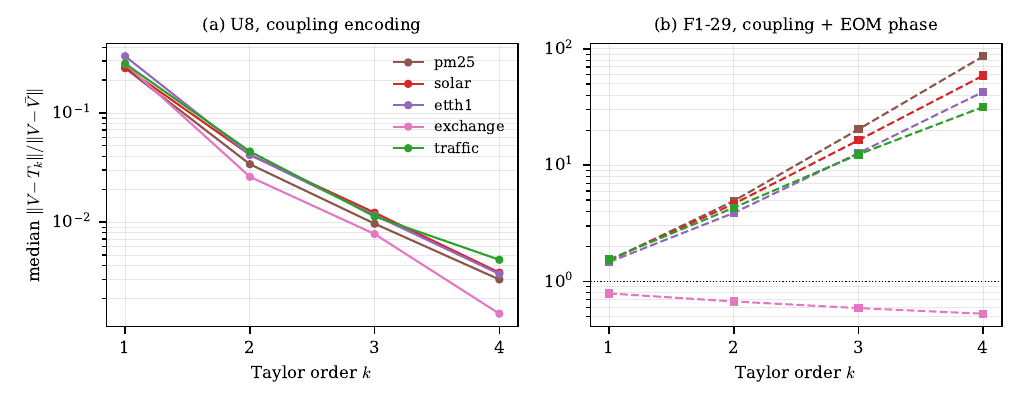}
\caption{Median relative state error of the order-$k$ Taylor truncation. (a) $U8$ (coupling encoding): geometric convergence on five datasets.
(b) F1-29 ($\tau=(1,2,4,8,16,32)$, coupling and EOM phase): accumulated phase puts low orders in the pre-asymptotic regime.}
\label{fig:truncation}
\end{figure*}

\paragraph{Risk accuracy.}\label{sec:numerics-risk}
\Cref{tab:risk} compares the first- and second-order truncations (state second order, receiver linear) with the exact risk at $U8$.
The $\lambda$ used for each receiver is selected on a training fit/hold split and is common to all approximations.
The first-order error exceeds $1\%$ in eight of the $28$ cells. The second-order error is at most $0.37\%$ in every cell.

\begin{table}[tb]
\centering\scriptsize
\caption{Relative risk error (\%) at $U8$: first $\to$ second order. The latter uses a second-order state and linear receiver.}
\label{tab:risk}
\setlength{\tabcolsep}{1.5pt}
\begin{tabular*}{\columnwidth}{@{\extracolsep{\fill}}lcccc@{}}
\toprule
Dataset & click & het. & hom. (3) & hom. (6)\\ \midrule
weather     & $0.22\to0.03$ & $0.28\to0.05$ & $0.28\to0.05$ & $0.28\to0.07$\\
electricity & $0.84\to0.05$ & $1.12\to0.11$ & $1.12\to0.11$ & $1.13\to0.10$\\
traffic     & $1.19\to0.10$ & $0.34\to0.28$ & $0.33\to0.28$ & $0.38\to0.26$\\
solar       & $1.47\to0.20$ & $1.13\to0.20$ & $1.13\to0.20$ & $1.15\to0.18$\\
ETTh1       & $0.28\to0.00$ & $0.02\to0.16$ & $0.02\to0.16$ & $0.01\to0.14$\\
PM2.5       & $0.55\to0.02$ & $0.82\to0.37$ & $0.82\to0.37$ & $0.82\to0.36$\\
exchange    & $0.01\to0.01$ & $0.00\to0.00$ & $0.00\to0.00$ & $0.00\to0.00$\\ \bottomrule
\end{tabular*}
\end{table}

\paragraph{Click receivers from finite dictionaries.}\label{sec:numerics-click}
We tested \cref{prop:click} on the same $116$-event click receiver at $U8$ on all seven datasets (\cref{tab:click}). Over all training windows and all mode sets of up to four modes, the spectra of $Q_J$ lie in $[0.734,1.724]$, inside the predicted interval $[0.731,1.730]$. The error of $P_{J,d}$ stays within its a-priori bound and falls by about a factor of five per degree, and so, less regularly, does the risk error of the order-$d$ approximant: it is at most $0.045\%$ at $d=4$, comparable to the second-order errors of \cref{tab:risk}, and $3\times10^{-5}$ at $d=6$. The risk certificate built from the a-priori errors alone is informative from $d=6$ at $\lambda=100$ and from $d=10$ at $\lambda=1$; with the actual size of the remainder it is informative from $d=2$ and $d=6$, respectively. At $\lambda=0.01$ it becomes informative only at $d=12$, even a posteriori, because $A_k$ is numerically singular and the stability condition reduces to $\|A-A_k\|<\lambda$. The limitation lies in the resolvent bound, not in the approximation. The approximant's features have frequencies in $\mathcal W^{(8d)}$, so this is a convergence statement, not a cheaper evaluator.

\paragraph{Within-family ranking.}\label{sec:numerics-family}
To test whether the second-order truncation tracks the exact risk \emph{across architectures}, we use two families.
F1 has $32$ geometries (delays, projection rank, encoding and readout ages) drawn by stratified sampling from a $256$-cell grid; F2 has $7$ sizes ($4$ to $16$ equally spaced readout ages at the $U8$ geometry).
$\lambda$ is fixed per dataset. We define the validity subfamily as the F1 designs whose median second-order state error is at most $0.2$ and smaller than the first-order error; this secondary analysis was added to the protocol after a smoke run and before the family was evaluated.
\Cref{fig:family} and \cref{tab:family} compare truncated and exact training risks. Inside the validity subfamily the Spearman correlation is $0.95$--$0.99$ and the maximum error is $0.29$--$1.26\%$.
Outside it, the errors are larger (median $0.63\%$ versus $0.12\%$ on solar). On the size family the ranking is exact on all four datasets.
These are statements about the training-segment risk, not about selection on held-out data.

\begin{figure*}[t]
\centering
\includegraphics[width=.9\linewidth]{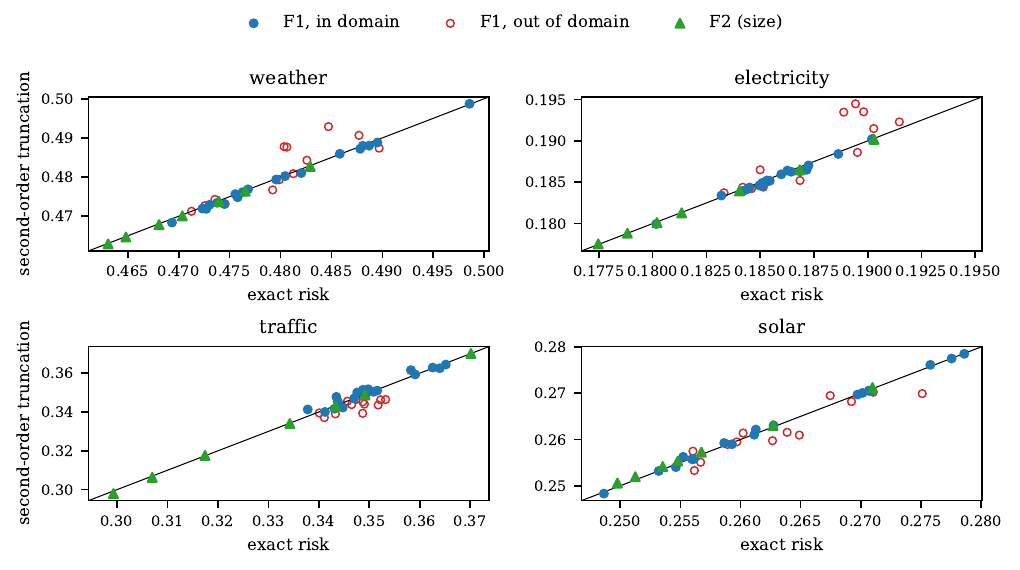}
\caption{Second-order truncation versus exact training risk for the geometry family F1 (filled: validity subfamily; open: outside) and the size family F2.}
\label{fig:family}
\end{figure*}

\begin{table}[tb]
\centering\scriptsize
\caption{Second-order truncation versus exact training risk within candidate families (click receiver, fixed $\lambda$). Entries are Spearman $\rho$ / maximal relative error; the validity subfamily has $19$ designs on every dataset.}
\label{tab:family}
\setlength{\tabcolsep}{2pt}
\begin{tabular*}{\columnwidth}{@{\extracolsep{\fill}}lcccc@{}}
\toprule
 & $\lambda$ & F1, all $32$ & F1, validity & F2, $7$ sizes\\ \midrule
weather     & $1$    & $0.95$ / $1.68\%$ & $0.99$ / $0.29\%$ & $1.00$ / $0.06\%$\\
electricity & $1$    & $0.95$ / $2.66\%$ & $0.99$ / $0.35\%$ & $1.00$ / $0.18\%$\\
traffic     & $0.01$ & $0.78$ / $2.68\%$ & $0.95$ / $1.26\%$ & $1.00$ / $0.44\%$\\
solar       & $1$    & $0.97$ / $1.89\%$ & $0.99$ / $0.40\%$ & $1.00$ / $0.33\%$\\ \bottomrule
\end{tabular*}
\end{table}

\subsection{Structural checks and limits}\label{sec:numerics-support}\label{sec:numerics-limits}
\paragraph{Path rule.}
On $24$ random out-of-sample designs (one to three loops, one or two rails, one to three channels, random operating points), the implemented coupling-gate structural rule classified all $3814$ second-order component pairs without error; first-order directions were checked separately, also without error.
It also classified $400$ stratified pairs at each of two $U8$ operating points without error.
For the stronger, all-drive-value statement, the full interaction function of $200$ randomly chosen non-admissible readable pairs vanishes to $2\times10^{-31}$ (PM2.5) and $6\times10^{-31}$ (traffic), relative to single-component effects, and non-readable components have exactly zero effect.
The sampled designs included coupling-only and combined coupling/phase encoding, but the path-graph implementation enters mixing gates only; these checks therefore certify the coupling-gate rule, not a phase-gate implementation. \Cref{thm:zeros} itself covers phase gates as well.
A single-circuit counterfactual shows what the rule says before fitting any task (\cref{app:checks}, \cref{tab:rail-placement}): the same inter-rail beam splitters placed at the beginning or end of a step act as fixed basis changes, whereas middle placement creates new interaction pairs.

\paragraph{Breakdown of the low-order approximation.}
In design F1-29 ($\tau=(1,2,4,8,16,32)$, coupling plus EOM with $a=1$), the state errors of the Taylor truncations grow with order on four datasets, for example $1.5\to4.9\to20\to86$ on PM2.5 (\cref{fig:truncation}b).
The same delays with coupling encoding only give $E_1=0.27$ and $E_2=0.12$.
The series converges by \eqref{eq:remainder}, but the EOM phase accumulated along paths through many modulated gates places the first orders in the pre-asymptotic regime.
State fidelity and risk accuracy also separate. On the weather task a design with $\tau=(1,12,24,48)$ and both encodings has second-order state error $E_2=7.7$, yet its second-order and exact risks agree to four digits ($0.4713$).

\begin{figure*}[t]
\centering
\includegraphics[width=.98\linewidth]{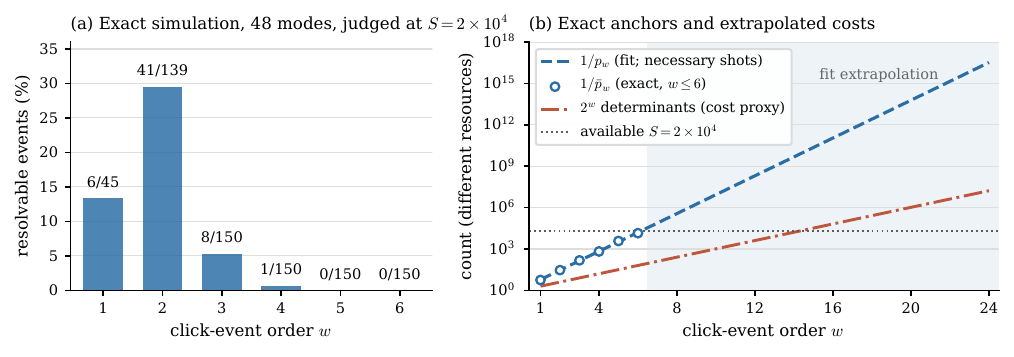}
\caption{A finite-shot boundary for a simulated 48-mode threshold receiver, computed from exact click probabilities on the traffic training windows; no shots are sampled and no hardware is measured. (a) Randomly chosen click events of order $w=1,\ldots,6$ (at most $150$ per order); an event counts as resolvable at $S=2\times10^4$ when its modulation depth $S\operatorname{Var}(p)/[\bar p(1-\bar p)]$ across windows is at least $1$. Labels give resolvable/tested counts. (b) Open circles are exact inverse mean event probabilities. The dashed curve extrapolates the fit $\bar p_w\approx0.206805^w$ beyond $w=6$; $1/\bar p_w$ is only the necessary shot count for one expected event, not sufficient for a resolvable feature. The $2^w$ curve counts determinant evaluations for a $w$-mode marginal and is a cost proxy, \emph{not} a hardness threshold. The shaded region is extrapolation.}
\label{fig:boundary}
\end{figure*}

\section{Discussion and outlook}\label{sec:discussion}\label{sec:design}
The bridge separates what each layer of a TDM reservoir decides. The
topology decides which interactions exist in the state: the path rule
(\cref{thm:zeros}) removes lag--channel interactions for every operating
angle, and independent components remove all interactions between them
(\cref{thm:components}). The encoding decides where nonlinearity enters
and hence the function class (\cref{cor:class}): trigonometric
interactions for gate encoding, an additive per-step model for squeezing
encoding, a linear model for displacement encoding. The receiver decides
which cross-component products a linear readout sees (\cref{cor:span})
and what they cost in shots (\cref{prop:shots}). The task decides which of
the remaining directions matter, through the statistics of
\cref{cor:task}. None of these statements needs a simulation of the
circuit, and each can exclude a design before one is run. The topology
floor is therefore an exclusion certificate, not a ranking of surviving
circuits;
a topology-only preference such as ``$\gcd(\tau)=1$ is better'' can
reverse on a real task (\cref{app:t3}). The exact identity is also
separate from the empirical accuracy of its truncations: when path
phases accumulate, the full circuit must be evaluated.

Squeezing encoding separates the roles of the optics and the receiver.
With the data in the source squeezing and a fixed network,
\cref{cor:squeeze} makes the covariance a linear filter of the per-step
encoded features, so every cross-time product used by a linear readout
must be created at detection by a receiver that is nonlinear in $V$,
for example by photon-number correlations (\cref{prop:gauss}). This
applies to experimentally accessible architectures in which the data
modulate the source squeezing while the network is held fixed. It
concerns only the low-order moments used as features and does not bear
on the sampling hardness of the full output distribution.

\subsection{Quantum--classical boundary}\label{sec:boundary}
TMSV source correlations are nonclassical, but Gaussian states under
Gaussian operations and measurements admit efficient classical
simulation~\cite{bartlett2002efficient,mari2012positive}, and the
finite-frequency path structure has classical wave analogues.
Threshold detection is non-Gaussian and its full sampling problem is
related to Gaussian boson sampling~\cite{hamilton2017gaussian,kruse2019detailed,bulmer2022boundary},
yet in this zero-mean passive model every studied receiver is a
function of the same covariance $V_\theta$. A nonlinear click map can
change the feature class available to a \emph{restricted linear head};
unrestricted nonlinear post-processing of a complete covariance
reconstruction can in principle reproduce those features.
The receivers may still differ in finite-shot efficiency. Bounded-order
click marginals have classical cost
$\mathrm{poly}(n)+m2^w\mathrm{poly}(w)$. In an exactly computed 48-mode
scan with resolvability judged at $S=2\times10^4$, some events up to
order $4$ are resolvable, but none of the randomly chosen order-5 or
order-6 events is (\cref{fig:boundary}a).
Extrapolating the operating-point event decay, as in
\cref{fig:boundary}(b), to $w=16$--$24$ gives
$10^{11}$--$10^{16}$ shots merely for one expected event. This is a
cost/shot mismatch, not a hardness threshold or a general no-go.

Exact GBS simulation can require exponential time even when its working memory is polynomial~\cite{quesada2020exact}, while loss, graph structure and low entanglement define distinct regimes with more efficient classical approximations~\cite{qi2020regimes,oh2022graph,liu2026efficient}.
These complexity results concern sampling or state simulation, not an established advantage in the prediction task studied here.

\subsection{Outlook}
The bridge turns questions about TDM reservoirs that are usually settled
by trial and error into questions about path structure, encoding, receiver
and task statistics. We see four directions in which it can shape the
field.

\paragraph{Predicting and designing TDM chips.}
The lower bounds from topology, encoding and receiver hold for every
operating point, and features whose path amplitudes are known in closed
form, as in the delay-family incumbent, give upper bounds that depend only
on task statistics, because at $\lambda=0$ removing features cannot lower
the population risk. Together they bracket the risk of a candidate
architecture before it is built or simulated, and the second-order
surrogate with a certified error ranks the architectures whose brackets
overlap. Developed into a predictor, this guidance can select TDM chips
tailored to a given dataset, or chips that perform well across a family
of tasks, within the constraints of a photonic platform. The same loop
architectures underlie programmable Gaussian boson sampling and
time-multiplexed cluster states~\cite{madsen2022quantum,yu2023universal,larsen2019deterministic,asavanant2019generation},
so a principled account of what a loop network can compute also informs
the design of photonic quantum processors more broadly.

\paragraph{Photonic attention.}
Attention mechanisms combine pairs of history positions through bilinear
scores. The bridge identifies exactly which such products a photonic
module offers a linear readout: products inside a path component come
from the optics, and products across components come from the
receiver's coincidence events (\cref{thm:components,cor:span}). The
event set therefore acts as a selector of interacting lags, and the
event-span support certifies which interactions a given module can
represent, while \cref{prop:shots} prices them in shots. This provides a
design language for photonic attention-like architectures.

\paragraph{More quantum resources.}
The architectures studied here use Gaussian sources and passive optics.
\Cref{thm:poly} already covers Fock-state inputs, whose photon number
bounds the accessible spectrum (\cref{cor:fock}). Natural next steps are
non-Gaussian sources, photon-number-resolving and adaptive measurements,
and measurement-conditioned feedback or Kerr-type interactions; the last
two break the polynomial structure and call for new analysis. The bridge
makes precise how each added resource changes the function class and the
shot cost for a given task.

\paragraph{Open questions.}
The identities describe the population risk on the training distribution;
a finite-sample or held-out theory requires additional assumptions.
\Cref{prop:shots} compares receivers at equal shots; comparisons at equal
optical energy, and against receivers optimized over all measurements,
remain open.

\section*{Author contributions}
Y.Z.\ designed the architecture, developed and proved the theory, and
carried out its numerical verification. T.Z., Y.J.\ and T.C.\ contributed
to the encoding at the interface with the Gaussian-boson-sampling
receiver. H.T.\ conceived the research direction and supervised the
development and refinement of the project.

\bibliographystyle{unsrtnat}
\bibliography{refs}

\begin{thebibliography}{64}
\providecommand{\natexlab}[1]{#1}
\providecommand{\url}[1]{\texttt{#1}}
\expandafter\ifx\csname urlstyle\endcsname\relax
  \providecommand{\doi}[1]{doi: #1}\else
  \providecommand{\doi}{doi: \begingroup \urlstyle{rm}\Url}\fi

\bibitem[Jaeger and Haas(2004)]{jaeger2004harnessing}
Herbert Jaeger and Harald Haas.
\newblock Harnessing nonlinearity: Predicting chaotic systems and saving energy
  in wireless communication.
\newblock \emph{Science}, 304\penalty0 (5667):\penalty0 78--80, 2004.
\newblock \doi{10.1126/science.1091277}.

\bibitem[Luko{\v{s}}evi{\v{c}}ius and Jaeger(2009)]{lukosevicius2009reservoir}
Mantas Luko{\v{s}}evi{\v{c}}ius and Herbert Jaeger.
\newblock Reservoir computing approaches to recurrent neural network training.
\newblock \emph{Computer Science Review}, 3\penalty0 (3):\penalty0 127--149,
  2009.
\newblock \doi{10.1016/j.cosrev.2009.03.005}.

\bibitem[Appeltant et~al.(2011)Appeltant, Soriano, Van~der Sande, Danckaert,
  Massar, Dambre, Schrauwen, Mirasso, and Fischer]{appeltant2011delay}
Lennert Appeltant, Miguel~C. Soriano, Guy Van~der Sande, Jan Danckaert, Serge
  Massar, Joni Dambre, Benjamin Schrauwen, Claudio~R. Mirasso, and Ingo
  Fischer.
\newblock Information processing using a single dynamical node as complex
  system.
\newblock \emph{Nature Communications}, 2:\penalty0 468, 2011.
\newblock \doi{10.1038/ncomms1476}.

\bibitem[Paquot et~al.(2012)Paquot, Duport, Smerieri, Dambre, Schrauwen,
  Haelterman, and Massar]{paquot2012optoelectronic}
Y.~Paquot, F.~Duport, A.~Smerieri, J.~Dambre, B.~Schrauwen, M.~Haelterman, and
  S.~Massar.
\newblock Optoelectronic reservoir computing.
\newblock \emph{Scientific Reports}, 2:\penalty0 287, 2012.
\newblock \doi{10.1038/srep00287}.

\bibitem[Brunner et~al.(2013)Brunner, Soriano, Mirasso, and
  Fischer]{brunner2013parallel}
Daniel Brunner, Miguel~C. Soriano, Claudio~R. Mirasso, and Ingo Fischer.
\newblock Parallel photonic information processing at gigabyte per second data
  rates using transient states.
\newblock \emph{Nature Communications}, 4:\penalty0 1364, 2013.
\newblock \doi{10.1038/ncomms2368}.

\bibitem[Vandoorne et~al.(2014)Vandoorne, Mechet, Van~Vaerenbergh, Fiers,
  Morthier, Verstraeten, Schrauwen, Dambre, and
  Bienstman]{vandoorne2014experimental}
Kristof Vandoorne, Pauline Mechet, Thomas Van~Vaerenbergh, Martin Fiers, Geert
  Morthier, David Verstraeten, Benjamin Schrauwen, Joni Dambre, and Peter
  Bienstman.
\newblock Experimental demonstration of reservoir computing on a silicon
  photonics chip.
\newblock \emph{Nature Communications}, 5:\penalty0 3541, 2014.
\newblock \doi{10.1038/ncomms4541}.

\bibitem[Fujii and Nakajima(2017)]{fujii2017harnessing}
Keisuke Fujii and Kohei Nakajima.
\newblock Harnessing disordered-ensemble quantum dynamics for machine learning.
\newblock \emph{Physical Review Applied}, 8:\penalty0 024030, 2017.
\newblock \doi{10.1103/PhysRevApplied.8.024030}.

\bibitem[Nakajima et~al.(2019)Nakajima, Fujii, Negoro, Mitarai, and
  Kitagawa]{nakajima2019boosting}
Kohei Nakajima, Keisuke Fujii, Makoto Negoro, Kosuke Mitarai, and Masahiro
  Kitagawa.
\newblock Boosting computational power through spatial multiplexing in quantum
  reservoir computing.
\newblock \emph{Physical Review Applied}, 11\penalty0 (3):\penalty0 034021,
  2019.
\newblock \doi{10.1103/PhysRevApplied.11.034021}.

\bibitem[Mujal et~al.(2021)Mujal, Mart{\'\i}nez-Pe{\~n}a, Nokkala,
  Garc{\'\i}a-Beni, Giorgi, Soriano, and Zambrini]{mujal2021opportunities}
Pere Mujal, Rodrigo Mart{\'\i}nez-Pe{\~n}a, Johannes Nokkala, Jorge
  Garc{\'\i}a-Beni, Gian~Luca Giorgi, Miguel~C. Soriano, and Roberta Zambrini.
\newblock Opportunities in quantum reservoir computing and extreme learning
  machines.
\newblock \emph{Advanced Quantum Technologies}, 4\penalty0 (8):\penalty0
  2100027, 2021.
\newblock \doi{10.1002/qute.202100027}.

\bibitem[Nokkala et~al.(2021)Nokkala, Mart{\'\i}nez-Pe{\~n}a, Giorgi, Parigi,
  Soriano, and Zambrini]{nokkala2021gaussian}
Johannes Nokkala, Rodrigo Mart{\'\i}nez-Pe{\~n}a, Gian~Luca Giorgi, Valentina
  Parigi, Miguel~C. Soriano, and Roberta Zambrini.
\newblock Gaussian states of continuous-variable quantum systems provide
  universal and versatile reservoir computing.
\newblock \emph{Communications Physics}, 4:\penalty0 53, 2021.
\newblock \doi{10.1038/s42005-021-00556-w}.

\bibitem[Garc{\'\i}a-Beni et~al.(2023)Garc{\'\i}a-Beni, Giorgi, Soriano, and
  Zambrini]{garciabeni2023scalable}
Jorge Garc{\'\i}a-Beni, Gian~Luca Giorgi, Miguel~C. Soriano, and Roberta
  Zambrini.
\newblock Scalable photonic platform for real-time quantum reservoir computing.
\newblock \emph{Physical Review Applied}, 20:\penalty0 014051, 2023.
\newblock \doi{10.1103/PhysRevApplied.20.014051}.

\bibitem[Motes et~al.(2014)Motes, Gilchrist, Dowling, and
  Rohde]{motes2014scalable}
Keith~R. Motes, Alexei Gilchrist, Jonathan~P. Dowling, and Peter~P. Rohde.
\newblock Scalable boson sampling with time-bin encoding using a loop-based
  architecture.
\newblock \emph{Physical Review Letters}, 113:\penalty0 120501, 2014.
\newblock \doi{10.1103/PhysRevLett.113.120501}.

\bibitem[Menicucci(2011)]{menicucci2011temporal}
Nicolas~C. Menicucci.
\newblock Temporal-mode continuous-variable cluster states using linear optics.
\newblock \emph{Physical Review A}, 83:\penalty0 062314, 2011.
\newblock \doi{10.1103/PhysRevA.83.062314}.

\bibitem[Yokoyama et~al.(2013)Yokoyama, Ukai, Armstrong, Sornphiphatphong,
  Kaji, Suzuki, Yoshikawa, Yonezawa, Menicucci, and
  Furusawa]{yokoyama2013ultra}
Shota Yokoyama, Ryuji Ukai, Seiji~C. Armstrong, Chanond Sornphiphatphong,
  Toshiyuki Kaji, Shigenari Suzuki, Jun-ichi Yoshikawa, Hidehiro Yonezawa,
  Nicolas~C. Menicucci, and Akira Furusawa.
\newblock Ultra-large-scale continuous-variable cluster states multiplexed in
  the time domain.
\newblock \emph{Nature Photonics}, 7:\penalty0 982--986, 2013.
\newblock \doi{10.1038/nphoton.2013.287}.

\bibitem[Yoshikawa et~al.(2016)Yoshikawa, Yokoyama, Kaji, Sornphiphatphong,
  Shiozawa, Makino, and Furusawa]{yoshikawa2016million}
Jun-ichi Yoshikawa, Shota Yokoyama, Toshiyuki Kaji, Chanond Sornphiphatphong,
  Yu~Shiozawa, Kenzo Makino, and Akira Furusawa.
\newblock Generation of one-million-mode continuous-variable cluster state by
  unlimited time-domain multiplexing.
\newblock \emph{APL Photonics}, 1:\penalty0 060801, 2016.
\newblock \doi{10.1063/1.4962732}.

\bibitem[Asavanant et~al.(2019)Asavanant, Shiozawa, Yokoyama,
  Charoensombutamon, et~al.]{asavanant2019generation}
Warit Asavanant, Yu~Shiozawa, Shota Yokoyama, Baramee Charoensombutamon, et~al.
\newblock Generation of time-domain-multiplexed two-dimensional cluster state.
\newblock \emph{Science}, 366\penalty0 (6463):\penalty0 373--376, 2019.
\newblock \doi{10.1126/science.aay2645}.

\bibitem[Larsen et~al.(2019{\natexlab{a}})Larsen, Guo, Breum,
  Neergaard-Nielsen, and Andersen]{larsen2019deterministic}
Mikkel~V. Larsen, Xueshi Guo, Casper~R. Breum, Jonas~S. Neergaard-Nielsen, and
  Ulrik~L. Andersen.
\newblock Deterministic generation of a two-dimensional cluster state.
\newblock \emph{Science}, 366\penalty0 (6463):\penalty0 369--372,
  2019{\natexlab{a}}.
\newblock \doi{10.1126/science.aay4354}.

\bibitem[Madsen et~al.(2022)Madsen, Laudenbach, Askarani,
  et~al.]{madsen2022quantum}
Lars~S. Madsen, Fabian Laudenbach, Mohsen~Falamarzi Askarani, et~al.
\newblock Quantum computational advantage with a programmable photonic
  processor.
\newblock \emph{Nature}, 606\penalty0 (7912):\penalty0 75--81, 2022.
\newblock \doi{10.1038/s41586-022-04725-x}.

\bibitem[Yu et~al.(2023)Yu, Zhong, Fang, Patel, Li, et~al.]{yu2023universal}
Shang Yu, Zhi-Peng Zhong, Yuhua Fang, Raj~B. Patel, Qing-Peng Li, et~al.
\newblock A universal programmable gaussian boson sampler for drug discovery.
\newblock \emph{Nature Computational Science}, 3:\penalty0 839--848, 2023.
\newblock \doi{10.1038/s43588-023-00526-y}.

\bibitem[Fu et~al.(2026)Fu, Shen, Hu, He, Nie, et~al.]{fu2026chip}
Yu-Xuan Fu, He-Yu Shen, Ke-Ming Hu, Jun-Jie He, Yun-Long Nie, et~al.
\newblock A chip-scale space-time multiplexed {Gaussian} boson sampling
  processor beyond 10,000 photons.
\newblock arXiv:2609.11922, 2026.
\newblock URL \url{https://arxiv.org/abs/2609.11922}.

\bibitem[Enomoto et~al.(2021)Enomoto, Yonezu, Mitsuhashi, Takase, and
  Takeda]{enomoto2021programmable}
Yutaro Enomoto, Kazuma Yonezu, Yosuke Mitsuhashi, Kan Takase, and Shuntaro
  Takeda.
\newblock Programmable and sequential gaussian gates in a loop-based
  single-mode photonic quantum processor.
\newblock \emph{Science Advances}, 7:\penalty0 eabj6624, 2021.
\newblock \doi{10.1126/sciadv.abj6624}.

\bibitem[Yonezu et~al.(2023)Yonezu, Enomoto, Yoshida, and
  Takeda]{yonezu2023time}
Kazuma Yonezu, Yutaro Enomoto, Takato Yoshida, and Shuntaro Takeda.
\newblock Time-domain universal linear-optical operations for universal quantum
  information processing.
\newblock \emph{Physical Review Letters}, 131:\penalty0 040601, 2023.
\newblock \doi{10.1103/PhysRevLett.131.040601}.

\bibitem[Dambre et~al.(2012)Dambre, Verstraeten, Schrauwen, and
  Massar]{dambre2012information}
Joni Dambre, David Verstraeten, Benjamin Schrauwen, and Serge Massar.
\newblock Information processing capacity of dynamical systems.
\newblock \emph{Scientific Reports}, 2:\penalty0 514, 2012.
\newblock \doi{10.1038/srep00514}.

\bibitem[Grigoryeva et~al.(2015)Grigoryeva, Henriques, Larger, and
  Ortega]{grigoryeva2015optimal}
Lyudmila Grigoryeva, Julie Henriques, Laurent Larger, and Juan-Pablo Ortega.
\newblock Optimal nonlinear information processing capacity in delay-based
  reservoir computers.
\newblock \emph{Scientific Reports}, 5:\penalty0 12858, 2015.
\newblock \doi{10.1038/srep12858}.

\bibitem[Garc{\'\i}a-Beni et~al.(2024)Garc{\'\i}a-Beni, Giorgi, Soriano, and
  Zambrini]{garciabeni2024squeezing}
Jorge Garc{\'\i}a-Beni, Gian~Luca Giorgi, Miguel~C. Soriano, and Roberta
  Zambrini.
\newblock Squeezing as a resource for time series processing in quantum
  reservoir computing.
\newblock \emph{Optics Express}, 32:\penalty0 6733--6747, 2024.
\newblock \doi{10.1364/OE.507684}.

\bibitem[Lee et~al.(2024)Lee, Wei, Stenning, Gartside, Prestwood, Seki, Aqeel,
  Karube, Kanazawa, Taguchi, Back, Tokura, Branford, and
  Kurebayashi]{lee2024taskadaptive}
Oscar Lee, Tianyi Wei, Kilian~D. Stenning, Jack~C. Gartside, Dan Prestwood,
  Shinichiro Seki, Aisha Aqeel, Kosuke Karube, Naoya Kanazawa, Yasujiro
  Taguchi, Christian Back, Yoshinori Tokura, Will~R. Branford, and Hidekazu
  Kurebayashi.
\newblock Task-adaptive physical reservoir computing.
\newblock \emph{Nature Materials}, 23:\penalty0 79--87, 2024.
\newblock \doi{10.1038/s41563-023-01698-8}.

\bibitem[Boyd and Chua(1985)]{boyd1985fading}
Stephen Boyd and Leon~O. Chua.
\newblock Fading memory and the problem of approximating nonlinear operators
  with {V}olterra series.
\newblock \emph{IEEE Transactions on Circuits and Systems}, 32\penalty0
  (11):\penalty0 1150--1161, 1985.
\newblock \doi{10.1109/TCS.1985.1085649}.

\bibitem[Gauthier et~al.(2021)Gauthier, Bollt, Griffith, and
  Barbosa]{gauthier2021next}
Daniel~J. Gauthier, Erik Bollt, Aaron Griffith, and Wendson A.~S. Barbosa.
\newblock Next generation reservoir computing.
\newblock \emph{Nature Communications}, 12:\penalty0 5564, 2021.
\newblock \doi{10.1038/s41467-021-25801-2}.

\bibitem[Grigoryeva and Ortega(2018)]{grigoryeva2018universal}
Lyudmila Grigoryeva and Juan-Pablo Ortega.
\newblock Echo state networks are universal.
\newblock \emph{Neural Networks}, 108:\penalty0 495--508, 2018.
\newblock \doi{10.1016/j.neunet.2018.08.025}.

\bibitem[Gonon et~al.(2026)Gonon, Mart{\'\i}nez-Pe{\~n}a, and
  Ortega]{gonon2026feedback}
Lukas Gonon, Rodrigo Mart{\'\i}nez-Pe{\~n}a, and Juan-Pablo Ortega.
\newblock Feedback-driven recurrent quantum neural network universality.
\newblock In \emph{The Fourteenth International Conference on Learning
  Representations}, 2026.
\newblock URL \url{https://arxiv.org/abs/2506.16332}.

\bibitem[Sch{\"u}tte et~al.(2025)Sch{\"u}tte, G{\"o}tting, M{\"u}ntinga, List,
  and Gies]{schuette2025expressive}
Nils-Erik Sch{\"u}tte, Niclas G{\"o}tting, Hauke M{\"u}ntinga, Meike List, and
  Christopher Gies.
\newblock Expressive limits of quantum reservoir computing, 2025.
\newblock URL \url{https://arxiv.org/abs/2501.15528}.

\bibitem[Paparelle et~al.(2026)Paparelle, Henaff, Garc{\'\i}a-Beni, Gillet,
  Montesinos, Giorgi, Soriano, Zambrini, and Parigi]{paparelle2026memory}
Iris Paparelle, Johan Henaff, Jorge Garc{\'\i}a-Beni, {\'E}milie Gillet, Daniel
  Montesinos, Gian~Luca Giorgi, Miguel~C. Soriano, Roberta Zambrini, and
  Valentina Parigi.
\newblock Experimental memory control in continuous-variable optical quantum
  reservoir computing.
\newblock \emph{Nature Photonics}, 20:\penalty0 413--420, 2026.
\newblock \doi{10.1038/s41566-026-01880-9}.

\bibitem[Di~Bartolo et~al.(2026)Di~Bartolo, Piacentini, Ceccarelli, Corrielli,
  Osellame, Cimini, and Sciarrino]{dibartolo2026forecasting}
Rosario Di~Bartolo, Simone Piacentini, Francesco Ceccarelli, Giacomo Corrielli,
  Roberto Osellame, Valeria Cimini, and Fabio Sciarrino.
\newblock Time-series forecasting with multiphoton quantum states and
  integrated photonics.
\newblock \emph{npj Quantum Information}, 12:\penalty0 91, 2026.
\newblock \doi{10.1038/s41534-026-01236-9}.

\bibitem[Cimini et~al.(2026)Cimini, Sohoni, Presutti, Malia, Ma, Yanagimoto,
  Wang, Onodera, Wright, and McMahon]{cimini2026large}
Valeria Cimini, Mandar~M. Sohoni, Federico Presutti, Benjamin~K. Malia,
  Shi-Yuan Ma, Ryotatsu Yanagimoto, Tianyu Wang, Tatsuhiro Onodera, Logan~G.
  Wright, and Peter~L. McMahon.
\newblock Large-scale quantum reservoir computing using a gaussian boson
  sampler.
\newblock \emph{npj Quantum Information}, 12:\penalty0 104, 2026.
\newblock \doi{10.1038/s41534-026-01251-w}.

\bibitem[Kobayashi et~al.(2024)Kobayashi, Fujii, and
  Yamamoto]{kobayashi2024feedback}
Kaito Kobayashi, Keisuke Fujii, and Naoki Yamamoto.
\newblock Feedback-driven quantum reservoir computing for time-series analysis.
\newblock \emph{PRX Quantum}, 5:\penalty0 040325, 2024.
\newblock \doi{10.1103/PRXQuantum.5.040325}.

\bibitem[Schuld et~al.(2021)Schuld, Sweke, and Meyer]{schuld2021effect}
Maria Schuld, Ryan Sweke, and Johannes~Jakob Meyer.
\newblock Effect of data encoding on the expressive power of variational
  quantum-machine-learning models.
\newblock \emph{Physical Review A}, 103\penalty0 (3):\penalty0 032430, 2021.
\newblock \doi{10.1103/PhysRevA.103.032430}.

\bibitem[P{\'e}rez-Salinas et~al.(2020)P{\'e}rez-Salinas, Cervera-Lierta,
  Gil-Fuster, and Latorre]{perezsalinas2020data}
Adri{\'a}n P{\'e}rez-Salinas, Alba Cervera-Lierta, Elies Gil-Fuster, and
  Jos{\'e}~I. Latorre.
\newblock Data re-uploading for a universal quantum classifier.
\newblock \emph{Quantum}, 4:\penalty0 226, 2020.
\newblock \doi{10.22331/q-2020-02-06-226}.

\bibitem[Yu et~al.(2022)Yu, Yao, Li, and Wang]{yu2022power}
Zhan Yu, Hongshun Yao, Mujin Li, and Xin Wang.
\newblock Power and limitations of single-qubit native quantum neural networks.
\newblock In \emph{Advances in Neural Information Processing Systems},
  volume~35, pages 27810--27823, 2022.

\bibitem[Landman et~al.(2023)Landman, Thabet, Dalyac, Mhiri, and
  Kashefi]{landman2023classically}
Jonas Landman, Slimane Thabet, Constantin Dalyac, Hela Mhiri, and Elham
  Kashefi.
\newblock Classically approximating variational quantum machine learning with
  random fourier features.
\newblock In \emph{The Eleventh International Conference on Learning
  Representations}, 2023.
\newblock URL \url{https://openreview.net/forum?id=ymFhZxw70uz}.

\bibitem[Gan et~al.(2022)Gan, Leykam, and Angelakis]{gan2022fock}
Beng~Yee Gan, Daniel Leykam, and Dimitris~G. Angelakis.
\newblock Fock state-enhanced expressivity of quantum machine learning models.
\newblock \emph{EPJ Quantum Technology}, 9:\penalty0 16, 2022.
\newblock \doi{10.1140/epjqt/s40507-022-00135-0}.

\bibitem[Nerenberg et~al.(2025)Nerenberg, Neill, Marcucci, and
  Faccio]{nerenberg2025photon}
Sam Nerenberg, Oliver~D. Neill, Giulia Marcucci, and Daniele Faccio.
\newblock Photon number-resolving quantum reservoir computing.
\newblock \emph{Optica Quantum}, 3\penalty0 (2):\penalty0 201--210, 2025.
\newblock \doi{10.1364/OPTICAQ.553294}.

\bibitem[Hamilton et~al.(2017)Hamilton, Kruse, Sansoni, Barkhofen, Silberhorn,
  and Jex]{hamilton2017gaussian}
Craig~S. Hamilton, Regina Kruse, Linda Sansoni, Sonja Barkhofen, Christine
  Silberhorn, and Igor Jex.
\newblock Gaussian boson sampling.
\newblock \emph{Physical Review Letters}, 119:\penalty0 170501, 2017.
\newblock \doi{10.1103/PhysRevLett.119.170501}.

\bibitem[Kruse et~al.(2019)Kruse, Hamilton, Sansoni, Barkhofen, Silberhorn, and
  Jex]{kruse2019detailed}
Regina Kruse, Craig~S. Hamilton, Linda Sansoni, Sonja Barkhofen, Christine
  Silberhorn, and Igor Jex.
\newblock Detailed study of {G}aussian boson sampling.
\newblock \emph{Physical Review A}, 100\penalty0 (3):\penalty0 032326, 2019.
\newblock \doi{10.1103/PhysRevA.100.032326}.

\bibitem[Quesada et~al.(2018)Quesada, Arrazola, and
  Killoran]{quesada2018gaussian}
Nicol{\'a}s Quesada, Juan~Miguel Arrazola, and Nathan Killoran.
\newblock Gaussian boson sampling using threshold detectors.
\newblock \emph{Physical Review A}, 98:\penalty0 062322, 2018.
\newblock \doi{10.1103/PhysRevA.98.062322}.

\bibitem[Bulmer et~al.(2022{\natexlab{a}})Bulmer, Paesani, Chadwick, and
  Quesada]{bulmer2022threshold}
J.~F.~F. Bulmer, S.~Paesani, R.~S. Chadwick, and N.~Quesada.
\newblock Threshold detection statistics of bosonic states.
\newblock \emph{Physical Review A}, 106:\penalty0 043712, 2022{\natexlab{a}}.
\newblock \doi{10.1103/PhysRevA.106.043712}.

\bibitem[He et~al.(2025)He, Hu, Zhu, Yan, Liang, Zhao, Wang, Guo, Lan, Shang,
  Yin, Jiang, Yang, Tang, and Jin]{he2025deepquantum}
Jun-Jie He, Ke-Ming Hu, Yu-Ze Zhu, Guan-Ju Yan, Shu-Yi Liang, Xiang Zhao, Ding
  Wang, Fei-Xiang Guo, Ze-Feng Lan, Xiao-Wen Shang, Zi-Ming Yin, Xin-Yang
  Jiang, Lin Yang, Hao Tang, and Xian-Min Jin.
\newblock {DeepQuantum}: A {PyTorch}-based software platform for quantum
  machine learning and photonic quantum computing.
\newblock arXiv:2512.18995, 2025.
\newblock URL \url{https://arxiv.org/abs/2512.18995}.

\bibitem[Larsen et~al.(2019{\natexlab{b}})Larsen, Guo, Breum,
  Neergaard-Nielsen, and Andersen]{larsen2019epr}
Mikkel~V. Larsen, Xueshi Guo, Casper~R. Breum, Jonas~S. Neergaard-Nielsen, and
  Ulrik~L. Andersen.
\newblock Fiber-coupled epr-state generation using a single temporally
  multiplexed squeezed light source.
\newblock \emph{npj Quantum Information}, 5:\penalty0 46, 2019{\natexlab{b}}.
\newblock \doi{10.1038/s41534-019-0170-y}.

\bibitem[Tomoda et~al.(2023)Tomoda, Yoshida, Kashiwazaki, Umeki, Enomoto, and
  Takeda]{tomoda2023programmable}
Hiroko Tomoda, Takato Yoshida, Takahiro Kashiwazaki, Takeshi Umeki, Yutaro
  Enomoto, and Shuntaro Takeda.
\newblock Programmable time-multiplexed squeezed light source.
\newblock \emph{Optics Express}, 31:\penalty0 2161--2176, 2023.
\newblock \doi{10.1364/OE.476025}.

\bibitem[Clements et~al.(2016)Clements, Humphreys, Metcalf, Kolthammer, and
  Walmsley]{clements2016optimal}
William~R. Clements, Peter~C. Humphreys, Benjamin~J. Metcalf, W.~Steven
  Kolthammer, and Ian~A. Walmsley.
\newblock Optimal design for universal multiport interferometers.
\newblock \emph{Optica}, 3:\penalty0 1460--1465, 2016.
\newblock \doi{10.1364/OPTICA.3.001460}.

\bibitem[Weedbrook et~al.(2012)Weedbrook, Pirandola, Garc{\'\i}a-Patr{\'o}n,
  Cerf, Ralph, Shapiro, and Lloyd]{weedbrook2012gaussian}
Christian Weedbrook, Stefano Pirandola, Ra{\'u}l Garc{\'\i}a-Patr{\'o}n,
  Nicolas~J. Cerf, Timothy~C. Ralph, Jeffrey~H. Shapiro, and Seth Lloyd.
\newblock Gaussian quantum information.
\newblock \emph{Reviews of Modern Physics}, 84\penalty0 (2):\penalty0 621--669,
  2012.
\newblock \doi{10.1103/RevModPhys.84.621}.

\bibitem[Asavanant et~al.(2021)Asavanant, Charoensombutamon, Yokoyama, Ebihara,
  Nakamura, Alexander, Endo, Yoshikawa, Menicucci, Yonezawa, and
  Furusawa]{asavanant2021timedomain}
Warit Asavanant, Baramee Charoensombutamon, Shota Yokoyama, Takeru Ebihara,
  Tomohiro Nakamura, Rafael~N. Alexander, Mamoru Endo, Jun-ichi Yoshikawa,
  Nicolas~C. Menicucci, Hidehiro Yonezawa, and Akira Furusawa.
\newblock Time-domain-multiplexed measurement-based quantum operations with
  25-mhz clock frequency.
\newblock \emph{Physical Review Applied}, 16:\penalty0 034005, 2021.
\newblock \doi{10.1103/PhysRevApplied.16.034005}.

\bibitem[Wu et~al.(2021)Wu, Xu, Wang, and Long]{wu2021autoformer}
Haixu Wu, Jiehui Xu, Jianmin Wang, and Mingsheng Long.
\newblock Autoformer: Decomposition transformers with auto-correlation for
  long-term series forecasting.
\newblock In \emph{Advances in Neural Information Processing Systems},
  volume~34, pages 22419--22430, 2021.

\bibitem[Lai et~al.(2018)Lai, Chang, Yang, and Liu]{lai2018modeling}
Guokun Lai, Wei-Cheng Chang, Yiming Yang, and Hanxiao Liu.
\newblock Modeling long- and short-term temporal patterns with deep neural
  networks.
\newblock In \emph{Proceedings of the 41st International ACM SIGIR Conference
  on Research and Development in Information Retrieval}, pages 95--104, 2018.
\newblock \doi{10.1145/3209978.3210006}.

\bibitem[Zhou et~al.(2021)Zhou, Zhang, Peng, Zhang, et~al.]{zhou2021informer}
Haoyi Zhou, Shanghang Zhang, Jieqi Peng, Shuai Zhang, et~al.
\newblock Informer: Beyond efficient transformer for long sequence time-series
  forecasting.
\newblock \emph{Proceedings of the AAAI Conference on Artificial Intelligence},
  35\penalty0 (12):\penalty0 11106--11115, 2021.
\newblock \doi{10.1609/aaai.v35i12.17325}.

\bibitem[Liang et~al.(2015)Liang, Zou, Guo, Li, et~al.]{liang2015pm25}
Xuan Liang, Tao Zou, Bin Guo, Shuo Li, et~al.
\newblock Assessing {B}eijing's {PM}$_{2.5}$ pollution: Severity, weather
  impact, {APEC} and winter heating.
\newblock \emph{Proceedings of the Royal Society A}, 471\penalty0
  (2182):\penalty0 20150257, 2015.
\newblock \doi{10.1098/rspa.2015.0257}.

\bibitem[Zeng et~al.(2023)Zeng, Chen, Zhang, and Xu]{zeng2023transformers}
Ailing Zeng, Muxi Chen, Lei Zhang, and Qiang Xu.
\newblock Are transformers effective for time series forecasting?
\newblock \emph{Proceedings of the AAAI Conference on Artificial Intelligence},
  37\penalty0 (9):\penalty0 11121--11128, 2023.
\newblock \doi{10.1609/aaai.v37i9.26317}.

\bibitem[Bartlett et~al.(2002)Bartlett, Sanders, Braunstein, and
  Nemoto]{bartlett2002efficient}
Stephen~D. Bartlett, Barry~C. Sanders, Samuel~L. Braunstein, and Kae Nemoto.
\newblock Efficient classical simulation of continuous variable quantum
  information processes.
\newblock \emph{Physical Review Letters}, 88\penalty0 (9):\penalty0 097904,
  2002.
\newblock \doi{10.1103/PhysRevLett.88.097904}.

\bibitem[Mari and Eisert(2012)]{mari2012positive}
Andrea Mari and Jens Eisert.
\newblock Positive {W}igner functions render classical simulation of quantum
  computation efficient.
\newblock \emph{Physical Review Letters}, 109\penalty0 (23):\penalty0 230503,
  2012.
\newblock \doi{10.1103/PhysRevLett.109.230503}.

\bibitem[Bulmer et~al.(2022{\natexlab{b}})Bulmer, Bell, Chadwick, Jones, Moise,
  Rigazzi, Thorbecke, Haus, Van~Vaerenbergh, Patel, Walmsley, and
  Laing]{bulmer2022boundary}
Jacob F.~F. Bulmer, Bryn~A. Bell, Rachel~S. Chadwick, Alex~E. Jones, Diana
  Moise, Alessandro Rigazzi, Jan Thorbecke, Utz-Uwe Haus, Thomas
  Van~Vaerenbergh, Raj~B. Patel, Ian~A. Walmsley, and Anthony Laing.
\newblock The boundary for quantum advantage in gaussian boson sampling.
\newblock \emph{Science Advances}, 8:\penalty0 eabl9236, 2022{\natexlab{b}}.
\newblock \doi{10.1126/sciadv.abl9236}.

\bibitem[Quesada and Arrazola(2020)]{quesada2020exact}
Nicol{\'a}s Quesada and Juan~Miguel Arrazola.
\newblock Exact simulation of gaussian boson sampling in polynomial space and
  exponential time.
\newblock \emph{Physical Review Research}, 2:\penalty0 023005, 2020.
\newblock \doi{10.1103/PhysRevResearch.2.023005}.

\bibitem[Qi et~al.(2020)Qi, Brod, Quesada, and
  Garc{\'\i}a-Patr{\'o}n]{qi2020regimes}
Haoyu Qi, Daniel~J. Brod, Nicol{\'a}s Quesada, and Ra{\'u}l
  Garc{\'\i}a-Patr{\'o}n.
\newblock Regimes of classical simulability for noisy gaussian boson sampling.
\newblock \emph{Physical Review Letters}, 124:\penalty0 100502, 2020.
\newblock \doi{10.1103/PhysRevLett.124.100502}.

\bibitem[Oh et~al.(2022)Oh, Lim, Fefferman, and Jiang]{oh2022graph}
Changhun Oh, Youngrong Lim, Bill Fefferman, and Liang Jiang.
\newblock Classical simulation of boson sampling based on graph structure.
\newblock \emph{Physical Review Letters}, 128:\penalty0 190501, 2022.
\newblock \doi{10.1103/PhysRevLett.128.190501}.

\bibitem[Liu et~al.(2026)Liu, Jin, Xiang, and Tu]{liu2026efficient}
Tong Liu, Hui-Ke Jin, Tao Xiang, and Hong-Hao Tu.
\newblock Efficient simulation of low-entanglement bosonic gaussian states in
  polynomial time.
\newblock \emph{npj Quantum Information}, 12:\penalty0 110, 2026.
\newblock \doi{10.1038/s41534-026-01254-7}.

\bibitem[Trefethen(2013)]{trefethen2013approximation}
Lloyd~N. Trefethen.
\newblock \emph{Approximation Theory and Approximation Practice}.
\newblock SIAM, Philadelphia, 2013.

\end{thebibliography}

\onecolumn
\makeatletter\@secpenalty=0\makeatother
\appendix
\section{Proofs}\label{app:proofs}
\begingroup
\setlength{\parskip}{0.35\baselineskip}
\setlength{\jot}{4pt}

\subsection{Finite-shot ridge risk}
Let $\xi=\hat z-z_b(G)$, $\tilde y=y-\E y$, and
$\tilde z=\hat z-\E\hat z=(z_b-\E z_b)+\xi$. Conditional
unbiasedness implies $\E[\xi\mid G]=0$ and hence
$\E[\xi(z_b-\E z_b)^\top]=0$. Conditions (C2)--(C3) then give
\begin{align*}
\E[\tilde z\tilde z^\top]
 &=\Sigma_z+\E\,\operatorname{Cov}(\hat z\mid G)
   =\Sigma_z+\bar\Omega_b/S=A,\\
\E[\tilde z\tilde y^\top]
 &=\operatorname{Cov}(z_b,y)+\E[\xi\tilde y^\top]=C_z.
\end{align*}
These are moments over a random training window and its measurement
shots; independence between different training windows is not required
for this population calculation. Expanding the square for an arbitrary
readout $W$ yields
\begin{equation*}
\mathcal R(W)=\frac1q\E\|\tilde y-W^\top\tilde z\|^2
=R_0-\frac2q\tr(C_z^\top W)+\frac1q\tr(W^\top AW).
\end{equation*}
The fitted \emph{population} ridge head minimizes
$\mathcal R(W)+\lambda\|W\|_F^2/q$. Differentiating this objective
gives $(A+\lambda I)W=C_z$, and therefore
$W_\lambda=(A+\lambda I)^{-1}C_z$ when the inverse exists.
Diagonalize $A=U\operatorname{diag}(a_i)U^\top$ and set
$\tilde c_i=(U^\top C_z)_{i,:}$. In this basis the $i$th row of
$U^\top W_\lambda$ is $\tilde c_i/(a_i+\lambda)$, so
\begin{align*}
\mathcal R(W_\lambda)
&=R_0-\frac1q\sum_i\|\tilde c_i\|^2
  \left(\frac{2}{a_i+\lambda}-\frac{a_i}{(a_i+\lambda)^2}\right)\\
&=R_0-\frac1q\sum_i\|\tilde c_i\|^2
  \frac{a_i+2\lambda}{(a_i+\lambda)^2},
\end{align*}
which proves~\eqref{eq:risk}. If $\lambda=0$ and $a_i=0$, then
$u_i^\top\tilde z=0$ almost surely, so $\tilde c_i=u_i^\top C_z=0$;
the corresponding summand is defined as zero and the Moore--Penrose
solution is used. Finally, with $A,C_z$ fixed,
\[
\frac{\partial R_{\theta,b,S}}{\partial\lambda}
=\frac{2\lambda}{q}\sum_{i:a_i+\lambda>0}
  \frac{\|\tilde c_i\|^2}{(a_i+\lambda)^3}\ge0.
\]
This is the derivative of prediction risk, not of the penalized
objective: evaluated at the same $W_\lambda$, the latter exceeds
prediction risk by $\lambda\|W_\lambda\|_F^2/q$.

\paragraph{Proof of \cref{prop:shots}.}
Collect the generalized eigenvectors in $W=[w_1,w_2,\dots]$, so that
$W^\top\bar\Omega_bW=I$ and $W^\top\Sigma_zW=\operatorname{diag}(g_i)$. Then
\begin{equation}
A(S)=\Sigma_z+\frac{\bar\Omega_b}{S}
=W^{-\top}\bigl(\operatorname{diag}(g_i)+S^{-1}I\bigr)W^{-1},\qquad
A(S)^{-1}=W\bigl(\operatorname{diag}(g_i)+S^{-1}I\bigr)^{-1}W^\top .
\end{equation}
At $\lambda=0$,~\eqref{eq:risk} reads
$R=R_0-\tfrac1q\tr(C_z^\top A^{-1}C_z)$, hence
\begin{equation}
R_{\theta,b,S}(0)
=R_0-\frac1q\sum_i\frac{\|w_i^\top C_z\|^2}{g_i+S^{-1}}
=R_0-\sum_{i:\,g_i>0}\omega_i\,\frac{g_iS}{1+g_iS}.
\end{equation}
If $g_i=0$, then $\operatorname{Var}(w_i^\top z_b)=w_i^\top\Sigma_zw_i=0$, so
$w_i^\top(z_b-\E z_b)=0$ almost surely and $w_i^\top C_z=0$; such directions
drop out. Writing $g_iS/(1+g_iS)=1-1/(1+g_iS)$ gives~\eqref{eq:shotlaw}, and
$1/(1+g_iS)\le1/(g_iS)$ gives the sufficient budget
$S\le\sum_i\omega_i/(g_i\varepsilon)$.

\subsection{Encoding-feature polynomial representation}\label{app:proof-poly}
\paragraph{Proof of \cref{thm:poly}.}
\emph{Path expansion.} Write the network as $a_{\rm out}=T(G)a_{\rm in}+Eb$
with vacuum environment modes $b$. The time-unrolled transfer matrix is an
ordered product of embedded gate matrices in which every gate instance
occurs once, and an admitted gate is affine in its features:
\begin{equation}\label{eq:app-T}
T=A_K\cdots A_1,\qquad
A_k=A_k^{(0)}+\sum_jA_k^{(j)}\phi_{k,j}(G).
\end{equation}
Distributing the product over the intermediate modes gives the path sum
\begin{equation}\label{eq:app-path}
T_{os}=\sum_{\pi:s\to o}\operatorname{amp}(\pi),\qquad
\operatorname{amp}(\pi)=\prod_{k\in\pi}(A_k)_{m_k'm_k},
\end{equation}
where $m_k\to m_k'$ is the step of $\pi$ through gate $k$. A path meets a
gate instance at most once, so $\operatorname{amp}(\pi)$ is multilinear:
it has degree at most $\mathbf 1[g\in\pi]$ in $(\phi_g,\bar\phi_g)$.

\emph{Moment expansion.} Take an order-$k$ normally ordered output moment
with $p$ creation and $q$ annihilation operators, $p+q=k$. Substituting the
input--output relation preserves normal order, and every term that
contains an environment operator has zero vacuum expectation. Hence
\begin{equation}\label{eq:app-moment}
\bigl\langle a^\dagger_{o_1}\cdots a^\dagger_{o_p}a_{o'_1}\cdots a_{o'_q}\bigr\rangle
=\sum_{\boldsymbol s,\boldsymbol s'}\ \prod_{i=1}^p\overline{T_{o_is_i}}\,
\prod_{j=1}^qT_{o'_js'_j}\,
\bigl\langle a^\dagger_{s_1}\cdots a^\dagger_{s_p}a_{s'_1}\cdots a_{s'_q}\bigr\rangle_{\rm in}.
\end{equation}
Inserting~\eqref{eq:app-path} writes every term as a product of $k$ path
amplitudes, the \emph{legs} $\pi_1,\dots,\pi_k$, times one input moment.
Terms whose input moment vanishes drop out, so only source-compatible
$k$-tuples contribute.

\emph{(a) Gate degree.} Each leg through $g$ contributes one entry of $A_g$
or $\bar A_g$, which is affine in $(\phi_g,\bar\phi_g)$; therefore
\begin{equation}\label{eq:app-gatedeg}
\deg_{\phi_g}\prod_{\ell=1}^k\operatorname{amp}(\pi_\ell)
\le\#\{\ell:\ g\in\pi_\ell\}\le k .
\end{equation}

\emph{(b) Source degree.} For a product input state
$\rho_{\rm in}=\bigotimes_s\rho_s$ the input moment factorizes,
\begin{equation}\label{eq:app-factor}
\bigl\langle a^\dagger_{s_1}\cdots a_{s'_q}\bigr\rangle_{\rm in}
=\prod_s\Bigl\langle{:}\prod_{\ell\in\mathcal L_s}c_\ell{:}\Bigr\rangle_{\rho_s},
\end{equation}
where $c_\ell\in\{a,a^\dagger\}$ is the input operator of leg $\ell$ and
$\mathcal L_s$ is the set of the $k_s=|\mathcal L_s|$ legs that start at a
mode of $s$, $\sum_sk_s=k$. The factor of $s$ has degree at
most $d_s(k_s)\le d_s(k)$ in $(\psi_s,\bar\psi_s)$, and it equals $1$ if
$k_s=0$.

\emph{(c) Support.} By~\eqref{eq:app-gatedeg} and~\eqref{eq:app-factor}, a
term depends on $\phi_g$ only if a leg passes $g$, and on $\psi_s$ only if
$k_s>0$, that is, only if a leg starts at a mode of $s$. Summing terms can
cancel a monomial but cannot create one that is absent from every term.

\paragraph{Proof of \cref{prop:gauss}.}
(i) For a product of zero-mean states the second moments between different
sources vanish, so the input moment matrices are block diagonal,
\begin{equation}
N_{\rm in}=\bigoplus_sN_s,\qquad M_{\rm in}=\bigoplus_sM_s,
\end{equation}
with $N_s,M_s$ affine in $(\psi_s,\bar\psi_s)$; one TMSV pair counts as one
source. With $T_{:,s}$ the columns of the modes of $s$,~\eqref{eq:NM} gives
\begin{equation}\label{eq:app-NMsum}
N_{\rm out}=\sum_s\bar T_{:,s}N_sT_{:,s}^\top,\qquad
M_{\rm out}=\sum_sT_{:,s}M_sT_{:,s}^\top .
\end{equation}
Every summand contains the features of a single source to first order, and
$V$ is affine in $(N,M,\bar M)$ up to a drive-independent vacuum term.

(ii) Passive linear optics maps Gaussian states to Gaussian states. For a
zero-mean Gaussian state, Wick's theorem for normally ordered products of
operators $c_1,\dots,c_{2j}\in\{a_o,a_o^\dagger\}$ reads
\begin{equation}\label{eq:app-wick}
\bigl\langle{:}c_1\cdots c_{2j}{:}\bigr\rangle
=\sum_{P}\ \prod_{\{a,b\}\in P}\bigl\langle{:}c_ac_b{:}\bigr\rangle ,
\end{equation}
the sum running over perfect matchings $P$ of $\{1,\dots,2j\}$; odd moments
vanish. Every factor is an entry of $N$, $M$ or $\bar M$, affine in the
features of one source by (i). Each term therefore has total degree at
most $j$, and its factors may belong to different sources. For the
examples, $N=u$ and $|M|^2=\sinh^2r\cosh^2r=u^2+u$ give
\begin{equation}
\langle a^{\dagger2}a^2\rangle=2N^2+|M|^2=3u^2+u,\qquad
\langle a^\dagger_1a_1a^\dagger_2a_2\rangle=N_{11}N_{22}+|N_{12}|^2+|M_{12}|^2=u_1u_2,
\end{equation}
where $N_{12}=M_{12}=0$ for independent, unmixed modes.

\paragraph{Proof of \cref{cor:squeeze}.}
With drive-independent gates $T$ is fixed. A squeezed vacuum and a TMSV
pair with squeezing $r$ and fixed phase $\vartheta$ have
\begin{equation}
N_s=u\,I,\qquad
M_s=e^{i\vartheta}v\ \ \text{(single mode)},\qquad
M_s=v\begin{pmatrix}0&1\\1&0\end{pmatrix}\ \ \text{(TMSV)},
\end{equation}
both affine in $(u,v)$. Inserting them into~\eqref{eq:app-NMsum} gives
\eqref{eq:squeeze} with fixed $A_t=\partial V/\partial u_t$ and
$B_t=\partial V/\partial v_t$, and \cref{prop:gauss}(i) excludes products
of two steps. Collect $\psi=(u_1,v_1,\dots,u_L,v_L)$. An affine receiver has
features $z_b=z_0+F\psi$ with fixed $F$, hence
\begin{equation}
\Sigma_z=F\operatorname{Cov}(\psi)F^\top,\qquad
C_z=F\operatorname{Cov}(\psi,y).
\end{equation}
Its Wick noise covariance is quadratic in $V$ and hence in $\psi$, so
$\bar\Omega_b$ requires only $\E\psi$ and $\E\psi\psi^\top$. Photon-number
second moments are quadratic in $V$ by~\eqref{eq:app-wick}, hence
quadratic in $\psi$.

\paragraph{Proof of \cref{cor:displace}.}
Write $a_{\rm in}=a^c_{\rm in}+\mu_{\rm in}$ with $(\mu_{\rm in})_s=\beta_s(G)$
and a centered state $a^c_{\rm in}$ that does not depend on the drive. Then
\begin{equation}
\mu_{\rm out}=T\mu_{\rm in},\qquad
N^c_{\rm out}=\bar TN^c_{\rm in}T^\top,\qquad
M^c_{\rm out}=TM^c_{\rm in}T^\top,
\end{equation}
so $V$ does not depend on the drive. Homodyne means are linear in
$\mu_{\rm out}$, and
\begin{equation}
\langle n_i\rangle=|\mu_i|^2+N^c_{ii}
\end{equation}
adds terms quadratic in $(\beta,\bar\beta)$.

\paragraph{Proof of \cref{cor:fock}.}
An input with at most $N$ photons is a mixture of pure states in
$\operatorname{span}\{\prod_s(a_s^\dagger)^{n_s}|0\rangle:\sum_sn_s\le N\}$.
The network acts as
$a_s^\dagger\mapsto\sum_oT_{os}a_o^\dagger+(\text{environment})$, so the
amplitude of an output pattern $\mathbf m$ is a sum over assignments
$\sigma$ of photons to output modes,
\begin{equation}\label{eq:app-fock}
\Bigl\langle\mathbf m\Bigm|\prod_s\Bigl(\sum_oT_{os}a_o^\dagger\Bigr)^{n_s}\Bigm|0\Bigr\rangle
=\sum_\sigma c_\sigma\prod_{p=1}^{N'}T_{\sigma(p)\,s(p)},\qquad
N'=\sum_sn_s\le N,
\end{equation}
with combinatorial constants $c_\sigma$ and $s(p)$ the source of photon
$p$. Expanding each $T$ entry into paths, the amplitude has degree at most
$n_g$ in $(\phi_g,\bar\phi_g)$, where $n_g$ is the largest number of photons
whose paths can pass $g$. A detection probability is a sum of squared
moduli of such amplitudes, also after tracing out loss modes, and has
degree at most $2n_g$. A coupling gate has entries in
$\operatorname{span}\{e^{\pm i\theta}\}$ on its two modes, whereas a phase
gate multiplies one mode by $e^{i\varphi}$ and leaves the others fixed.
Hence, per gate,
\begin{equation}
\begin{aligned}
&\text{coupling:}&&\text{amplitude}\in\operatorname{span}\{e^{in\kappa G_\alpha}\}_{|n|\le n_g},
&&\text{probability}\in\operatorname{span}\{e^{in\kappa G_\alpha}\}_{|n|\le2n_g},\\
&\text{phase:}&&\text{amplitude}\in\operatorname{span}\{e^{inaG_\alpha}\}_{0\le n\le n_g},
&&\text{probability}\in\operatorname{span}\{e^{inaG_\alpha}\}_{|n|\le n_g}.
\end{aligned}
\end{equation}

\paragraph{Proof of \cref{prop:noise}.}
The shot-averaged state is $\bar\rho=\E_\epsilon\,\rho(G+\epsilon e_\alpha)$,
and expectation values are linear in the state. For one harmonic and
$\epsilon\sim\mathcal N(0,\sigma^2)$,
\begin{equation}
\E_\epsilon\,e^{i\omega_\alpha(G_\alpha+\epsilon)}
=e^{i\omega_\alpha G_\alpha}\,e^{-\omega_\alpha^2\sigma^2/2}.
\end{equation}
For $f(G)=\sum_\omega\hat f_\omega e^{i\omega\cdot G}$ with
$\sum_\omega|\hat f_\omega|<\infty$, dominated convergence then gives
\begin{equation}
\E_\epsilon f(G+\epsilon e_\alpha)
=\sum_\omega\hat f_\omega\,e^{-\omega_\alpha^2\sigma^2/2}\,e^{i\omega\cdot G},
\end{equation}
and independent perturbations of several components multiply the factors.
The covariance of the measured features decomposes as
\begin{equation}
\operatorname{Cov}(\hat z)
=\E_\epsilon\operatorname{Cov}(\hat z\mid\epsilon)
+\operatorname{Cov}_\epsilon\E(\hat z\mid\epsilon),
\end{equation}
and the second term follows by applying the same rule to products of
harmonics.

\subsection{Gate encoding and the characteristic-function interface}\label{app:proof-trig}
Let $e_\alpha$ denote the coordinate vector of drive component
$G_\alpha$. Every time-unrolled gate has a finite harmonic expansion
$A_k(G)=\sum_{h\in\mathcal H_k}A_{k,h}e^{ih\cdot G}$.
For example, on its two modes a beam splitter is
\[
\begin{pmatrix}\cos\theta&-\sin\theta\\ \sin\theta&\cos\theta\end{pmatrix},
\qquad
\cos\theta=\tfrac12(e^{i\theta}+e^{-i\theta}),\quad
\sin\theta=\tfrac1{2i}(e^{i\theta}-e^{-i\theta}).
\]
With~\eqref{eq:affine}, its nonconstant frequencies belong to
$\{\pm\kappa e_\alpha\}$; the identity on unaffected modes adds
$0$. An EOM has frequencies in $\{0,ae_\alpha\}$. Fixed phases,
readout beam splitters and loss beam splitters have only frequency
$0$. The source moments $N_{\rm in},M_{\rm in}$ are independent of $G$.

Write the ordered product as $T=A_K\cdots A_1$. Distributing this
product gives the explicit finite convolution
\[
T(G)=\sum_\nu T_\nu e^{i\nu\cdot G},\qquad
T_\nu=\sum_{h_1+\cdots+h_K=\nu}
 A_{K,h_K}\cdots A_{1,h_1}.
\]
Thus a transfer coefficient is a sum of terms, each selecting one
harmonic at every gate on an allowed optical path. Coincident numerical
frequencies are combined in $T_\nu$. Conjugating $T$ reverses the
frequency sign. Substitution into~\eqref{eq:NM} therefore gives
\begin{align*}
(\widehat N_{\mathrm{out}})_\omega
&=\sum_{\nu'-\nu=\omega}\bar T_\nu N_{\mathrm{in}}T_{\nu'}^\top,\\
(\widehat M_{\mathrm{out}})_\omega
&=\sum_{\nu+\nu'=\omega}T_\nu M_{\mathrm{in}}T_{\nu'}^\top .
\end{align*}
The normal moment uses a frequency \emph{difference}, whereas the
anomalous moment uses a frequency \emph{sum}. Across the two legs of
either expression, each coupling gate fed by $\alpha$ contributes at
most two units of $\kappa$ in absolute value, and each EOM at most two
units of $a$. Since $V$ is affine in $(N,M,\bar M)$, including a
drive-independent vacuum term, every covariance frequency obeys
\[
\omega_\alpha=n_1\kappa+n_2a,
\qquad |n_1|\le 2B_\alpha,\quad |n_2|\le 2E_\alpha .
\]
We include $0$ in the candidate set even if a particular coefficient
vanishes. Reality of $V$ gives
$\hat V_{-\omega}=\overline{\hat V_\omega}$.

It remains to justify the receiver-noise part of the theorem. In a
fixed heterodyne or homodyne setting $s$, the measured quadrature vector
$X_s$ is zero-mean Gaussian with covariance $Q_s$, an affine function
of $V_\eta$. For a single-shot quadratic feature, Wick's identity gives
\[
\operatorname{Cov}(X_{s,i}X_{s,j},X_{s,k}X_{s,l})
=(Q_s)_{ik}(Q_s)_{jl}+(Q_s)_{il}(Q_s)_{jk}.
\]
Each setting's fixed linear reconstruction and its fixed shot fraction
preserve this quadratic dependence; shots from different settings are
independent. Consequently $\Omega_b(V)$ is a quadratic polynomial in
$V$, and its drive frequencies lie in
$\mathcal W_\theta+\mathcal W_\theta$. The fixed training-segment diagonal
standardization multiplies features by a constant matrix $D$ and
noise covariances by $D$ and $D^\top$; it creates no frequency.

Finally, for $z_b=\sum_\omega\hat z_\omega e^{i\omega\cdot G}$,
averaging finite sums term by term yields
\begin{align*}
\E z_b&=\sum_\omega\hat z_\omega\Gamma(\omega),\\
\Sigma_z&=\sum_{\omega,\omega'}\hat z_\omega\hat z_{\omega'}^\top
 [\Gamma(\omega+\omega')-\Gamma(\omega)\Gamma(\omega')],\\
C_z&=\sum_\omega\hat z_\omega
 [\Gamma_y(\omega)-\Gamma(\omega)\E y^\top],\\
\bar\Omega_b&=\sum_\nu\hat\Omega_\nu\Gamma(\nu).
\end{align*}
The use of ordinary transpose is correct because the real feature
vector contains both members of each conjugate frequency pair.
With $0\in\mathcal W_\theta$, all required unweighted characteristic
values belong to $\mathcal W_\theta+\mathcal W_\theta$. Inserting these
moments into~\eqref{eq:risk} proves the finite-frequency risk interface.

\paragraph{Proof of \cref{cor:task}.}
By \cref{thm:bridge}, the risk depends on the task only through
\begin{equation}\label{eq:app-taskmoments}
\E m,\qquad \E(m\tilde m),\qquad \E(my^\top),\qquad \E m',
\qquad m,\tilde m\in\mathcal M,\ m'\in\mathcal M',
\end{equation}
together with $\E y$ and $R_0$. Each row of \cref{tab:task} identifies these
monomials as functions of the drive.

\emph{Gate encoding.} By~\eqref{eq:app-gatedeg} an order-$k$ moment has at
most $k$ legs per gate, and each leg through a coupling gate contributes a
factor in $\operatorname{span}\{1,e^{\pm i\kappa G_\alpha}\}$, through an EOM
one in $\operatorname{span}\{1,e^{\pm iaG_\alpha}\}$. Feature monomials are
therefore $e^{i\omega\cdot G}$ with $\omega\in\mathcal W^{(k)}$, noise monomials
have $\omega'\in\mathcal W^{(k')}$, and
\begin{equation}
\E\bigl(e^{i\omega\cdot G}e^{i\tilde\omega\cdot G}\bigr)=\Gamma(\omega+\tilde\omega),\qquad
\E\bigl(e^{i\omega\cdot G}y^\top\bigr)=\Gamma_y(\omega),\qquad
\E\,e^{i\omega'\cdot G}=\Gamma(\omega').
\end{equation}
With at most $N$ photons, \cref{cor:fock} bounds each gate by $2n_g\le2N$
units, so detection probabilities lie in $\mathcal W^{(2N)}$, and their
multinomial noise $p_i\delta_{ij}-p_ip_j$ in
$\mathcal W^{(2N)}+\mathcal W^{(2N)}$.

\emph{Squeezing encoding} with $r=r^0+\Delta G_\alpha$. Writing
$\sinh^2r$ and $\sinh r\cosh r$ in exponentials,
\begin{equation}
u=\tfrac14\bigl(e^{2r^0}e^{2\Delta G_\alpha}+e^{-2r^0}e^{-2\Delta G_\alpha}-2\bigr),\qquad
v=\tfrac14\bigl(e^{2r^0}e^{2\Delta G_\alpha}-e^{-2r^0}e^{-2\Delta G_\alpha}\bigr),
\end{equation}
so a monomial of degree $j$ in the $(u_t,v_t)$ is a combination of
$e^{\zeta\cdot G}$ with $\zeta\in\Lambda^{+j}$, and
\begin{equation}
\E\bigl(e^{\zeta\cdot G}e^{\tilde\zeta\cdot G}\bigr)=\mathcal L(\zeta+\tilde\zeta),\qquad
\E\bigl(e^{\zeta\cdot G}y^\top\bigr)=\mathcal L_y(\zeta).
\end{equation}
An affine receiver has features of degree $1$ and Wick noise of degree $2$
in $(u,v)$, which gives $\mathcal L$ on $\Lambda^{+2}$ and $\mathcal L_y$ on
$\Lambda$. Photon-number second moments have degree $2$ by
\cref{prop:gauss}(ii), and their noise involves moments of order $8$, of
degree $4$; this gives $\mathcal L$ on $\Lambda^{+4}$ and $\mathcal L_y$ on
$\Lambda^{+2}$. For a general map $r(G)$ the moments of $(u_t,v_t)$ enter
directly (\cref{cor:squeeze}).

\emph{Displacement encoding} with $\beta=\beta^0+bG$. By \cref{cor:displace},
\begin{equation}
\mu_{\rm out}=T\mu_{\rm in}\ \ \text{is affine in }G,\qquad
\langle n_i\rangle=|\mu_i|^2+N^c_{ii}\ \ \text{is quadratic in }G.
\end{equation}
Quadrature means are affine in $G$ and their noise covariance is the
drive-independent centered covariance, so only $\E G$, $\operatorname{Cov}G$
and $\operatorname{Cov}(G,y)$ enter. Second moments and $\langle n_i\rangle$
are quadratic in $G$ with Wick noise quadratic in $\mu$, so moments of $G$ up
to order $4$ and $\E(G_\alpha G_\beta y)$, $\E(G_\alpha y)$ enter.

\paragraph{Proof of \cref{cor:class}.}
(i) By \cref{cor:displace}, $\mu_{\rm out}=T\mu_{\rm in}$ is affine in
$\beta$ and hence in $G$, and quadrature means are affine in
$\mu_{\rm out}$; components outside $\mathcal R$ have zero coefficient.
Thus $z_b-\E z_b=B\,(G_{\mathcal R}-\E G_{\mathcal R})$ for a fixed $B$, and the
argument of \cref{app:proof-topology-floor} with the dictionary
$G_{\mathcal R}$ gives $R\ge R_{\rm or}(G_{\mathcal R})$.
(ii) By \cref{cor:squeeze},
\begin{equation}
z_b=z_0+\sum_s\bigl[a_s\,u(G_{\alpha(s)})+b_s\,v(G_{\alpha(s)})\bigr],
\end{equation}
which is additive over drive components, and the same argument with the
dictionary $\psi_{\mathcal R}$ gives the floor. If $y$ is orthogonal to every
additive function of single drive components, then $C_z=\operatorname{Cov}(z_b,y)=0$,
so~\eqref{eq:risk} gives $R=R_0$ for every $S$ and $\lambda$.
(iii) By \cref{thm:zeros} a nonzero coefficient needs a source-compatible
path pair through the interacting sites; such a pair ends in retained
modes of one component, so by \cref{thm:components} all its sites lie on
paths into that component.

\subsection{Threshold-click approximants}\label{app:proof-click}
\paragraph{Spectral bound.} Represent the passive network together with its
environment modes by an orthogonal symplectic matrix $O$ on the quadratures
of sources and environment. Then
\begin{equation}
V_{\rm full}=O\,(V_{\rm in}\oplus I)\,O^\top,\qquad
\operatorname{spec}V_{\rm full}=\operatorname{spec}(V_{\rm in}\oplus I)\subset[e^{-2r},e^{2r}],
\end{equation}
because the vacuum eigenvalue $1$ lies in that interval. Every block $V_J$
of the retained covariance is a principal submatrix of $V_{\rm full}$, so
Cauchy interlacing gives
$e^{-2r}\le\lambda_{\min}(V_J)\le\lambda_{\max}(V_J)\le e^{2r}$ for every
drive value. Detection loss maps an eigenvalue $\lambda$ to
$\eta_{\det}\lambda+1-\eta_{\det}$, hence
\begin{equation}
\operatorname{spec}Q_J\subset[a,b],\qquad
a=1-\tfrac12\eta_{\det}\bigl(1-e^{-2r}\bigr),\qquad
b=1+\tfrac12\eta_{\det}\bigl(e^{2r}-1\bigr).
\end{equation}

\paragraph{Determinant.} Let $q_1,\dots,q_n$, $n=2|J|$, be the eigenvalues of
$Q_J$ and write $p_d(q_i)=q_i^{-1/2}(1+\delta_i)$ with
$|\delta_i|\le\varepsilon_d$. Then
\begin{equation}
P_{J,d}=\det p_d(Q_J)=\prod_{i=1}^np_d(q_i)=p_0(J)\prod_{i=1}^n(1+\delta_i),\qquad
\Bigl|\prod_{i=1}^n(1+\delta_i)-1\Bigr|\le(1+\varepsilon_d)^n-1,
\end{equation}
which is~\eqref{eq:p0-approx}. The entries of $p_d(Q_J)$ are polynomials of
degree at most $d$ in the entries of $Q_J$, which are affine in $V$, and the
determinant has degree $n$ in them; hence $P_{J,d}$ has degree $2|J|d$ in
$V$. The function $q^{-1/2}$ is analytic inside the Bernstein ellipse with
foci $a,b$ that passes through $0$, whose parameter is
$R=(\sqrt b+\sqrt a)/(\sqrt b-\sqrt a)$. For $1<\varrho<R$ the Chebyshev
interpolant therefore obeys~\cite{trefethen2013approximation}
\begin{equation}
\max_{q\in[a,b]}\bigl|q^{-1/2}-p_d(q)\bigr|\le\frac{4M_\varrho\,\varrho^{-d}}{\varrho-1},\qquad
M_\varrho=\Bigl[\tfrac12(a+b)-\tfrac14(b-a)(\varrho+\varrho^{-1})\Bigr]^{-1/2},
\end{equation}
where $M_\varrho$ bounds $|q^{-1/2}|$ on the ellipse; multiplying by
$\sqrt b$ bounds $\varepsilon_d$.

\paragraph{Events, noise and frequencies.} With $p_0(\emptyset)=1$ and
$0\le p_0(J)\le1$,
\begin{equation}
p_A=\sum_{J\subseteq A}(-1)^{|J|}p_0(J),\qquad
|\tilde p_A-p_A|\le\sum_{\emptyset\ne J\subseteq A}\bigl[(1+\varepsilon_d)^{2|J|}-1\bigr]=e_A,
\end{equation}
and $\tilde p_A$ is a polynomial of degree $2|A|d$ in $V$. For the noise,
using $|p_A|\le1$ and $|\tilde p_A|\le1+e_A$,
\begin{equation}
\bigl|(\tilde p_{A\cup B}-\tilde p_A\tilde p_B)-(p_{A\cup B}-p_Ap_B)\bigr|
\le e_{A\cup B}+e_A+e_B+e_Ae_B,
\end{equation}
with degree at most $2|A\cup B|d\le4wd$. Under gate encoding every entry of
$V$ has frequencies in $\mathcal W^{(2)}$ and a product of $D$ entries in
$\mathcal W^{(2D)}$, which gives $\mathcal W^{(4wd)}$ and $\mathcal W^{(8wd)}$.
The approximant is moment-polynomial, so \cref{thm:bridge} applies to it.
Finally, in the fixed standardized coordinates of (C5), with training
standard deviations $s_A$,
\begin{equation}
\rho^2\le\sum_A\frac{e_A^2}{s_A^2},\qquad
\bigl|(\bar\Omega_b-\bar\Omega_d)_{AB}\bigr|\le\frac{e_{A\cup B}+e_A+e_B+e_Ae_B}{S\,s_As_B},
\end{equation}
where $\bar\Omega_d$ is the approximant's averaged noise covariance.
Inserting these bounds into \cref{app:proof-riskerror} gives the risk
certificate.

\subsection{All-order structural zeros}
For an edge $e$ in the time-unrolled graph, write its gate weight as a
finite sum $w_e(G)=\sum_{h\in\mathcal H_e}c_{e,h}e^{ih\cdot G}$.
The product along a path $\pi$ is consequently
$a_\pi(G)=\operatorname{amp}(\pi)
=\sum_\nu a_{\pi,\nu}e^{i\nu\cdot G}$; a nonzero component
$\nu_\alpha$ requires a gate fed by $G_\alpha$ on that path.
Substituting $T_{js}=\sum_{\pi:s\to j}a_\pi$ into the two input-moment
transformations gives, entry by entry,
\begin{align*}
N_{jk}(G)
&=\sum_{s,s'}(N_{\rm in})_{ss'}
  \sum_{\pi:s\to j}\sum_{\pi':s'\to k}
  \overline{a_\pi(G)}a_{\pi'}(G),\\
M_{jk}(G)
&=\sum_{s,s'}(M_{\rm in})_{ss'}
  \sum_{\pi:s\to j}\sum_{\pi':s'\to k}
  a_\pi(G)a_{\pi'}(G).
\end{align*}
For example, a summand in $N_{jk}$ with local path frequencies
$\nu,\nu'$ has global frequency $\omega=\nu'-\nu$, while the
corresponding summand in $M_{jk}$ has $\omega=\nu+\nu'$.
The nonzero entries of $N_{\rm in}$ are local to a TMSV source pair
(including single-mode diagonal entries); $M_{\rm in}$ also correlates
the two arms of that pair. Environment vacuum modes have zero $N$ and
$M$ moments. Hence every non-vacuum term uses a source-compatible path
pair as defined in \cref{thm:zeros}.

Let $E(\pi,\pi')$ be the set of drive components feeding gates on
$\pi\cup\pi'$. If $\alpha\notin E(\pi,\pi')$, then both path
amplitudes are independent of $G_\alpha$ and the corresponding term
has $\omega_\alpha=0$. Thus a term with
$S\subseteq\supp\omega$ requires a compatible pair with
$S\subseteq E(\pi,\pi')$. If no such pair exists, every Fourier
coefficient in $\mathcal I_S V_\theta$ is zero, which proves the
all-order assertion. Summing allowed terms can only cancel
coefficients, never create a frequency missing from every term;
source-compatible reachability is therefore necessary, not sufficient.

\subsection{Independent components}\label{app:proof-components}
\paragraph{Proof of \cref{thm:components}.}
For complex $\xi$ on the retained modes let
$\chi(\xi)=\bigl\langle e^{\sum_o\xi_oa_o^\dagger}e^{-\sum_o\bar\xi_oa_o}\bigr\rangle$
be the normally ordered characteristic function. Substituting
$a_{\rm out}=Ta_{\rm in}+Eb$, the input and environment operators commute,
and the vacuum environment contributes the factor $1$; hence
\begin{equation}
\chi_{\rm out}(\xi)=\chi_{\rm in}(T^\dagger\xi),\qquad
(T^\dagger\xi)_m=\sum_o\overline{T_{om}}\,\xi_o .
\end{equation}
Write $\xi=\sum_k\xi^{(k)}$ with $\xi^{(k)}$ supported on $C_k$, and let
$\mathcal S_k$ be the sources joined to $C_k$. By the definition of the
components, $T_{om}=0$ for $o\in C_k$ and $m$ a mode of a source outside
$\mathcal S_k$, so $T^\dagger\xi^{(k)}$ is supported on the modes of
$\mathcal S_k$. The sets $\mathcal S_k$ are disjoint and the input is a
product over sources, therefore
\begin{equation}
\chi_{\rm out}(\xi)=\chi_{\rm in}\Bigl(\sum_kT^\dagger\xi^{(k)}\Bigr)
=\prod_k\chi_{\mathcal S_k}\bigl(T^\dagger\xi^{(k)}\bigr),
\end{equation}
where $\chi_{\mathcal S_k}$ belongs to $\bigotimes_{s\in\mathcal S_k}\rho_s$. The
right side is the characteristic function of $\bigotimes_k\rho_k$, where
$\rho_k$ has characteristic function
$\xi^{(k)}\mapsto\chi_{\mathcal S_k}(T^\dagger\xi^{(k)})$; since $\chi$
determines the state, $\rho_{\rm out}=\bigotimes_k\rho_k$. The state $\rho_k$
depends on $T$ only through its rows $C_k$ and the columns of
$\mathcal S_k$, which by~\eqref{eq:app-path} involve only the sites on
paths from $\mathcal S_k$ into $C_k$. These site sets are disjoint: a gate
on paths into $C_k$ and $C_{k'}$ would connect the two paths, so a source
of one component would reach the other. For a product observable,
$\langle\bigotimes_kO_k\rangle=\prod_k\operatorname{tr}(\rho_kO_k)$.
(a) A zero-mean state has $\langle x_i\rangle=\langle p_i\rangle=0$, so for $i,j$
in different components $\langle x_ix_j\rangle=\langle x_i\rangle\langle x_j\rangle=0$,
and likewise for the other quadrature pairs.
(b) The event ``every mode of $E$ clicks'' is
$\bigotimes_{o\in E}(I-|0\rangle\langle0|_o)$; every click pattern on $E$ and
every photon-number monomial $\prod_on_o^{m_o}$ is likewise a product over
modes, hence over components.

\paragraph{Proof of \cref{cor:span}.}
By \cref{thm:components}(b) the feature of $E$ is
$\prod_{k\in K(E)}f_{E,k}$ with $f_{E,k}$ a function of the sites of
$C_k$ alone, so no monomial involves a site outside the components met
by $E$. (i) If every event lies in one component, every feature is
component-local and every linear head gives
\begin{equation}
W^\top z_b=c+\sum_kh_k\bigl(\text{sites of }C_k\bigr),
\end{equation}
an additive function across components. The proof of
\cref{cor:topology-floor} in \cref{app:proof-topology-floor} uses only that
the centered features are a fixed linear map of the chosen dictionary plus
conditionally unbiased noise; it therefore holds for any dictionary that
spans these additive functions. In the delay family the certificate
of \cref{app:analytic-certificate} uses exactly such a dictionary.
(ii) A feature of an event that meets at most $w$ components is a product
of at most $w$ component-local functions. (iii) follows from (a).

\subsection{Topology-conditioned oracle risk floor}\label{app:proof-topology-floor}
Center the dictionary, $\tilde\phi=\phi-\E\phi$, and set
$\tilde y=y-\E y$. For the exact affine case,
\cref{thm:bridge,thm:zeros} show that every nonconstant Fourier term
of $z_b$ belongs to the real sine--cosine outer dictionary. Hence
$z_b-\E z_b=B\tilde\phi_{\mathcal T}^{\rm ex}$ for a real matrix
$B$, irrespective of cancellations among its columns. For the
second-order surrogate, the chain rule in~\eqref{eq:second} and the
outer support~\eqref{eq:second-support} similarly give
$z_b^{(2)}-\E z_b^{(2)}=B\tilde\phi_{\mathcal T,b}^{(2)}$ when
the remainder is set to zero. The following argument applies to either
dictionary, provided its surrogate shot covariance is positive
semidefinite.

Write $\tilde{\hat z}=B\tilde\phi+\xi$ and $L=B^\top W$.
Conditional unbiasedness gives
$\E[\xi\tilde\phi^\top]=0$; (C3) gives
$\E[\xi\tilde y^\top]=0$, and
$\E[\xi\xi^\top]=\bar\Omega/S\succeq0$. Expanding the
readout's risk without discarding any term therefore yields
\begin{align*}
\mathcal R_\phi(W)
={}&R_0-\frac2q\tr(C_\phi^\top L)
 +\frac1q\tr(L^\top\Sigma_\phi L)\\
&+\frac1{qS}\tr(W^\top\bar\Omega W).
\end{align*}
If $v\in\ker\Sigma_\phi$, then
$\E(v^\top\tilde\phi)^2=0$, so
$v^\top\tilde\phi=0$ almost surely and $v^\top C_\phi=0$.
Consequently $C_\phi$ lies in $\operatorname{range}(\Sigma_\phi)$,
even when $\Sigma_\phi$ is singular. Put
$L_*=\Sigma_\phi^\dagger C_\phi$; then
$\Sigma_\phi L_*=C_\phi$. Adding and subtracting $L_*$ in the
quadratic expression gives the exact completion of the square
\begin{align*}
\mathcal R_\phi(W)=R_0
&-\frac1q\tr(C_\phi^\top\Sigma_\phi^\dagger C_\phi)\\
&+\frac1q\|\Sigma_\phi^{1/2}(L-L_*)\|_F^2
 +\frac1{qS}\tr(W^\top\bar\Omega W).
\end{align*}
The first line is $R_{\rm or}(\phi)$, while the other terms are
nonnegative. For the exact affine dictionary this proves
\eqref{eq:exact-floor} for every $W$, including the
fixed-$\lambda$ ridge solution, and proves~\eqref{eq:oracle-gap}.
For the truncated dictionary the identical calculation instead
establishes the surrogate floor~\eqref{eq:second-floor}.

To incorporate the feature budget, set
$T=\Sigma_\phi^{\dagger/2}C_\phi$ and
$M=B\Sigma_\phi^{1/2}$. The range argument above implies
$C_\phi=\Sigma_\phi^{1/2}T$; thus the noiseless feature covariance
and label cross-covariance are $MM^\top$ and $MT$. Minimizing over
all linear heads on those noiseless features captures at most
\[
\frac1q\tr[T^\top M^\top(MM^\top)^\dagger MT]
=\frac1q\|P_{\mathcal U}T\|_F^2,
\quad \mathcal U=\operatorname{range}(M^\top).
\]
Indeed, an SVD of $M$ shows that
$M^\top(MM^\top)^\dagger M=P_{\mathcal U}$, the projector onto
its row space. This space has dimension at most the receiver's
$d_b$ features. By the variational principle applied to
$TT^\top=\Sigma_\phi^{\dagger/2}C_\phi C_\phi^\top
\Sigma_\phi^{\dagger/2}$,
$\|P_{\mathcal U}T\|_F^2\le\sum_{i=1}^{d_b}\mu_i$.
Measurement noise and a restricted ridge head cannot increase the
noiseless optimum; subtracting this upper bound on captured target
variance from $R_0$ proves~\eqref{eq:oracle-rank}. Finally, a
certified $|R-R^{(2)}|\le\varepsilon_R$ combines with
\eqref{eq:second-floor} to give
$R\ge R_{\rm or}(\phi_{\mathcal T,b}^{(2)})-\varepsilon_R$.

\subsection{Derivation of the delay-family certificate}\label{app:analytic-certificate}
In a one-loop circuit of delay $d$, a time-bin amplitude either exits
at its present step or makes a loop traversal from $t$ to $t+d$.
Repeated traversals therefore keep a path within one residue class
modulo $d$. The normal moments of the independent TMSV inputs are
diagonal in injection time, and an anomalous moment connects only the
reference and memory modes injected at the \emph{same} time. A
source-compatible path pair consequently also stays within one
residue class. Denote the modulated indices in class $r$ by
$I_r=\{i\in\{1,\ldots,8\}:i\equiv r\pmod d\}$.
Because a covariance entry is a \emph{sum} over source-compatible path
pairs, every feature of a fixed, input-independent affine receiver
belongs to the additive function space
\[
\left\{c_0+\sum_{r=0}^{d-1}f_r(u_{I_r}):
       f_r\text{ arbitrary functions}\right\}.
\]
This statement grants the circuit more freedom than its actual gate
amplitudes, and so is safe for an optimistic lower bound.

Under the full-factorial measure the characters
$\chi_A(u)=\prod_{i\in A}u_i$ form an orthonormal basis:
$\E[\chi_A\chi_B]=\mathbf1_{A=B}$. A function of $u_{I_r}$
has Walsh coefficients only on subsets $A\subseteq I_r$.
Thus every target pair $u_i u_j=\chi_{\{i,j\}}$ with
$i,j$ in different residue classes is orthogonal to \emph{every}
attainable affine feature, while the $\sum_r\binom{n_r}{2}$
within-class pairs are optimistically granted in full. Since the
$28$ pair terms of~\eqref{eq:analytic-target} are orthonormal,
their missing squared norm is
\[
\frac{28-\sum_{r=0}^{d-1}\binom{n_r}{2}}{28}.
\]
This proves~\eqref{eq:analytic-floor} after dropping shot noise,
ridge penalty and the feature-dimension constraint. For example, the
class sizes at $d=2,3,4,8$ are respectively
$(4,4)$, $(3,3,2)$, $(2,2,2,2)$ and eight singletons, yielding
the four floors stated in the text.

For $d=1$, the specified memory path enters through a fixed
$\pi/4$ coupler, stays in the loop through all eight modulated gates,
and exits through another fixed $\pi/4$ coupler. Its two fixed
amplitude factors give $\sin^2(\pi/4)=1/2$, and each staying
factor is
$\cos(\pi/4+\delta u_i)=2^{-1/2}(c-su_i)$.
Multiplication by the TMSV cross-quadrature moment
$2\sinh r\cosh r$ gives~\eqref{eq:analytic-feature}, with
$|A_0|=\sinh r\cosh r$. Expanding the product in the Walsh basis,
\[
z(u)=A_0 2^{-4}\sum_{A\subseteq\{1,\ldots,8\}}
 (-s)^{|A|}c^{8-|A|}\chi_A(u).
\]
Every pair character has the same coefficient
$A_0 2^{-4}c^6s^2$. Orthogonality now gives, step by step,
\begin{align*}
\E z&=A_0 2^{-4}c^8,\\
\operatorname{Var}(z)
 &=A_0^2 2^{-8}\bigl((c^2+s^2)^8-c^{16}\bigr)
   =A_0^2 2^{-8}(1-c^{16}),\\
\operatorname{Cov}(z,y)
 &=A_0 2^{-4}\sqrt{28}\,c^6s^2.
\end{align*}
The measured feature is unbiased and has unconditional variance
$\operatorname{Var}(z)+\bar\Omega_z/S$; its covariance with $y$
is unchanged. The one-feature least-squares reduction
$\operatorname{Cov}(z,y)^2/
(\operatorname{Var}(z)+\bar\Omega_z/S)$ gives
\eqref{eq:analytic-gap}. At $\delta=\pi/6$, the ideal reduction is
$20412/58975$, and hence the ideal risk is $38563/58975$.
More explicitly, the same feasible circuit has risk below the
$d=3$ floor of $3/4$ whenever
\[
S>\frac{2^{24}}{22673}\frac{\bar\Omega_z}{A_0^2}.
\]
This supplies a finite-shot witness for excluding
$d\in\{3,\ldots,8\}$ in the stated family; it asserts nothing about
nonlinear click features.

\subsection{Gate-level directional recursion}\label{app:gate-recursion}
Let the gates be applied in order $k=1,2,\dots$ with matrices $A_k$,
and let $x_k,y_k$ be the directional drive increments seen by gate $k$.
Write $T^{(k)}=A_k\cdots A_1$. For a gate driven by one scalar
coordinate, its first and mixed directional derivatives are
$DA_k[x]=x_kA_k'$ and $D^2A_k[x,y]=x_ky_kA_k''$; a fixed gate has
$x_k=y_k=0$. Differentiating
$T^{(k)}=A_kT^{(k-1)}$ once gives
$DT^{(k)}[x]=A_kDT^{(k-1)}[x]+x_kA_k'T^{(k-1)}$.
Differentiating this identity in direction $y$ gives four terms:
\begin{align*}
D^2T^{(k)}[x,y]={}&A_kD^2T^{(k-1)}[x,y]
 +x_ky_kA_k''T^{(k-1)}\\
&+x_kA_k'DT^{(k-1)}[y]
 +y_kA_k'DT^{(k-1)}[x].
\end{align*}
The second term differentiates one gate twice; the final two select
two distinct, time-ordered gates. We track these contributions
separately as $T_S$ and $T_P$.
Initialize $(T,T_x,T_y,T_S,T_P)=(I,0,0,0,0)$ and update
\begin{align*}
T&\leftarrow A_kT,\\
T_x&\leftarrow A_kT_x+x_kA_k'T,\\
T_y&\leftarrow A_kT_y+y_kA_k'T,\\
T_S&\leftarrow A_kT_S+x_ky_kA_k''T,\\
T_P&\leftarrow A_kT_P+x_kA_k'T_y+y_kA_k'T_x,
\end{align*}
using pre-update values on every right-hand side. The product rule
then proves inductively that $DT[x]=T_x$ and
$D^2T[x,y]=T_S+T_P=:T_{xy}$. In particular, the update for $T_P$
uses the \emph{old} $T_x,T_y$; otherwise it would count the current
gate twice.

The covariance has two transfer-matrix legs. Applying the product
rule to each moment transformation in~\eqref{eq:NM} yields
\begin{align*}
D^2N[x,y]={}&\bar T_{xy}N_{\rm in}T^\top
 +\bar TN_{\rm in}T_{xy}^\top\\
&+\bar T_xN_{\rm in}T_y^\top
 +\bar T_yN_{\rm in}T_x^\top,\\
D^2M[x,y]={}&T_{xy}M_{\rm in}T^\top
 +TM_{\rm in}T_{xy}^\top\\
&+T_xM_{\rm in}T_y^\top
 +T_yM_{\rm in}T_x^\top.
\end{align*}
The first two terms in each line include the single-path
single-gate curvature $T_S$ and the ordered same-path pairs $T_P$.
The last two place the differentiated gates on opposite legs joined
by a nonzero source moment. Since $V$ is affine in
$(N,M,\bar M)$, these expressions also give $D^2V[x,y]$ without a
full Hessian.

\subsection{Conditional risk error}\label{app:proof-riskerror}
Let $z^{(k)}$ denote the \emph{centered} $k$th-order feature
surrogate in~\eqref{eq:lifted}, and write the exact centered feature
as $z^{(k)}+r_c$. The exact and surrogate risks below use the same
feature coordinates, target centering and fixed ridge parameter
$\lambda$; any difference in shot covariance is counted explicitly.
Put $\rho^2=\E\|r_c\|^2$,
$\sigma_k^2=\E\|z^{(k)}\|^2$,
$\sigma_y^2=qR_0$,
$A_k=\Sigma_k+\bar\Omega_k/S$,
$\delta=\|A-A_k\|$, and
$\alpha=\lambda_{\min}(A_k)+\lambda>0$.
If $\delta<\alpha$, the same fixed-$\lambda$ prediction risks obey
\[
\begin{aligned}
|\Delta R|\le\frac1q\Bigl[&
  2\|C_k\|_F\|\mathcal G_k\|\|\Delta_C\|_F
 +\|\mathcal G_k\|\|\Delta_C\|_F^2\\
&+\|\Delta\mathcal G\|
  (\|C_k\|_F+\|\Delta_C\|_F)^2\Bigr],
\end{aligned}
\]
where $C_k=\operatorname{Cov}(z^{(k)},y)$,
$\mathcal G_k=(A_k+\lambda I)^{-1}+\lambda(A_k+\lambda I)^{-2}$,
$\|\mathcal G_k\|\le2/\alpha$, and
$\|\Delta\mathcal G\|\le
\delta[2+\alpha/(\alpha-\delta)]/[\alpha(\alpha-\delta)]$.
The bounds below prove this certificate.
Indeed, with $\tilde y=y-\E y$, expanding the covariance and
cross-covariance of $z^{(k)}+r_c$ gives
\begin{align*}
\Delta_C&=\E[r_c\tilde y^\top],
&
\Delta_\Sigma
&=\E[z^{(k)}r_c^\top+r_c(z^{(k)})^\top+r_cr_c^\top].
\end{align*}
Cauchy--Schwarz yields
$\|\Delta_C\|_F\le\rho\sigma_y$ and
$\|\Delta_\Sigma\|\le2\rho\sigma_k+\rho^2$:
for example,
$\|\E[z^{(k)}r_c^\top]\|
\le\E[\|z^{(k)}\|\|r_c\|]\le\sigma_k\rho$.
If shot covariances are not held fixed, the total Gram perturbation is
$\Delta_A=\Delta_\Sigma+(\bar\Omega_b-\bar\Omega_k)/S$,
so $\delta=\|\Delta_A\|\le2\rho\sigma_k+\rho^2+
\|\bar\Omega_b-\bar\Omega_k\|/S$.

Set $R=(A_k+\lambda I)^{-1}$ and
$R'=(A_k+\Delta_A+\lambda I)^{-1}$. The assumption $\delta<\alpha$
implies
\[
\|R\|\le\frac1\alpha,\qquad
\|R'\|\le\frac1{\alpha-\delta},\qquad
\|R'-R\|\le\frac{\delta}{\alpha(\alpha-\delta)}
\]
by the resolvent identity. Since
$R'^2-R^2=R'(R'-R)+(R'-R)R$ and
$\mathcal G_k=R+\lambda R^2$, using $\lambda\le\alpha$ gives
\begin{align*}
\|\Delta\mathcal G\|
&\le\frac{\delta}{\alpha(\alpha-\delta)}
 +\lambda\left[
   \frac{\delta}{\alpha(\alpha-\delta)^2}
  +\frac{\delta}{\alpha^2(\alpha-\delta)}
 \right]\\
&\le\frac{\delta}{\alpha(\alpha-\delta)}
  \left(2+\frac{\alpha}{\alpha-\delta}\right).
\end{align*}
Finally, put $C'=C_k+\Delta_C$ and
$\mathcal G'=\mathcal G_k+\Delta\mathcal G$. Expanding the quadratic form,
\begin{align*}
\tr(C'^\top\mathcal G'C')-\tr(C_k^\top\mathcal G_kC_k)
={}&\tr(C'^\top\Delta\mathcal G\,C')\\
&+2\tr(\Delta_C^\top\mathcal G_kC_k)
 +\tr(\Delta_C^\top\mathcal G_k\Delta_C).
\end{align*}
For matrices with multiple target columns, use
$|\tr(X^\top GX)|\le\|G\|\|X\|_F^2$ and
$|\tr(X^\top GY)|\le\|G\|\|X\|_F\|Y\|_F$.
Applying these inequalities and the resolvent bounds gives the claimed risk bound.
\endgroup

\section{Numerical checks}\label{app:checks}
\paragraph{Gate encoding.}
We sample $V_\theta$ on an equispaced grid covering a full period of the base frequency, obtain the Fourier coefficients by a discrete Fourier transform, and reconstruct $V_\theta$ at random off-grid points. Any frequency outside the predicted set would alias and spoil the reconstruction (\cref{tab:exact}). The predicted frequency bound is tight when a component feeds one gate: the $\pm2\kappa$ harmonic carries a median $56\%$ of the energy. When a component feeds several gates the predicted set is sufficient but not tight at its corners.

\begin{table}[htb]
\centering\small
\caption{Checks of \cref{cor:gate,thm:bridge} for gate encoding (PM2.5 training windows). Residuals are relative off-grid reconstruction errors.}
\label{tab:exact}
\begin{tabular}{@{}>{\raggedright}p{6.0cm}>{\raggedright}p{3.0cm}>{\raggedright\arraybackslash}p{7.2cm}@{}}
\toprule
Check & Design & Result\\ \midrule
single components (40), 5 harmonics & $U8$, coupling & residual $\le3.6\times10^{-14}$; out-of-band energy $\le7.7\times10^{-29}$\\
single components (12), 81 harmonics & F1-29, coupling + EOM & residual $\le4.2\times10^{-13}$; out-of-band energy $\le2.0\times10^{-26}$\\
component pairs (20), $5\times5$ grid & $U8$ & residual $\le2.0\times10^{-14}$\\
component triples (10), $5^3$ grid & $U8$ & residual $\le2.0\times10^{-14}$\\
spectral risk identity, one triple, heterodyne & $U8$ & $\max_{\lambda\in\{0.01,1,100\}}|\Delta R|/R=2.4\times10^{-15}$\\ \bottomrule
\end{tabular}
\end{table}

\paragraph{Beyond gate encoding.}
\Cref{tab:general} checks the remaining statements of \cref{sec:bridge} on small synthetic circuits with random drives. With source squeezing varied step by step, the covariance is affine in $(u_t,v_t)$, distinct steps never multiply in it, and the gate harmonics keep the predicted bound when a drive component feeds two gates. Every gate--source product predicted absent by the path rule vanishes, and every predicted one is present in this circuit. At photon-number order the picture changes as \cref{prop:gauss} states: a product of two sources appears in $\langle n_in_j\rangle$ exactly when one source reaches mode $i$ and the other mode $j$, products of three sources never appear, and the gate harmonics extend to four legs. For Fock inputs, computed exactly from permanents in a two-loop circuit whose paths interfere, the harmonic bounds of \cref{cor:fock} hold and are attained for coupling encoding. The damping factor of \cref{prop:noise} agrees with direct quadrature over the drive noise.

\begin{table}[htb]
\centering\small
\caption{Checks of \cref{thm:poly,prop:gauss,cor:squeeze,cor:fock,prop:noise} beyond gate encoding (synthetic drives; residuals are relative).}
\label{tab:general}
\begin{tabular}{@{}>{\raggedright}p{7.2cm}>{\raggedright}p{3.2cm}>{\raggedright\arraybackslash}p{5.8cm}@{}}
\toprule
Statement & Circuit & Result\\ \midrule
covariance affine in one step's $(u_t,v_t)$ (\cref{cor:squeeze}) & 2 loops, 1 rail, $L=10$ & fit residual $\le4.3\times10^{-14}$ over $9$ steps\\
no product of two sources in the covariance (\cref{prop:gauss}(i)) & same & all $45$ step pairs: $\le2.0\times10^{-16}$\\
gate harmonics $|n|\le2B_\alpha=4$ under two squeezing profiles & same, fan-out $2$ & residual $\le2.1\times10^{-14}$; out-of-band energy $\le2.1\times10^{-33}$\\
gate--source products vs.\ path rule (\cref{thm:poly}(c)) & same & $100$ pairs: $36$ predicted, $36$ nonzero; no miss, no false alarm\\
photon-number moments at most quadratic in one source; non-affine iff the source reaches both modes (\cref{prop:gauss}(ii)) & same & $172$ cases: quadratic-fit residual $\le4.7\times10^{-12}$; $53$ predicted non-affine, $53$ found\\
two-source products in photon-number moments vs.\ path rule (\cref{thm:poly}(c)) & same & $1620$ cases: $239$ predicted, $239$ nonzero; no miss, no false alarm\\
no three-source products in photon-number moments (\cref{prop:gauss}(ii)) & same & $120$ triples: $\le2.0\times10^{-15}$\\
gate harmonics of photon-number moments $|n|\le4B_\alpha=8$ (\cref{thm:poly}(a)) & same, fan-out $2$ & residual $\le2.9\times10^{-14}$; out-of-band $\le9.5\times10^{-33}$; highest harmonic $8$ (covariance: $4$)\\
Fock inputs, $N=1,2,3$, coupling encoding: $|n|\le2N$ (\cref{cor:fock}) & 2 loops $\tau=(1,2)$, lossless & residual $\le1.8\times10^{-15}$; out-of-band $\le7.1\times10^{-32}$; bound attained\\
Fock inputs, $N=1,2,3$, phase encoding: $|n|\le N$ & same & residual $\le5.4\times10^{-15}$; out-of-band $\le1.7\times10^{-32}$\\
drive-noise damping $e^{-\omega^2\sigma^2/2}$ (\cref{prop:noise}) & 2 loops, fan-out $2$ & $\sigma\in\{0.3,1,2.5\}$: $\le8.1\times10^{-12}$\\ \bottomrule
\end{tabular}
\end{table}

\paragraph{Taylor hierarchy.}
\Cref{tab:hierarchy} lists the state and risk errors of the Taylor truncations at $U8$ summarized in \cref{sec:numerics-hierarchy}.

\begin{table}[htb]
\centering\small
\caption{Taylor truncations at $U8$: median state error $E_k$ and relative risk error $\delta_k$ for the homodyne-type and click receivers. Click risk errors use the full composite Taylor expansion along each path, which includes receiver curvature; on exchange the risk is insensitive to the features, so its risk errors carry no information.}
\label{tab:hierarchy}
\begin{tabular}{@{}lccc@{}}
\toprule
Dataset & $E_1\to E_2\to E_3\to E_4$ & homodyne-type $\delta_1\to\delta_4$ (\%) & click $\delta_1\to\delta_4$ (\%)\\ \midrule
PM2.5    & $25.7\to3.4\to0.97\to0.30\%$ & $0.82\to0.37\to0.08\to0.02$ & $0.55\to0.11\to0.13\to0.04$\\
solar    & $26.2\to4.1\to1.2\to0.35\%$  & $1.13\to0.20\to0.17\to0.02$ & $1.47\to0.11\to0.24\to0.02$\\
ETTh1    & $33.1\to4.1\to1.2\to0.34\%$  & $0.02\to0.16\to0.05\to0.00$ & $0.28\to0.06\to0.10\to0.00$\\
exchange & $28.2\to2.6\to0.78\to0.15\%$ & ${\approx}0$ & ${\approx}0$\\
traffic  & $28.3\to4.5\to1.1\to0.45\%$  & $0.34\to0.28\to0.14\to0.00$ & $1.19\to0.51\to0.53\to0.05$\\ \bottomrule
\end{tabular}
\end{table}

\paragraph{Interaction-order truncation.}
\Cref{tab:interaction} compares truncation in total Taylor order $T_k$ with truncation in interaction order $I_r$, which keeps every function of at most $r$ drive components exactly. Single-component exactness alone ($I_1$) is worse than second-order Taylor, and $I_2$ lies between $T_2$ and $T_3$.

\begin{table}[htb]
\centering\small
\caption{Median state error on 64 windows at $U8$: Taylor order $T_k$ versus interaction order $I_r$.}
\label{tab:interaction}
\begin{tabular}{@{}lcccccc@{}}
\toprule
 & $T_1$ & $I_1$ & $T_2$ & $I_2$ & $T_3$ & $T_4$\\ \midrule
PM2.5   & $26.6\%$ & $7.8\%$ & $3.3\%$ & $1.7\%$  & $0.87\%$ & $0.23\%$\\
traffic & $30.1\%$ & $8.0\%$ & $4.5\%$ & $1.24\%$ & $1.22\%$ & $0.42\%$\\ \bottomrule
\end{tabular}
\end{table}

\paragraph{Click approximants.}
\Cref{tab:click} lists the checks of \cref{prop:click} summarized in \cref{sec:numerics-click}. Computed from the eigenvalues of $Q_J$, the exact click features and noise reproduce the implementation of \cref{sec:numerics-setup} to $1.3\times10^{-14}$. The a-priori certificate inserts the bounds of \cref{app:proof-click} into \cref{app:proof-riskerror}; the a-posteriori certificate inserts the actual $\rho$, $\sigma_k$ and $\|\bar\Omega_b-\bar\Omega_d\|$ into the same inequality.

\begin{table}[htb]
\centering\small
\caption{Order-$d$ click approximants at $U8$ (maxima over seven datasets). The $p_0$ bound is $(1+\varepsilon_d)^8-1$ for mode sets of up to four modes. Certificates bound $|\Delta R|/R$; a dash means that the stability condition fails or the bound exceeds $1$.}
\label{tab:click}
\setlength{\tabcolsep}{4pt}
\begin{tabular}{@{}rccccc@{}}
\toprule
 & & & risk error & a-priori certificate & a-posteriori certificate\\
$d$ & $\varepsilon_d$ & $p_0$ error (bound) & $\lambda=0.01/1/100$ & $\lambda=100\,/\,1$ & $\lambda=100\,/\,1\,/\,0.01$\\ \midrule
$2$  & $7.0\times10^{-3}$ & $3.3\times10^{-2}$ ($5.7\times10^{-2}$) & $3.3/1.5/1.7\times10^{-3}$ & --- / --- & $4.9\times10^{-2}$ / --- / ---\\
$3$  & $1.3\times10^{-3}$ & $8.7\times10^{-3}$ ($1.1\times10^{-2}$) & $2.8/2.3/0.39\times10^{-3}$ & --- / --- & $1.9\times10^{-2}$ / --- / ---\\
$4$  & $2.5\times10^{-4}$ & $1.6\times10^{-3}$ ($2.0\times10^{-3}$) & $3.6/4.5/1.7\times10^{-4}$ & --- / --- & $2.5\times10^{-3}$ / --- / ---\\
$6$  & $9.7\times10^{-6}$ & $6.5\times10^{-5}$ ($7.8\times10^{-5}$) & $28/6.9/1.9\times10^{-6}$ & $6.2\times10^{-3}$ / --- & $8.3\times10^{-5}$ / $0.12$ / ---\\
$8$  & $3.9\times10^{-7}$ & $2.6\times10^{-6}$ ($3.1\times10^{-6}$) & $15/10/1.5\times10^{-7}$ & $2.3\times10^{-4}$ / $0.37$ & $5.0\times10^{-6}$ / $6.7\times10^{-3}$ / ---\\
$10$ & $1.6\times10^{-8}$ & $1.1\times10^{-7}$ ($1.3\times10^{-7}$) & $40/43/11\times10^{-9}$ & $9.4\times10^{-6}$ / $1.3\times10^{-2}$ & $2.5\times10^{-7}$ / $3.3\times10^{-4}$ / ---\\
$12$ & $6.6\times10^{-10}$ & $4.4\times10^{-9}$ ($5.3\times10^{-9}$) & $2.8/1.3/0.12\times10^{-9}$ & $3.9\times10^{-7}$ / $5.5\times10^{-4}$ & $9.4\times10^{-9}$ / $1.3\times10^{-5}$ / $0.12$\\ \bottomrule
\end{tabular}
\end{table}

\paragraph{Gate placement.}
To illustrate what the path rule says before fitting any task, we place the same inter-rail beam splitters at the beginning, in the middle, or at the end of a four-rail, two-loop step ($\tau=(3,7)$, two driven rails, equal TMSV squeezing). Every pair certified structurally zero remains numerically zero among the $1176$ checked drive pairs in each placement. The beginning placement has $38$ additional zero pairs that the graph does not certify: equal-squeezing source symmetry cancels their amplitudes. The placements are not equivalent as design choices: the beginning and end cases are fixed output-basis changes in this symmetric-source example, while middle placement creates paths through both modulated loop interactions (\cref{tab:rail-placement}). This is a single-circuit structural counterfactual using synthetic drives, not a forecast improvement or a general statement that boundary beam splitters are always redundant.

\begin{table}[htb]
\centering\small
\caption{Inter-rail beam-splitter placement in one four-rail counterfactual. Readable drives and nonzero interaction pairs are evaluated numerically; parenthesized counts are cross-channel pairs. The variation rank is the rank of sampled covariance changes, not a universal capacity.}
\label{tab:rail-placement}
\begin{tabular}{@{}lccc@{}}
\toprule
Placement & Drives & Pairs & Variation rank\\ \midrule
none & $16$ & $51\;(0)$ & $17$\\
beginning & $16$ & $51\;(0)$ & $17$\\
end & $16$ & $51\;(0)$ & $17$\\
middle & $25$ & $170\;(80)$ & $43$\\ \bottomrule
\end{tabular}
\end{table}

\section{Illustration: task-conditioned readout ages}\label{app:t3}
This development-stage comparison illustrates how task statistics and
circuit structure combine; it is not part of the theory and uses the
development validation segment. The B1 history profile of each task
measures how much its own-history baseline depends on lags $64$--$95$.
The paired test T3 then moved only selected readout ages into that lag
block (\cref{tab:t3}). Traffic and weather favor the aligned arm, while
exchange does not. These comparisons are not independent confirmation
or a universal age-selection rule. A topology-only prescription tested
in the same way, ``$\gcd(\tau)=1$ is better'', reversed on traffic
($+1.47\%$).

\begin{table}[h]
\centering\small
\caption{T3 paired age-alignment test (development validation segment): task history, changed readout ages and finite-shot forecasting error.}
\label{tab:t3}
\begin{tabular}{@{}lll@{}}
\toprule
Task / B1 & Ages $B\to A$ & MSE $B\to A$ ($\Delta$)\\ \midrule
traffic, $0.148$ & $32,40,48,56\to64,72,80,95$ & $0.5134\to0.5041$ $(-1.82\%)$\\
weather, $0.105$ & $24,40,56\to64,80,95$ & $0.4750\to0.4688$ $(-1.30\%)$\\
exchange, $0.00094$ & $32,40,56\to64,80,95$ & $0.1264\to0.1282$ $(+1.42\%)$\\ \bottomrule
\end{tabular}
\par\smallskip\raggedright\footnotesize
B1 is the increase in own-history baseline MSE when lags $64$--$95$ are removed, not GBS gain. Only changed ages are shown; the others remain fixed. $\Delta=(\mathrm{MSE}_A-\mathrm{MSE}_B)/\mathrm{MSE}_B$. Within each pair: $\tau=(1,12,36)$, eight ages, two rails, equal source energy, $S=2\times10^4$ threshold-click shots; MSEs average eight noise realizations. On exchange, the control head nearly suppresses the reservoir features and the aligned arm has larger shot variation, so its reversal is descriptive.
\end{table}

\section{Protocols and stopped studies}\label{app:protocol}
The risk comparisons in \cref{tab:risk} and the family comparisons in \cref{tab:family} were run under protocols registered before each run (A16, A20, A21), as was the second-order state test behind the $U8$ state errors (A15). The exactness checks (\cref{tab:exact}), the Taylor and interaction hierarchies (\cref{tab:hierarchy,tab:interaction}), the structural-support counts, the gate-placement counterfactual (\cref{tab:rail-placement}), the synthetic checks (\cref{tab:general}), the squeezing-encoding identity check (\cref{tab:squeeze}), the click-approximant check (\cref{tab:click}), the component and receiver checks (\cref{tab:components}) and the shot-law check (\cref{tab:shots}) are descriptive and were not pre-registered. Two studies were stopped by decision before completion and are reported for transparency.
(i) A selection study on the development validation segment was stopped after its second-order predictions had been locked and before any of its criteria were evaluated. Only its training-segment part enters \cref{tab:family}.
(ii) A four-dataset family study was stopped once \cref{thm:bridge} made exact risk available. Its single-point part is complete (\cref{tab:risk}, last four rows), and its family part is complete for solar only (\cref{tab:family}).
No computation in this paper reads the test segments (the final $20\%$ of each series). These segments are not unexposed data, however: earlier stages of this project evaluated forecasts on the same time ranges. We therefore make no held-out claim from them, and none of the results here requires one.

\section{Reproducibility}\label{app:repro}
\sloppy
The optical circuit replay uses the DeepQuantum Python framework~\cite{he2025deepquantum} and is checked against an independent NumPy transfer-matrix implementation.
Study-specific receiver, second-order propagation and structural-support routines are implemented in
\texttt{general\_tdm.py}, \texttt{bridge\_layers.py} and \texttt{second\_order\_state.py}. The A15, A16, A20 and A21 scripts write a pre-registration file before touching data; every script writes a JSON result file, and the remaining checks are descriptive (\cref{app:protocol}). \Cref{tab:repro} maps reported numbers and generated figures to their sources. Unqualified script names are under \texttt{scripts/}; unqualified results are under \texttt{experiments/}.
Code and result files are available from the corresponding author upon reasonable request.

{\footnotesize
\renewcommand{\arraystretch}{0.84}
\begin{longtable}{@{}>{\raggedright}p{6.0cm}>{\raggedright}p{4.8cm}>{\raggedright\arraybackslash}p{5.8cm}@{}}
\caption{Provenance of reported results.}\label{tab:repro}\\
\toprule
Result & Script & Evidence or output\\ \midrule
\endfirsthead
\multicolumn{3}{@{}l}{\emph{\cref{tab:repro}, continued}}\\
\toprule
Result & Script & Evidence or output\\ \midrule
\endhead
\bottomrule
\endlastfoot
Implementation checks, receiver layers (\cref{sec:numerics-setup}) & \texttt{a14\_layer\_split.py} & \texttt{a14/layer\_split.json}\\
State error $E_1,E_2$ at $U8$, Taylor regime & \texttt{a15\_second\_order.py} & \texttt{a15/second\_order.json}\\
Gate-level recursion vs.\ finite differences & \texttt{a17\_hvp\_support.py} & \texttt{a17/hvp\_support.json}\\
Exactness of \cref{thm:bridge} (\cref{tab:exact}) & \texttt{a22\_trig\_exact.py} & \texttt{a22/a22\_trig\_exact.json}\\
Taylor hierarchy (\cref{fig:truncation}, \cref{tab:hierarchy}); F1-29 & \texttt{a21b\_higher\_order.py} & \texttt{a21/a21b\_higher\_order.json}\\
Interaction truncation (\cref{tab:interaction}); non-admissible pairs & \texttt{a21c\_anova\_screen.py -{}-windows 64} & \texttt{a21/a21c\_w64.json}\\
Risk accuracy, weather/electricity/traffic (\cref{tab:risk}) & \texttt{a16\_second\_order\_risk.py} & \texttt{a16/second\_order\_risk.json}\\
Risk accuracy, solar/ETTh1/PM2.5/exchange (\cref{tab:risk}) & \texttt{a21\_theory\_check.py} (part 1) & \texttt{a21/a21\_backends.json}\\
Family ranking, weather/electricity/traffic (\cref{fig:family}, \cref{tab:family}) & \texttt{a20\_selection\_regret.py -{}-stage bridge} & \texttt{a20/a20\_bridge.json}\\
Family ranking, solar & \texttt{a21\_theory\_check.py} (part 2) & \texttt{a21/a21\_family.json}\\
Structural support, out of sample (3814 pairs) & \texttt{a17c\_dag\_rule\_oos.py} & \texttt{a17/dag\_rule\_oos.json}\\
Structural support at $U8$ (400 pairs) & \texttt{a17d\_u8\_support.py} & \texttt{a17/u8\_support.json}\\
Gate-placement counterfactual (\cref{tab:rail-placement}) & \texttt{a23\_rail\_place\_check.py} & \texttt{a23/a23\_rail\_place.json}\\
Squeezing encoding, source products, gate--source support (\cref{tab:general}) & \texttt{a24\_unified\_encoding\_check.py} & \texttt{a24/a24\_unified\_encoding.json}\\
Fock inputs and drive-noise damping (\cref{tab:general}) & \texttt{a25\_fock\_and\_jitter.py} & \texttt{a25/a25\_fock\_and\_jitter.json}\\
Photon-number-order checks (\cref{tab:general}) & \texttt{a26\_photon\_number\_order.py} & \texttt{a26/a26\_photon\_number\_order.json}\\
Squeezing-encoded risk from moments (\cref{tab:squeeze}) & \texttt{a27\_squeeze\_encoding\_risk.py} & \texttt{a27/a27\_squeeze\_risk.json}\\
Same risk from $\mathcal L$ and $\mathcal L_y$ only (\cref{tab:squeeze}) & \texttt{a27b\_mgf\_route.py} & \texttt{a27/a27b\_mgf\_route.json}\\
Click approximants (\cref{tab:click}) & \texttt{a28\_click\_polynomial.py} & \texttt{a28/a28\_click\_polynomial.json}\\
A-posteriori click certificates (\cref{tab:click}) & \texttt{a28b\_click\_certificate\_post.py} & \texttt{a28/a28b\_click\_\allowbreak certificate\_post\_d12.json}\\
Components and receivers, delay family (\cref{tab:components}) & \texttt{a29\_event\_span.py} & \texttt{a29/a29\_event\_span.json}\\
Event-span click support at $U8$ (\cref{sec:numerics-components}) & \texttt{a19c\_event\_span\_support.py} & \texttt{a19/a19c\_event\_span\_support.json}\\
Shot law (\cref{tab:shots}) & \texttt{a31\_shot\_law.py} & \texttt{a31/a31\_shot\_law.json}\\
Delay-family certificate (\cref{fig:analytic-witness}) & Analytic construction in \cref{app:analytic-certificate} & \cref{eq:analytic-floor,eq:analytic-gap}; figure drawn in TikZ in the manuscript source\\
Support sizes $41{,}616\to2{,}346\to238$ & \texttt{a19\_lifted\_closed\_form.py} & \texttt{a19/lifted\_closed\_form.json}\\
Click curvature support ($997$ pairs) & \texttt{a19b\_receiver\_\allowbreak curvature\_support.py} & \texttt{a19/receiver\_curvature\_support.json}\\
Paired tests T1--T3 (\cref{tab:t3}) & \texttt{a4\_t123.py} & \texttt{a4/t123.json}, \texttt{a4/PROTOCOL\_T123.md}\\
Shot and classical-cost walls (\cref{sec:boundary}) & \texttt{a5\_two\_walls.py} & \texttt{a5/two\_walls.json}\\
Stopped studies (\cref{app:protocol}) & --- & \texttt{a20/A20\_STOPPED.md}, \texttt{a21/A21\_STOPPED.md}\\
Circuit in \cref{fig:chain} & \texttt{paper/figs/\allowbreak make\_u8\_deepquantum.py} & \texttt{paper/figs/u8\_step\_dq.svg}\\
\Cref{fig:truncation,fig:family,fig:boundary} & \texttt{paper/figs/make\_figs.py}, \texttt{make\_boundary.py} & read the JSON files above\\
\end{longtable}}

\end{document}